\documentclass[11pt]{article}
\usepackage{fullpage}
\usepackage{graphicx}
\usepackage{upref}
\usepackage{enumerate}
\usepackage{latexsym}
\usepackage{pdfpages}
\usepackage[bookmarks]{hyperref}
\usepackage{color,graphics}
\usepackage{comment} 
\usepackage{caption}
\usepackage{subcaption}
\usepackage{wrapfig}
\usepackage{braket}

\usepackage{amssymb}
\usepackage{amsmath}
\usepackage{bbm}
\usepackage{amsthm}
\usepackage{mathrsfs}
\usepackage{mathtools}

\usepackage{epsfig,graphicx}
\usepackage{algorithmic}
\usepackage[ruled,linesnumbered]{algorithm2e}
\usepackage{amsfonts}
\usepackage{tikz} 
\usepackage{varioref}
\usepackage[ansinew]{inputenc}
\usepackage{qcircuit}
\usepackage{authblk}
\usepackage{multirow}
\usepackage{lineno}
\theoremstyle{plain}
\newtheorem{theorem}{Theorem}[section]
\theoremstyle{plain}
\newtheorem{lemma}[theorem]{Lemma}
\theoremstyle{plain}

\newtheorem{corollary}[theorem]{Corollary}
\theoremstyle{plain}

\theoremstyle{plain}
\newtheorem{conjecture}[theorem]{Conjecture}
\theoremstyle{plain}

\theoremstyle{definition}

\newtheorem{definition}[theorem]{Definition}
\theoremstyle{definition}
\newtheorem{fact}[theorem]{Fact}

\theoremstyle{remark}
\newtheorem{remark}{Remark}[section]

\theoremstyle{definition}

\AtBeginDocument{}

\DeclareMathOperator{\real}{\mathbb{R}}
\DeclareMathOperator{\nat}{\mathbb{N}}
\newcommand{\cmplx}{\mathbb{C}}
\newcommand{\intg}{\mathbb{Z}}

\newcommand{\poly}{\text{poly}}
\newcommand{\conj}[1]{\overline{#1}}

\newcommand{\tr}{\text{Tr}}

\newcommand{\id}{\mathbb{I}}
\newcommand{\cliff}{\mathcal{C}}
\newcommand{\pauli}{\mathcal{P}}
\newcommand{\clifft}{\mathcal{J}}

\newcommand{\X}{\text{X}}
\newcommand{\Y}{\text{Y}}
\newcommand{\Z}{\text{Z}}
\newcommand{\had}{\text{H}}

\newcommand{\CNOT}{\text{CNOT}}
\newcommand{\phase}{\text{S}}

\newcommand{\tof}{\text{TOF}}
\newcommand{\swap}{\text{SWAP}}
\newcommand{\cs}{\text{CS}}
\newcommand{\cz}{\text{CZ}}

\newcommand{\diag}{\text{diag}}
\newcommand{\chan}[1]{\widehat{#1}}

\newcommand{\sde}{\text{sde}}

\newcommand{\gen}{\mathcal{G}}

\newcommand{\q}[1]{(#1)}

\newcommand{\Path}{\text{Path}}
\newcommand{\ham}{\text{ham}}

\newcommand{\tofcount}{\mathcal{T}^{of}}
\newcommand{\tofeps}{\mathcal{T}_{\epsilon}^{of}}

\begin{document}

\title{Synthesizing Toffoli-optimal quantum circuits for arbitrary multi-qubit unitaries}
\author{Priyanka Mukhopadhyay \thanks{mukhopadhyay.priyanka@gmail.com, priyanka.mukhopadhyay@utoronto.ca}}

\affil[1]{Department of Computer Science, University of Toronto, ON, Canada}

\date{}

\maketitle

\begin{abstract}
    In this paper we study the Clifford+Toffoli universal fault-tolerant gate set. We introduce a generating set in order to represent any unitary implementable by this gate set and with this we derive a bound on the Toffoli-count of arbitrary multi-qubit unitaries. We analyse the channel representation of the generating set elements, with the help of which we infer $|\clifft_n^{Tof}|<|\clifft_n^T|$, where $\clifft_n^{Tof}$ and $\clifft_n^T$ are the set of unitaries exactly implementable by the Clifford+Toffoli and Clifford+T gate set, respectively. We develop Toffoli-count optimal synthesis algorithms for both approximately and exactly implementable multi-qubit unitaries. With the help of these we prove $|\clifft_n^{Tof}|=|\clifft_n^{CS}|$, where $\clifft_n^{CS}$ is the set of unitaries exactly implementable by the Clifford+CS gate set. 
\end{abstract}

%\tableofcontents

\section{Introduction}
\label{sec:intro}

Quantum circuits is a popular model for describing and implementing quantum algorithms, similar to the case of classical computation. These consist of a series of elementary operations or gates belonging to a universal set, which usually consists of Clifford group gates and at least one non-Clifford gate \cite{1998_G, 2010_NC}. Example of some universal gate sets include Clifford+T, V-basis, Clifford+CS, Toffoli+H. Often, choice of implementing gate set is determined by the underlying hardware technology. An important component of any quantum computer compilation process is quantum circuit synthesis and optimization. Depending upon the application the goal is often to optimize the number of a specific gate. For example, optimization of multi-qubit gates like CNOT is more relevant for NISQ era, when their implementations are more error-prone than the single qubit gates. Optimization of non-Clifford gates are more relevant for the fault-tolerant regime because most error-correction schemes implement Clifford gates transversally, allowing the logical operations to be performed precisely and with time proportional to the physical gate time. The non-Clifford gates, however, require large ancilla factories and additional operations like gate teleportation and state distillation \cite{2005_BK, 1999_GC}. These are less accurate procedures which require both additional time and space compared to a single physical gate. Thus the optimization of multi-qubit non-Clifford gate like Toffoli is important for both the NISQ as well as fault-tolerant regime. It is also worth noting that the minimum number of non-Clifford gates required to implement certain unitaries is a quantifier of difficulty in many algorithms that try to classically simulate qunatum computation \cite{2016_BG, 2016_BSS}.

The Solovay-Kitaev algorithm \cite{1997_K, 2006_DN} guarantees that given an $n$-qubit unitary $W$, we can generate a circuit with a "discrete finite" universal gate set, like Clifford+T, V-basis, Clifford+CS, Clifford+Toffoli, such that the unitary $U$ implemented by the circuit is at most a certain distance from $W$. In fact, in quantum computation a set of gates is said to be universal if any quantum operation can be approximated to arbitrary accuracy by a quantum circuit involving only these gates \cite{2010_NC}. A unitary is called exactly implementable by a gate set if there exists a quantum circuit with these gates, that implements it (up to some global phase). Otherwise, it is approximately implementable.  Accordingly, a synthesis algorithm can be (a) exact when $U=e^{i\phi}W$ ($\phi$ is the global phase); or (b) approximate when $d(U,W)\leq\epsilon$ for some $\epsilon>0$. $d(.)$ is a distance metric and for quantum circuit synthesis two popular metrics are the operator norm \cite{2015_R, 2016_RS} and global phase invariant distance \cite{2013_BGS, 2015_BBG, 2015_KMM, 2021_M, 2022_GMM2}. In this paper we use the global phase invariant distance because it ignores the global phase and hence avoids unnecessarily long approximating sequences that achieves a specific global phase. This distance is composable and inequalities relating it to operator norm can be found in \cite{2021_M}. It has also been used to synthesize unitaries in other models like topological quantum computation \cite{2014_KBS, 2021_JS}.

In this paper we consider the Clifford+Toffoli gate set, where the Toffoli gate (TOF) is defined as follows.
\begin{eqnarray}
    \tof=\begin{bmatrix}
        1 & 0 & 0 & 0 & 0 & 0 & 0 & 0 \\
        0 & 1 & 0 & 0 & 0 & 0 & 0 & 0 \\
        0 & 0 & 1 & 0 & 0 & 0 & 0 & 0 \\
        0 & 0 & 0 & 1 & 0 & 0 & 0 & 0 \\
        0 & 0 & 0 & 0 & 1 & 0 & 0 & 0 \\
        0 & 0 & 0 & 0 & 0 & 1 & 0 & 0 \\
        0 & 0 & 0 & 0 & 0 & 0 & 0 & 1 \\
        0 & 0 & 0 & 0 & 0 & 0 & 1 & 0 \\
    \end{bmatrix}
    \label{eqn:tof}
\end{eqnarray}
This gate set is universal for quantum computing, in fact universality has been shown for a more restricted gate set i.e. the Toffoli+H \cite{2003_A,2003_S,2020_AGR}. The Toffoli+H quantum circuits have been studied in the context of diagrammatic calculi \cite{2018_V}, path-sums \cite{2023_V} and quantum logic \cite{2013_DLSG}. Fault-tolerant implementations of the Toffoli gate has been shown in \cite{1997_G, 1998_CPMetal, 2012_FSBetal, 2012_RDNetal, 2013_J, 2013_J2, 2013_PR, 2017_Y, 2018_HH, 2019_BBCR, 2020_BCHK}. Toffoli gates occur frequently while synthesizing circuits for Hamiltonian simulation \cite{2023_MSW, 2023_MWZ} and many other applications where multi-controlled unitaries are used \cite{2011_MWS, 2012_BV, 2023_VAAetal, 2023_NM3}. Algorithms have been developed to synthesize Toffoli+H circuits for exactly implementable unitaries, but with no constraint on optimality of the non-Clifford Toffoli gate \cite{2023_AGLR}. Also, some of these algorithms only work for 3 qubit unitaries and none work for approximately implementable unitaries.    

%--------------------------------------------------------------------------
\subsection{Our contributions}

In the following points we discuss briefly the main results in our paper and with these we also try to shed light on the motivations for studying the mentioned problems related to this gate set. 

We already mentioned the importance of optimizing the Toffoli gate both for the NISQ as well as fault-tolerant era. So one aim is to have some representation or formalism to express the unitaries implementable by this gate set and also to characterize them. With this we hope to compare the different universal gate sets, specifically the size of the set of unitaries, exactly implementable by these gate sets. Such questions are very relevant for resource estimations of various algorithms. Due to various reasons like the fidelity of gates, ease of implementation, etc, different hardware platforms may prefer different universal gate sets. The ability to compare the relative number of different non-Clifford gates required to implement a certain unitary with different universal gate sets, will help in making wiser decisions about which gate set to use \cite{2013_LCJ, 2021_SBLetal}. Also, in literature some popular unitaries have implementations only in a specific gate set. We can then have a more convenient estimate on the non-Clifford gate count simply by replacing each such gate with the new gate set. Another important goal is to use these formalisms in order to design Toffoli-optimal synthesis algorithms for arbitrary multi-qubit unitaries. We develop both provable and heuristic algorithms. The former type of algorithms have a rigorous proof of optimality and the claimed complexity. For the latter type both the claimed optimality and complexity depends on some unproven conjectures. The need of developing heuristic algorithms usually arise due to the lack of efficient provable algorithms. Then, depending on some observations some conjectures are coined, which if/when proven, will prove the stated claims.     

\begin{table}[h]
    \centering
    \footnotesize
    \begin{tabular}{|c|p{1.6cm}|p{2.1cm}|c|c|}
    \hline
        \textbf{Algorithm} & \textbf{Type of algorithm} & \textbf{Type of unitary} & \textbf{Time complexity} & \textbf{Space complexity}   \\
        \hline
        APPROX-TOF-OPT & Provable & Approximately implementable & $O\left(n^{2m_{\epsilon}}2^{(6n-12)m_{\epsilon}} \right)$ & $O\left(n^22^{6n-12}\right)$ \\
        \hline
        Nested MITM & Provable & Exactly implementable & $O\left(n^{2(c-1)\left\lceil\frac{m}{c}\right\rceil}2^{(6n-12)(c-1)\left\lceil\frac{m}{c}\right\rceil} \right)$ & $O\left(n^{2\left\lceil\frac{m}{c}\right\rceil}2^{(6n-12)\left\lceil\frac{m}{c}\right\rceil} \right)$  \\
        \hline
        EXACT-TOF-OPT & Heuristic & Exactly implementable & $\poly(n^22^{6n-12},m)$ & $\poly(n^22^{6n-12},m)$    \\
        \hline
    \end{tabular}
    \caption{Summary of space and time complexity of the different algorithms developed in this paper. Here $m$ and $m_{\epsilon}$ are the Toffoli-count and $\epsilon$-Toffoli-count of the input $n$-qubit unitary ($n > 0$), respectively. $c\geq 2$ is the level of nesting.}
    \label{tab:results}
\end{table}

\begin{enumerate}
    \item We define a generating set $\gen_{Tof}$ such that any unitary exactly implementable by the Clifford+Toffoli gate set can be written as product of unitaries from this gate set and a Clifford (Section \ref{subsec:genTOF}). We show that $|\gen_{Tof}|\in O\left(n^22^{6n-12} \right)$. 
    
    \item With the help of $\gen_{Tof}$ we also derive a lower bound of $\Omega\left(\log_4\left(\frac{1}{\epsilon}\right) \right)$ on the $\epsilon$-Toffoli count of arbitrary multi-qubit unitaries (Section \ref{subsec:bound}). For a unitary $W$, its $\epsilon$-Toffoli count is the minimum number of Toffoli gates required to implement any unitary that is within $\epsilon$ distance of $W$.

    Most of the previous works have proven bounds on the total gate count or 1-qubit non-Clifford gate count like T \cite{2015_KMM, 2016_RS} and V \cite{2002_HRC, 2013_BGS, 2015_BBG, 2015_R}. Most of these are upper bounds and work for 1-qubit gates, specifically z-rotation gates. The bound on T-count derived in \cite{2015_KMM} is empirical, meaning it is interpolated from experimental data and there are no rigorous analysis. In fact, it is not clear how to extend the number-theoretic arguments used in \cite{2013_BGS, 2015_BBG, 2015_R, 2016_RS} for multi-qubit gates. 
    To the best of our knowledge, this is the first paper that derives bound on the count of a multi-qubit non-Clifford gate, while implementing arbitrary multi-qubit unitary. An upper bound on the CS-count of 2-qubit unitaries, has been shown in \cite{2021_GRT}.

    For the special case of exactly implementable unitaries, an upper bound on the Toffoli-count of $n$-qubit unitaries has been derived in \cite{2023_AGLR}. Our lower bound for this special case is quadratically smaller (in $n$) than this upper bound. 

    \item  We then derive the channel representation of unitaries in this set (Section \ref{subsec:chanRep}). This representation was introduced in the context of Clifford+T gate set in \cite{2014_GKMR}, and then extensively studied in \cite{2021_MM, 2022_GMM}. Its main use was in the design of efficient algorithms for exactly implementable unitaries. In this paper we have found another interesting aspect. It can also be used to characterize unitaries that can be exactly implemented by these gate sets. For example, we are able to derive interesting conclusions, like the impossibility of the exact implementation of the T gate. While, we know that the Toffoli gate can be implemented exactly with 7 T gates \cite{2014_GKMR, 2021_MM} or 4 T gates \cite{2013_J} if we allow extra ancillae and classical resources like measurement. Such inferences are not possible with the results derived in earlier works \cite{2020_AGR, 2023_AGLR}, that did not work with this representation. 

    This implies that $\left|\clifft_n^{Tof}\right|<\left|\clifft_n^{T}\right|$, where $\clifft_n^{Tof}$ and $\clifft_n^T$ are the set of unitaries exactly implementable by the Clifford+Toffoli and Clifford+T gate set, respectively. Using similar arguments we conclude that the V-basis gates cannot be implemented exactly by the Clifford+Toffoli gate set. In fact, Toffoli also cannot be implemented exactly by the V-basis gate set.

    It is worth mentioning here that one of the advantages of using channel representation in order to derive such exact implementability results is Fact \ref{fact:chanRepCliff}, which ensures that a unitary is the channel representation of a Clifford operator if and only if it has exactly one $\pm 1$ in each row and column. This implies that to check implementability of a group of the form Clifford+non-Clifford we need to focus on the channel representation of the non-Clifford operators only. This is not possible if we work with the unitary matrices only, as done in \cite{2020_AGR, 2023_AGLR, 2013_GS}.
    
    \item The channel representation also inherits the advantages already established in \cite{2021_MM}, for designing algorithms.  First, channel representation of $\gen_{Tof}$ helps us develop Toffoli-count-optimal synthesis algorithms, using the nested meet-in-the-middle (MITM) algorithm, introduced in \cite{2021_MM}. This algorithm was developed as a recursive generalization of the MITM algorithm of \cite{2014_GKMR}. The space and time complexity of our algorithm (Section \ref{subsubsec:nestMITM})  is  $O\left(n^{2\lceil\frac{m}{c}\rceil}2^{(6n-12)\left\lceil\frac{m}{c}\right\rceil} \right)$ and $O\left(n^{2(c-1)\lceil\frac{m}{c}\rceil}2^{(6n-12)(c-1)\left\lceil\frac{m}{c}\right\rceil} \right)$, respectively. Here $m=\tofcount(U)$ is the Toffoli-count of the input unitary $U$ and $c\geq 2$ is the level of nesting used. 
    Second, we define the smallest denominator exponent ($\sde_2$) with respect to 2 and discuss how it can be used to design exponentially faster heuristic algorithms (Section \ref{subsubsec:heuristic}). The time and space complexity of our algorithm is $\poly(n^22^{6n-12},\tofcount(U))$. Third, we can develop clever data structures, such that no floating  point operations are required at any step of our algorithms for exactly implementable unitaries. Fourth, we design a fast multiplication algorithm to multiply a unitary with the channel representation of a generating set element, which basically boils down to some row additions involving part of the matrices. This is especially useful when we have to perform many such multiplications.
    
    To the best of our knowledge, these are the first Toffoli-count-optimal-synthesis algorithms for exactly implementable unitaries. With the help of these algorithms we prove that CS can be implemented optimally with 3 Toffoli gates. In a separate paper we show that Toffoli can be optimally implemented with 3 CS gates. This implies that $\left|\clifft_n^{Tof}\right|=\left|\clifft_n^{CS}\right|$, where $\clifft_n^{Tof}$ and $\clifft_n^{CS}$ are the set of unitaries exactly implementable by the Clifford+Toffoli and Clifford+CS gate set, respectively.

    \item We design a provable Toffoli-count-optimal synthesis algorithm for approximately implementable unitaries (Section \ref{subsec:algoApprox}), analogous to the one in \cite{2022_GMM2} for Clifford+T. 
    The time complexity of this algorithm is $O\left(n^{2\tofeps(W)}2^{(6n-12)\tofeps(W)} \right)$, where $\tofeps(W)$ is the $\epsilon$-Toffoli-count of the input unitary $W$. The space complexity is $O\left(n^22^{6n-12}\right)$. To the best of our knowledge, this is the first algorithm that synthesizes Toffoli-count-optimal circuits for approximately implementable unitaries. We have summarized the complexities of the various algorithms developed in this paper in Table \ref{tab:results}.
\end{enumerate}

%------------------------------------------------------------------------------------
\subsection{Relevant work}

Much work has been done to characterize the unitaries exactly implementable by the various universal gate sets \cite{2020_AGR} like Clifford+T \cite{2013_GS, 2014_GKMR, 2021_MM}, V-basis \cite{2002_HRC}, Clifford+CS \cite{2021_GRT} (2 qubits only), Toffoli+H \cite{2023_AGLR}. Extensive work has been done to synthesize arbitrary multi-qubit unitary \cite{1997_K, 2002_KSVV, 2006_DN, 2011_F, 2020_dBBVA, 2021_MIC}, without optimality constraint on any particular gate.
Some synthesis algorithms have been developed for Toffoli+H \cite{2023_AGLR} gate set, but these work only for exactly implementable unitaries on 3 qubits, without any optimality constraint. In \cite{2021_GRT} a CS-count-optimal synthesis algorithm has been developed for exactly implementable 2-qubit unitaries. More work has been done for the synthesis of V-count-optimal circuits \cite{2013_BGS, 2015_BBG, 2015_R}, but these work only for 1-qubit unitaries, specifically for the z-rotations. 
Extensive work has been done for T-count-optimal synthesis of exactly implementable multi-qubit unitaries \cite{2013_KMM, 2014_GKMR, 2021_MM} and approximately implementable 1-qubit unitaries \cite{2013_KMM2, 2015_KMM, 2016_RS}. The algorithm in \cite{2022_GMM2} is the only one that returns T-count-optimal circuits for arbitrary multi-qubit unitaries. The algebraic procedures introduced in this work circumvents the limitations of the number-theoretic techniques in other optimal-synthesis papers, like \cite{2013_BGS, 2015_BBG, 2015_R, 2015_KMM, 2016_RS}, which exploit the structure of 1-qubit unitaries and hence it is not known how to generalize these for multi-qubit unitaries. In fact, they often guarantee the optimal count for specific unitaries like z-rotations.
Work has also been done for T-depth-optimal synthesis \cite{2013_AMMR, 2022_GMM}. 

So far we only discussed about optimal-synthesis algorithms where the input is an $n$-qubit unitary and output is a circuit. The complexity of these algorithms cannot avoid the exponential dependence on $n$. If additional optimality constraints are imposed then complexity increases even further, and often it also exponentially depends on other parameters, for example the minimum non-Clifford gate count. Heuristics have been designed that reduces the dependence on these other parameters to polynomial \cite{2021_MM}. However, with exponential dependence on $n$, these algorithms become intractable on a personal computer. So significant amount of work has been done to develop re-synthesis algorithms where the input is usually an already synthesized circuit and the output is a circuit with reduce gate count. Due to relaxed constraint and more input information, the complexity of these algorithms are polynomial in input size. However, in literature this reported complexity deos not include the cost of generating the input circuit which itself is exponential in $n$ \cite{2014_AMM}. Usually, optimal-synthesis algorithms are used to synthesize smaller unitaries that are more frequent in quantum algorithms, for example rotations. One can do piece-wise optimal synthesis of a larger circuit and then estimate the non-Clifford gate cost.

%-------------------------------------------------------------------
%\subsection{Organization of the paper}

\section{Preliminaries}
\label{sec:prelim}

We denote the $n\times n$ identity matrix by $\id_n$ or $\id$ if dimension is clear from the context. We denote the $(i,j)^{th}$ entry of a matrix $M$ by $M[i,j]$. Here we note that we denote the commutator bracket by $[A,B]=AB-BA$, where $A,B$ are operators. Though we use the same kind of brackets, the meaning should be clear from the context. We denote the $i^{th}$ row of $M$ by $M[i,.]$. We denote the set of $n$-qubit unitaries by $\mathcal{U}_n$. The size of an $n$-qubit unitary is $N\times N$ where $N=2^n$. The qubits on which a gate acts is mentioned in the subscript with brackets. For example, $X_{\q{q}}$ implies an X gate acting on qubit $q$. For multi-qubit controlled gate the control qubit(s) is (are) mentioned first, followed by the target qubit(s), separated by semi-colon. For example, $\CNOT_{(i;j)}$ denotes CNOT gate controlled on qubit $i$ and target on qubit $j$ and $\tof_{(i,j;k)}$ denotes the Toffoli gate with controls on qubits $i,j$ and target on qubit $k$. For symmetric multi-qubit gates like CS, where the unitary does not change if we interchange the control and target qubit, we replace the semi-colon with a comma. For convenience, we skip the subscript, when it is clear from the context. We have given detail description about the n-qubit Pauli operators ($\pauli_n$) and the Clifford group ($\cliff_n$) in Appendix \ref{app:clifford}. The Clifford+Toffoli group is generated by $\cliff_n$ and the Toffoli gate (Equation \ref{eqn:tof}).

A unitary $U$ is \textbf{exactly implementable by the Clifford+Toffoli} gate set if there exists an implementing circuit (up to some global phase) consisting of these gates, else it is \textbf{approximately implementable by} this gate set. We denote the set of $n$-qubit exactly implementable (by Clifford+Toffoli) unitaries by $\clifft_n^{Tof}$. Specifically, we say $W$ is $\epsilon$-\textbf{approximately implementable by the Clifford+Toffoli} gate set if there exists an exactly implementable unitary $U$ such that $d(U,W)\leq\epsilon$. The Solovay-Kitaev algorithm \cite{1997_K, 2006_DN} guarantees that any unitary is $\epsilon$-approximately implementable, for arbitrary precision $\epsilon\geq 0$. In this paper we use the following distance measure $d(.,.)$, which has been used in previous works like \cite{2011_F, 2015_KMM, 2022_GMM2} (qubit based computing), \cite{2014_KBS, 2021_JS} (topological quantum computing).
\begin{definition}[\textbf{Global phase invariant distance}]
 Given two unitaries $U,W\in\mathcal{U}_n$, we define the global phase invariant distance between them as follows.
 \begin{eqnarray}
  d(U,W)=\sqrt{1-\frac{\left|\tr\left(U^{\dagger}W\right)\right|}{N}}   \nonumber
 \end{eqnarray}
\end{definition}
Like the operator norm, this distance is composable and inequalities relating these two distance metrics have been derived in \cite{2021_M}.

%------------------------------------------------------
\subsection{Toffoli-count of circuits and unitaries}
\label{subsec:TofCountDefn}

\subsubsection*{Toffoli-count of circuits}

The \emph{Toffoli-count of a circuit} is the number of Toffoli-gates in it. 

\subsubsection*{Toffoli-count of exactly implementable unitaries}

The \emph{Toffoli-count of an exactly implementable unitary} $U$, denoted by $\tofcount(U)$, is the minimum number of Toffoli-gates required to implement it (up to a global phase) with a Clifford+Toffoli circuit. 

\subsubsection*{$\epsilon$-Toffoli-count of approximately implementable unitaries}

Let $W\in\mathcal{U}_n$ be an approximately implementable unitary. The \emph{$\epsilon$-Toffoli-count} of $W$, denoted by $\tofeps(W)$, is the minimum number of Toffoli gates required to implement any unitary that is within distance $\epsilon$ of $W$. In other words $\tofeps(W) =\tofcount(U)$, where $U\in\clifft_n^{Tof}$, $d(U,W)\leq\epsilon$ and $\tofcount(U)\leq\tofcount(U')$ for any exactly implementable unitary $U'$ within distance $\epsilon$ of $W$. 
We call a Toffoli-count-optimal circuit for any such $U$ as the \emph{$\epsilon$-Toffoli-count-optimal} circuit for $W$. 

It is not hard to see that the above definition is very general and can be applied to any unitary $W\in\mathcal{U}_n$, exactly or approximately implementable. If a unitary is exactly implementable then $\epsilon=0$.

\section{Results and Discussions}
\label{sec:results}

\subsection{Generating set}
\label{subsec:genTOF}

Any unitary $U$, exactly implementable by Clifford+Toffoli can be written as alternative product of a Clifford and a Toffoli gate. We can alternatively express $U$ as product of unitaries that are Clifford conjugation of Toffoli gates. We remember that Toffoli is a self-inverse unitary and so we need not consider its inverse.
\begin{eqnarray}
    U&=&e^{i\phi}C_m(\tof)C_{m-1}(\tof)\ldots C_1(\tof)C_0;\qquad \left[C_i\in\cliff_n,\phi\in [0,2\pi) \right]  \nonumber \\
    &=&e^{i\phi}\left(C_m(\tof)C_m^{\dagger}\right)\left(C_mC_{m-1}(\tof)C_{m-1}^{\dagger}C_{m}^{\dagger}\right)\ldots\left(C_m\ldots C_1(\tof)C_1^{\dagger}\ldots C_m^{\dagger}\right)C_m\ldots C_1C_0   \nonumber  \\
    &=&e^{i\phi}\left(\prod_{j=m}^1C_j'(\tof)C_j'^{\dagger}\right)C_0'\qquad [C_j'\in\cliff_n; j=0,\ldots,m]   \label{eqn:cliffConj}
\end{eqnarray}
Let $\tof_{\q{i,j;k}}$ be a Toffoli gate with controls on qubits $i, j$ and target on qubit $k$. The following Pauli expansion can be obtained by observing that any $n$-qubit unitary $U$ can be expressed in the Pauli basis as $U=\sum_i\alpha_iP_i$, where $P_i\in\pauli_n$ and $\alpha_i=\frac{1}{2^n}\tr(UP_i)$.
\begin{eqnarray}
   \tof_{\q{i,j;k}}&=&\frac{3}{4}\id+\frac{1}{4}\left(\X_{\q{i}}+\Z_{\q{j}}+\Z_{\q{k}}-\X_{\q{i}}\Z_{\q{j}}-\Z_{\q{j}}\Z_{\q{k}}-\Z_{\q{k}}\X_{\q{i}} +\X_{\q{i}}\Z_{\q{j}}\Z_{\q{k}}\right)    
   \label{eqn:tofPauli}
\end{eqnarray}
If $C\in\cliff_n$ is an $n$-qubit Clifford then,
\begin{eqnarray}
    C(\tof_{\q{i,j;k}})C^{\dagger}&=&\frac{3}{4}\id+\frac{1}{4}\left(P_1+P_2+P_3-P_1P_2-P_2P_3-P_3P_1+P_1P_2P_3\right) := G_{P_1,P_2,P_3}
    \label{eqn:Gp1p2p3}
\end{eqnarray}
where $P_1=C\X_{\q{i}}C^{\dagger}$, $P_2=C\Z_{\q{j}}C^{\dagger}$ and $P_3=C\Z_{\q{k}}C^{\dagger}$. Since $[\Z_{\q{j}},\Z_{\q{k}}]=[\X_{\q{i}},\Z_{\q{j}}] =0$ when $i\neq j\neq k$, so using Fact \ref{fact:commute} we have $[P_1,P_2]=[P_2,P_3]=[P_3,P_1]=0$. Also, $P_1\neq P_2$, $P_2\neq P_3$, $P_3\neq P_1$ and $P_1,P_2,P_3\neq\id$. In Appendix \ref{app:genTof} (Lemma \ref{app:lem:GtofProperties}) we have given detailed proof about the following properties of the unitaries $G_{P_1,P_2,P_3}$.

\begin{lemma}
If $P_1,P_2,P_3\in\pauli_n\setminus\{\id\}$ pair-wise commutes, then
\begin{enumerate}
    \item $G_{P_1,P_2,P_3}$ remains same for any permutation of the Paulis;

    \item $G_{P_1,-P_1P_2,P_3}=G_{P_1,P_2,P_3}$;

    \item $G_{P_1,-P_2,P_3}=G_{P_1,P_2,P_3}C$, for some $C\in\cliff_n$;

    \item $G_{P_1,P_2,P_1P_2}=\id$;

    \item $G_{P_1,P_2,P_1P_2P_3}=G_{P_1,P_2,P_3}$.
\end{enumerate}
    \label{lem:GtofProperties}
\end{lemma}

From Equation \ref{eqn:cliffConj} and using the above lemma we can define the following set, which we call a \textbf{generating set} (modulo Clifford), because we can express any exactly implementable unitary (up to a global phase) as product of unitaries from this set and a Clifford. 
\begin{eqnarray}
    \gen_{Tof}&=&\left\{ G_{P_1,P_2,P_3} : P_1,P_2,P_3\in\pauli_n\setminus\{\id\};\quad P_1\neq P_2\neq P_3;\quad [P_1,P_2]=[P_2,P_3]=[P_3,P_1]=0;\right.   \nonumber \\
    && \left. P_3\neq P_1P_2;\quad (P_1,P_2,P_3)\equiv (P_1,\pm P_1P_2,P_3)\equiv (P_1,P_2,P_1P_2P_3)\equiv\pi_3(P_1,P_2,P_3);  \right. \nonumber \\
    &&\left. \pi_3 \text{ is a permutation of } 3\text{ elements}. \right\}
    \label{eqn:genTOF}
\end{eqnarray}
We use $(P_1,P_2,P_3)\equiv (P_1',P_2',P_3')$ to imply that only one pair is included in the set. In Algorithm \ref{alg:genTOF} (Appendix \ref{app:pseudocode}) we have given the pseudocode for constructing this set. Specifically, we have the following result.
\begin{theorem}
Any unitary $U$ that is exactly implementable by the Clifford+Toffoli gate set can be expressed as,
\begin{eqnarray}
    U=e^{i\phi} \left( \prod_{j=m}^1 G_{P_{1j},P_{2j},P_{3j} } \right)C_0;\qquad [C_0\in\cliff_n; \phi\in [0,2\pi) ]; \nonumber
\end{eqnarray}
where $G_{P_{1j},P_{2j},P_{3j}} \in\gen_{Tof}$, a set defined in Equation \ref{eqn:genTOF}.
    \label{thm:decompose}
\end{theorem}

\begin{remark}
Let us consider what happens if we change the position of the target in the Toffoli unitary described in Equation \ref{eqn:tofPauli}.
\begin{eqnarray}
    \tof_{i,k;j}&=&\frac{3}{4}\id+\frac{1}{4}\left(\X_{\q{j}}+\Z_{\q{i}}+\Z_{\q{k}}-\Z_{\q{i}}\X_{\q{j}}-\Z_{\q{k}}\X_{\q{j}}-\Z_{\q{i}}\Z_{\q{k}}+\X_{\q{j}}\Z_{\q{i}}\Z_{\q{k}}\right)    \nonumber \\
   \tof_{j,k;i}&=&\frac{3}{4}\id+\frac{1}{4}\left(\X_{\q{k}}+\Z_{\q{i}}+\Z_{\q{j}}-\Z_{\q{j}}\X_{\q{k}}-\Z_{\q{i}}\X_{\q{k}}-\Z_{\q{i}}\Z_{\q{j}}+\X_{\q{k}}\Z_{\q{i}}\Z_{\q{j}}\right) 
   \label{eqn:tofPauli2}
\end{eqnarray}
So, changing the controls and targets, simply permutes the position of the Pauli operators. We can use the following Clifford conjugation equalities in order to change the position of the operators.
\begin{eqnarray}
\swap (\id\otimes \X)\swap = \X\otimes\id;\qquad \swap (\id\otimes \Z)\swap = \Z\otimes\id;\qquad
 \had\Z\had=\X;\qquad \had\X\had=\Z     \label{eqn:cliffConj1}    
\end{eqnarray}
Thus in the generating set $\gen_{Tof}$ we need not cosider Clifford conjugations of the unitaries in Equation \ref{eqn:tofPauli2}, that are obtained by changing the position of the target of the unitary in Equation \ref{eqn:tofPauli}.  
\label{remark:targetChange}
\end{remark}

\paragraph{Cardinality of $\gen_{Tof}$ : }We prove the following bound on $|\gen_{Tof}|$. The proof has been given in Appendix \ref{app:genTof} (Theorem \ref{app:thm:genTOFsize}).
\begin{theorem}
\begin{eqnarray}
    |\gen_{Tof}|&\leq& \frac{1}{384}[64^n-52^n-17(16^n-4^n)]+\frac{1}{576}[48^n-2\cdot 24^n]-\frac{1}{36}[12^n-2\cdot 6^n]+\frac{1}{24}(13^n-1)    \nonumber \\
    &\in& O\left( n^264^{n-2} \right) \nonumber
\end{eqnarray}
    \label{thm:genTOFsize}
\end{theorem}

In Table \ref{tab:genTOF} we have given the cardinality of the generating set obtained in practice, for different number of qubits, using Algorithm \ref{alg:genTOF} (Appendix \ref{app:pseudocode}).

\begin{table}[h]
    \centering
    \begin{tabular}{|c|c|c|}
    \hline
       $n$  & 3 & 4  \\
       \hline
        $|\gen_{tof}|$ & 129 & 7024  \\
        \hline
        Time & 1.087 sec & 1 hr 44 secs     \\
        \hline
    \end{tabular}
    \caption{Cardinality of $\gen_{tof}$ for different number of qubits $n$.}
    \label{tab:genTOF}
\end{table}

\paragraph{Circuit construction for $G_{P_1,P_2,P_3}$ : } We can construct a circuit implementing $G_{P_1,P_2,P_3}$ by deriving the conjugating Clifford and determining the control and target qubits of the Toffoli gate. Given 3 commuting non-identity Paulis $P_1,P_2,P_3$, we use the algorithm in \cite{2020_dBT} in order to derive a conjugating Clifford $C'\in\cliff_n$ such that $C'P_1C'^{\dagger}=P_1'$, $C'P_2C'^{\dagger}=P_2'$ and $C'P_3C'^{\dagger}=P_3'$,, where $P_1'=\bigotimes_{j=1}^nQ_j$, $P_2'=\bigotimes_{j=1}^nR_j$, $P_3'=\bigotimes_{j=1}^nS_j$ and $Q_j,R_j,S_j\in\{\id,\Z\}$. That is, the output of this algorithm is a triple of $\Z$-operators i.e. $n$-qubit Paulis that are tensor product of either $\id$ or $\Z$. Then we use the following conjugation relations,
\begin{eqnarray}
  &&  \swap (\id\otimes\Z)\swap=\Z\otimes\id;\quad \CNOT_{\q{j;k}} (\id_{\q{j}}\otimes\Z_{\q{k}})\CNOT_{\q{j;k}}=\Z_{\q{j}}\otimes\Z_{\q{k}}    \nonumber \\
   && \CNOT_{\q{j;k}} (\Z_{\q{j}}\otimes\id_{\q{k}})\CNOT_{\q{j;k}}=\Z_{\q{j}}\otimes\id_{\q{k}}; \nonumber
\end{eqnarray}
in order to derive Clifford $C''\in\cliff_n$ such that $C''P_1'C''^{\dagger}=\Z_{\q{a}}$, $C''P_2'C''^{\dagger}=\Z_{\q{b}}$ and $C''P_3'C''^{\dagger}=\Z_{\q{c}}$, where $1\leq a,b,c\leq n$ and $a\neq b \neq c$.  Now we use $\had$ gate (Equation \ref{eqn:cliffConj1}) in order to apply $\X$ in any one of the qubits. Suppose we apply $\had_{\q{a}}$. If $C=C'C''\had_{\q{a}}$, then we get $C'C''\had_{\q{a}}\X_{\q{a}}\had_{\q{a}}C''^{\dagger}C'^{\dagger}=P_1$, $C'C''\had_{\q{a}}\Z_{\q{b}}\had_{\q{a}}C''^{\dagger}C'^{\dagger}=P_2$ and $C'C''\had_{\q{a}}\Z_{\q{c}}\had_{\q{a}}C''^{\dagger}C'^{\dagger}=P_3$. Thus a circuit for $G_{P_1,P_2,P_3}$ consists of $\tof_{\q{a,b;c}}$, conjugated by Clifford $C=C'C''\had_{\q{a}}$. From Remark \ref{remark:targetChange} we can infer that interchanging the position of the target and control qubits will change the conjugating Clifford.

%---------------------------------------------------------------
\subsection{Bound on Toffoli-count of unitaries }
\label{subsec:bound}

In this section we derive lower bound on the Toffoli-count of arbitrary $n$-qubit unitaries. We know that any unitary can be expanded in the Pauli basis. Let $W$ be an $n$-qubit unitary and
\begin{eqnarray}
    W = \sum_{i=1}^{4^n} \conj{\alpha_i}P_i \qquad [P_i\in\pauli_n]. \label{eqn:Wpauli}
\end{eqnarray}

\paragraph{Not exactly implementable :} First we consider the case when $W$ is strictly approximately implementable.
Let $U = \left(\prod_{j=1}^mG_{P_{1_j},P_{2_j},P_{3_j}} \right) C_0e^{i\phi}$ is a unitary such that $W = UE$, for some $C_0\in\cliff_n$, $\phi\in [0,2\pi)$ and unitary $E$. Also $d(U, W) \leq\epsilon$, implying $|\tr(E)|\geq N(1-\epsilon^2)$, where $N=2^n$ and $m=\tofeps(W)$. Let $\widetilde{U} = \prod_{j=1}^mG_{P_{1_j},P_{2_j},P_{3_j}}$. Then,
\begin{eqnarray}
    \left| \tr\left(\widetilde{U} W^{\dagger}  \right) \right| = \left| \tr(EC_0) \right| \label{eqn:URec}
\end{eqnarray}
and if $C_0=\sum_{P\in\pauli_n} r_PP$, $M=\left| \left\{P:r_P \neq 0  \right\} \right|$, then from Theorem 3.1 of \cite{2022_GMM2} we have the following.
   \begin{eqnarray}
       && \frac{1-\epsilon^2}{\sqrt{M}} - \sqrt{M(2\epsilon^2-\epsilon^4)} \leq | \tr(EC_0P'/N) | \leq \frac{1}{\sqrt{M}} + \sqrt{M(2\epsilon^2-\epsilon^4)} \quad [\text{if }\quad r_{P'}\neq 0] \label{eqn:nonZeroP} \\
       && 0 \leq | \tr(EC_0P'/N) | \leq \sqrt{M(2\epsilon^2-\epsilon^4)} \quad [\text{if }\quad r_{P'} = 0] \label{eqn:zeroP}
   \end{eqnarray}
Here $1\leq M\leq N^2$. We use this in order to lower bound $m=\tofeps(W)$ as a function of $\epsilon$. First we observe that we can expand $\widetilde{U}$ as in the following lemma, the proof of which has been given in Appendix \ref{app:sec:lowBound}.
\begin{lemma}
If $\widetilde{U}= \prod_{j=1}^mG_{P_{1_j},P_{2_j},P_{3_j}}$, where $G_{P_{1_j},P_{2_j,P_{3_j}}} = a\id+bQ_j $, $Q_j = P_{1_j}+P_{2_j}+P_{3_j}-P_{1_j}P_{2_j}-P_{2_j}P_{3_j}-P_{3_j}P_{1_j}+P_{1_j}P_{2_j}P_{3_j}$, $a = \frac{3}{4}$ and $b = \frac{1}{4}$, then
\begin{eqnarray}
    \widetilde{U} &=& a^m\id+a^{m-1}b\left(\sum_{j}Q_j\right)+a^{m-2}b^2\left(\sum_{j_1 < j_2} Q_{j_1}Q_{j_2} \right)+a^{m-3}b^3\left(\sum_{j_1< j_2 < j_3}Q_{j_1}Q_{j_2}Q_{j_3}  \right)+\cdots    \nonumber \\
    &&\cdots+ab^{m-1} \left(\sum_{j_1< j_2< \cdots < j_{m-1}} Q_{j_1}Q_{j_2}\cdots Q_{j_{m-1}}  \right) + b^m \prod_{j=1}^mQ_j.  \nonumber
\end{eqnarray}
\label{lem:widetildeU}
\end{lemma}
Now from Equations \ref{eqn:Wpauli} and \ref{eqn:URec} we have 
\begin{eqnarray}
 | \tr(EC_0) | = \left|\tr\left( \widetilde{U}W^{\dagger} \right) \right|= \left|\sum_{i=1}^{4^n}\alpha_i\tr\left(\widetilde{U}P_i\right) \right|,   \nonumber 
\end{eqnarray}
and thus using triangle inequality and Equation \ref{eqn:nonZeroP}, we have
\begin{eqnarray}
     \sum_{i=1}^{4^n} |\alpha_i| |\tr(\widetilde{U}P_i P'/N)| &\leq& \frac{1}{\sqrt{M}} +\sqrt{M(2\epsilon^2-\epsilon^4)} \leq \frac{1}{\sqrt{M}} +\sqrt{2M}\epsilon  \label{eqn:ECPub} \\
     \left| |\alpha_k| |\tr(\widetilde{U}P_kP'/N)| - \sum_{i\neq k}|\alpha_i| |\tr(\widetilde{U}P_i P'/N)| \right| &\geq& \frac{1-\epsilon^2}{\sqrt{M}} -\sqrt{M(2\epsilon^2-\epsilon^4)} \geq \frac{1-\epsilon^2}{\sqrt{M}} -\sqrt{2M}\epsilon \nonumber \\
     -\left| |\alpha_k| |\tr(\widetilde{U}P_kP'/N)| - \sum_{i\neq k}|\alpha_i| |\tr(\widetilde{U}P_i P'/N)| \right| &\leq&  \frac{-1+\epsilon^2}{\sqrt{M}} +\sqrt{2M}\epsilon. \label{eqn:ECPlb}
\end{eqnarray}
Without loss of generality let us assume that $|\alpha_k| |\tr(\widetilde{U}P_kP'/N)| - \sum_{i\neq k}|\alpha_i| |\tr(\widetilde{U}P_i P'/N)| > 0$. So, adding the above inequalities we get
\begin{eqnarray}
     \sum _{i\neq k} |\alpha_i| |\tr(\widetilde{U}P_i P' /N )| \leq \frac{\epsilon^2}{2\sqrt{M}} + \sqrt{2M}\epsilon. \nonumber
\end{eqnarray}
Now $\sum _{i\neq k} |\alpha_i| |\tr(\widetilde{U}P_i P' /N )| \geq |\alpha_{\ell}| |\tr(\widetilde{U}P_{\ell} P' /N )|$, for some $P_{\ell}$ such that the absolute value of coefficients on the right is non-zero. For $M\geq 1$ and for non-Clifford $\widetilde{U}$, $EC_0$ there exists at least one such $P_{\ell}$. So,
\begin{eqnarray}
    |\alpha_{\ell}| |\tr(\widetilde{U}P_{\ell} P' /N )| \leq \frac{\epsilon^2}{2\sqrt{M}} + \sqrt{2M}\epsilon. \label{eqn:boundUZP}
\end{eqnarray}
Now we observe that $\tr(\widetilde{U} P_{\ell}P'/N)$ is the coefficient of $P_{\ell}P'$ in the Pauli expansion of $\widetilde{U}$, as given in Lemma \ref{lem:widetildeU}. 
Consider group of r-product terms of the form $a^{m-r}b^r \left( \sum_{j_1<j_2<\cdots<j_r} Q_{j_1}Q_{j_2}\ldots Q_{j_r} \right)$. We remember that each $Q_j$ is a linear combination of 7 Paulis and $Q_j^2=7\id-6Q_j$ (Lemma \ref{lem:trQ_Q2} in Appendix \ref{app:chanRep}), which implies that $Q^m$ is a linear combination of 8 Paulis, for any $m\geq 2$. Thus each summand in a r-product group is a product of at most $r$ Paulis and no two summands have the same combination of Paulis. We remember that we are interested in the minimum $m$ for Toffoli-count. The expression of $\widetilde{U}$ is sum of $m+1$ such groups of $r$-products, where each combination is distinct from the other. Overall, there can be at least a constant number of terms that are $P_{\ell}P'$. Since $|b|<|a|<1$ so minimum value of the coefficient can be $c|b|^m$, where $c$ is a constant. Thus from Equation \ref{eqn:boundUZP},
\begin{eqnarray}
    |\alpha_{\ell}| c|b|^m &\leq& |\alpha_{\ell}| |\tr(\widetilde{U}P_{\ell} P' /N )| \leq \frac{\epsilon^2}{\sqrt{2M}} +\sqrt{2M}\epsilon = \epsilon \left( \frac{\epsilon}{2\sqrt{M}} +\sqrt{2M} \right). \nonumber    \\ 
    |b|^m &\leq& \frac{\epsilon}{c|\alpha_{\ell}|} \left( \frac{\epsilon}{2\sqrt{M}} +\sqrt{2M} \right)    \nonumber \\
    m \log\frac{1}{|b|} &\geq& \log\frac{1}{\epsilon} + \log\left(\frac{ c|\alpha_{\ell}| }{ \frac{\epsilon}{2\sqrt{M}} +\sqrt{2M} } \right) \nonumber \\
    m &\geq& \log_{\frac{1}{|b|}} \left(\frac{1}{\epsilon} \right)- \log_{\frac{1}{|b|}} \left(\frac{ \frac{\epsilon}{2\sqrt{M}} +\sqrt{2M} }{ c|\alpha_{\ell}| } \right). \nonumber
\end{eqnarray}
Since $|b|=\frac{1}{4}$, we have the following asymptotic lower bound, with respect to $\epsilon$.
\begin{theorem}
    Let $W = \sum_i\alpha_iP_i$ is an $n$-qubit unitary, approximately implementable by the Clifford+Toffoli gate set. Here $P_i\in\pauli_n$. Then for some constant $c$,
    \begin{eqnarray}
        \tofeps\left(W\right) &\geq& \log_{4} \left(\frac{1}{\epsilon} \right)- \log_4 \left(\frac{ \frac{\epsilon}{2\sqrt{M}} +\sqrt{2M} }{ c|\alpha_{max}| } \right) \in \Omega \left( \log_4 \left( \frac{1}{\epsilon} \right) \right).  \nonumber
    \end{eqnarray}
    \label{thm:tofBound}
\end{theorem}

\paragraph{Exactly implementable : } Now we consider the case when $W$ is exactly implementable i.e. $E=\id$ and $\epsilon = 0$. Then, from Fact \ref{fact:cliffCoeff} in Appendix \ref{app:prelim} (\cite{2010_BS}) or plugging in $\epsilon = 0$ in Equation \ref{eqn:nonZeroP} ,we obtain 
\begin{eqnarray}
    \sum_{i=1}^{4^n} |\alpha_i| |\tr(\widetilde{U}P_i P'/N)| &\leq& \frac{1}{\sqrt{M}}. \label{eqn:CPub}    \nonumber 
\end{eqnarray}
Now, $|\alpha_{\ell}| |\tr(\widetilde{U}P_{\ell} P'/N)| \leq \sum_{i=1}^{4^n} |\alpha_i| |\tr(\widetilde{U}P_i P'/N)| $ , where $|\alpha_{\ell}|\neq 0$. As previously argued $ c|b|^m\leq  |\tr(\widetilde{U}P_{\ell} P'/N)|$, for some constant $c$. So,
\begin{eqnarray}
    |\alpha_{\ell}| c|b|^m &\leq&\frac{1}{\sqrt{M}}  \nonumber \\
    \left(\frac{1}{|b|}\right)^m &\geq& |\alpha_{\ell}|c\sqrt{M}    \nonumber \\
    \implies m &\geq& \log_{\frac{1}{|b|}} \left(|\alpha_{max}|c\sqrt{M} \right) 
\end{eqnarray}
Thus we have the following result.
\begin{theorem}
    Let $W = \sum_i\alpha_iP_i$ is an $n$-qubit unitary, exactly implementable by the Clifford+Toffoli gate set. Here $P_i\in\pauli_n$. Then for some constant $c$ and $1\leq M\leq 4^n$,
    \begin{eqnarray}
        \tofcount\left(W\right) &\geq& \log_4 \left(|\alpha_{max}|c\sqrt{M} \right).  \nonumber
    \end{eqnarray}
    \label{thm:tofBoundExact}
\end{theorem}

\paragraph{Upper bound : } An upper bound on the Toffoli-count can be derived from the Solovay-Kitaev theorem \cite{1997_K, 2002_KSVV, 2006_DN}, which gives an upper bound of $O\left(\log^c\left(\frac{1}{\epsilon}\right)\right)$ on the total number of gates, where $3\leq c\leq 4$ is a constant. This trivially implies an upper bound on the non-Clifford gate count, as well. Our lower bound in Theorem \ref{thm:tofBound} is polynomially smaller than this upper bound.

For exactly implementable unitaries over the Toffoli-Hadamard gate set, an upper bound of $O(n^2\log(n)k)$ (k is a non-negative integer) on the total gate count was shown in \cite{2023_AGLR}. This implies an upper bound on the Toffoli-count, though for a restricted gate set (not Clifford+Toffoli). If we plug in $M=4^n$ then our lower bound in Theorem \ref{thm:tofBoundExact} is quadratically smaller than this upper bound.

\subsubsection{Illustration : Toffoli-count of rotation gates}
\label{subsubsec:boundRot}

As examples, we consider the following rotation gates. In Appendix \ref{app:subsec:pauliBasis} we have explicitly derived the Pauli basis expansions that we show below, for each unitary.
\begin{eqnarray}
    R_z(\theta) &=& e^{-i\frac{\theta}{2}} \begin{bmatrix}
    1 & 0 \\
    0 & e^{i\theta}
    \end{bmatrix} = e^{-i\frac{\theta}{2}} \left(\alpha_1\id+\alpha_2\Z  \right)  
    \nonumber \\
   \id\otimes\id\otimes R_z(\theta) &=& e^{-i\frac{\theta}{2}} \left(\alpha_1 \id\id\id+\alpha_2\id\id\Z  \right)
   \label{eqn:RZ}
\end{eqnarray}
Here $\alpha_1 = \frac{1+e^{i\theta}}{2} $, $\alpha_2 = \frac{1-e^{i\theta}}{2} $ and so
\begin{eqnarray}
    |\alpha_1| &=& \left| \frac{1}{2}\left(1+\cos\theta+i\sin\theta  \right)  \right| = \frac{1}{2}\sqrt{((1+\cos\theta)^2+\sin^2\theta)} \nonumber \\
    &=&\frac{1}{2}\sqrt{2(1+\cos\theta)  }=\frac{1}{2}\sqrt{4\cos^2\frac{\theta}{2}} = \cos\frac{\theta}{2}   \nonumber \\
    |\alpha_2| &=& \left| \frac{1}{2}\left(1-\cos\theta+i\sin\theta  \right)  \right| = \frac{1}{2}\sqrt{((1-\cos\theta)^2+\sin^2\theta)} \nonumber \\
    &=&\frac{1}{2}\sqrt{2(1-\cos\theta)  }=\frac{1}{2}\sqrt{4\sin^2\frac{\theta}{2}} = \sin\frac{\theta}{2}.
    \label{eqn:alphaBeta}
\end{eqnarray}
Next, we consider a controlled $R_z(\theta)$ unitary, as defined below.
\begin{eqnarray}
    cR_z(\theta) &=& \begin{bmatrix}
    1 & 0 & 0 & 0 \\
    0 & 1 & 0 & 0 \\
     0 & 0 & e^{-i\theta/2} & 0 \\
    0 & 0 & 0 & e^{i\theta/2}
    \end{bmatrix} = \alpha_1'\id\id+\alpha_2'\Z\id+\alpha_3'\id\Z+\alpha_4'\Z\Z \nonumber \\
   \id\otimes cR_z(\theta) &=& \alpha_1'\id\id\id+\alpha_2'\id\Z\id+\alpha_3'\id\id\Z+\alpha_4'\id\Z\Z
    \label{eqn:cRZ}
\end{eqnarray}
Here $\alpha_1' = \frac{1+\cos\frac{\theta}{2}}{2}$, $\alpha_2'=\frac{1-\cos\frac{\theta}{2}}{2}$, $\alpha_3' = \frac{i}{2}\sin\frac{\theta}{2}$, $\alpha_4' = -\alpha_3'$. So,
\begin{eqnarray}
    |\alpha_1'|&=&\frac{1+\cos\frac{\theta}{2}}{2};\qquad |\alpha_2'|=\frac{1-\cos\frac{\theta}{2}}{2}  \nonumber \\
     |\alpha_3'| = |\alpha_4'| &=& \frac{\sin\frac{\theta}{2}}{2} .    \label{eqn:cRzAlphaBeta}
\end{eqnarray}

Next, we consider the following controlled rotation gate, that is widely used in Quantum Fourier Transform \cite{2010_NC}.
\begin{eqnarray}
    cR_n(\theta) &=& \begin{bmatrix}
    1 & 0 & 0 & 0 \\
    0 & 1 & 0 & 0 \\
    0 & 0 & 1 & 0 \\
    0 & 0 & 0 & e^{i\theta}
    \end{bmatrix} = \alpha_1''\id\id+\alpha_2''\id\Z+\alpha_3''\Z\id+\alpha_4''\Z\Z \nonumber \\
 \id\otimes cR_n(\theta) &=&  \alpha_1''\id\id\id+\alpha_2''\id\id\Z+\alpha_3''\id\Z\id+\alpha_4''\id\Z\Z
   \label{eqn:cRn}
\end{eqnarray}
Here $\alpha_1'' = \frac{3+e^{i\theta}}{4}$, $\alpha_2''=\alpha_3''=\frac{1-e^{i\theta}}{4}$, $\alpha_4'' = -\alpha_2''$. So,
\begin{eqnarray}
    |\alpha_1''|&=&\frac{1}{4}\sqrt{ (3+\cos\theta)^2 +\sin^2\theta } = \frac{1}{4}\sqrt{10+6\cos\theta}=\frac{1}{2}\sqrt{1+3\cos^2\frac{\theta}{2}}  \nonumber \\
    |\alpha_2''| = |\alpha_3''| = |\alpha_4''| &=& \frac{1}{4}\sqrt{(1-\cos\theta)^2+\sin^2\theta  } = \frac{1}{4}\sqrt{2-2\cos\theta}=\frac{1}{2}\sin\frac{\theta}{2}.    \label{eqn:cRnAlphaBeta}
\end{eqnarray}

Next, we consider the 2-qubit Given's rotation unitary, that is used in quantum chemistry \cite{2004_VMS, 2022_AdMQetal}, quantum simulations \cite{2018_KMWetal, 2020_google}, quantum machine learning \cite{2022_KP}, variational quantum algorithms \cite{2021_DAJetal}.
\begin{eqnarray}
    Givens(\theta) &=& \begin{bmatrix}
    1 & 0 & 0 & 0 \\
    0 & \cos(\theta) & -\sin(\theta) & 0 \\
    0 & \sin(\theta) & \cos(\theta) & 0 \\
    0 & 0 & 0 & 1
    \end{bmatrix} = \alpha_1'''\id\id+\alpha_2'''\X\Y+\alpha_3'''\Y\X+\alpha_4'''\Z\Z \nonumber \\
 \id\otimes Givens(\theta) &=&  \alpha_1'''\id\id\id+\alpha_2'''\id\X\Y+\alpha_3'''\id\Y\X+\alpha_4'''\id\Z\Z
   \label{eqn:givens}
\end{eqnarray}
Here $\alpha_1''' = \frac{1+\cos\theta}{2}$, $\alpha_2'''=i\frac{\sin\theta}{2}$, $\alpha_3'''=-i\frac{\sin\theta}{2}$ $\alpha_4''' = \frac{1-\cos\theta}{2}$. So,
\begin{eqnarray}
    |\alpha_1'''|&=&\frac{1+\cos\theta}{2};\qquad |\alpha_4'''| = \frac{1-\cos\theta}{2}  \nonumber \\
    |\alpha_2'''| = |\alpha_3'''| &=& \frac{\sin\theta}{2} .    \label{eqn:givensAlphaBeta}
\end{eqnarray}

We use $\diag(d_1,d_2,\ldots,d_N)$ to denote a $N\times N$ diagonal matrix with $d_1, d_2,\ldots,d_N$ as diagonal entries, and 0 in the remaining places. Now we consider the following rotation gates, controlled on 2 qubits.
\begin{eqnarray}
    ccR_n(\theta) &=&\diag\left(1,1,1,1,1,1,1,e^{i\theta}  \right)  \nonumber \\
    &=&\widehat{\alpha_1}\id\id\id + \widehat{\alpha_2}\Z\id\id+\widehat{\alpha_3}\id\Z\id+\widehat{\alpha_4}\id\id\Z+\widehat{\alpha_5}\Z\Z\id+\widehat{\alpha_6}\id\Z\Z+\widehat{\alpha_7}\Z\id\Z+\widehat{\alpha_8}\Z\Z\Z
    \label{eqn:ccRn}
\end{eqnarray}
Here $\widehat{\alpha_1}=\left(\frac{7+e^{i\theta}}{8}\right)$, $\widehat{\alpha_2}=\widehat{\alpha_3}=\widehat{\alpha_4}=\widehat{\alpha_8}=\left(\frac{1-e^{i\theta}}{8}\right)$, $\widehat{\alpha_5}=\widehat{\alpha_6}=\widehat{\alpha_7}=-\widehat{\alpha_8}$ and so,
\begin{eqnarray}
    |\widehat{\alpha_j}| &=& \frac{1}{8}\sqrt{(7+\cos\theta)^2+\sin^2\theta }=\frac{1}{4}\sqrt{9+7\cos^2\frac{\theta}{2}}   \qquad [j=1]    \nonumber \\
    &=& \frac{1}{8}\sqrt{(1-\cos\theta)^2+\sin^2\theta}=\frac{1}{4}\sin\frac{\theta}{2} \qquad [j\neq 1]. \label{ccRnAlphaBeta}   
\end{eqnarray}

Finally, we consider the following rotation gate.
\begin{eqnarray}
    ccR_z(\theta) &=&\diag\left(1,1,1,1,1,1,e^{-i\theta/2},e^{i\theta/2}  \right)   \nonumber \\
    &=& \widetilde{\alpha_1}\id\id\id + \widetilde{\alpha_2}\Z\id\id+\widetilde{\alpha_3}\id\Z\id+\widetilde{\alpha_4}\id\id\Z+\widetilde{\alpha_5}\Z\Z\id+\widetilde{\alpha_6}\id\Z\Z+\widetilde{\alpha_7}\Z\id\Z+\widetilde{\alpha_8}\Z\Z\Z
    \label{eqn:ccRz}
\end{eqnarray}
Here $\widetilde{\alpha_1}=\frac{3+\cos\frac{\theta}{2} }{4}$, $\widetilde{\alpha_2} = \widetilde{\alpha_3} = \frac{1-\cos\frac{\theta}{2}}{4}$, $\widetilde{\alpha_5}=-\widetilde{\alpha_2}$, $\widetilde{\alpha_6} = \widetilde{\alpha_7}= i\frac{\sin\frac{\theta}{2}}{4}$ and $\widetilde{\alpha_4}=\widetilde{\alpha_8}=-\widetilde{\alpha_6}$, and so,
\begin{eqnarray}
    |\widetilde{\alpha_j}| &=&\frac{3+\cos\frac{\theta}{2} }{4}\qquad [j=1]   \nonumber \\
    &=&\frac{1-\cos\frac{\theta}{2}}{4}\qquad [j=2,3,5] \nonumber \\
    &=&\frac{\sin\frac{\theta}{2}}{4}\qquad [j = 4,6,7,8]. \label{eqn:ccRzAlphaBeta}   \nonumber
\end{eqnarray}
Since each of these are approximately implementable unitaries, so we can obtain a lower bound on their $\epsilon$-Toffoli-count by plugging in the respective values in Theorem \ref{thm:tofBound}.

%------------------------------------------------
\subsection{Channel representation}
\label{subsec:chanRep}

An $n$-qubit unitary $U$ can be completely determined by considering its action on a Pauli $P_s\in\pauli_n : P_s\rightarrow UP_sU^{\dagger}$. The set of all such operators (with $P_s\in\pauli_n$) completely determines $U$ up to a global phase. Since $\pauli_n$ is a basis for the space of all Hermitian $2^n\times 2^n$ matrices we can write
\begin{eqnarray}
    UP_sU^{\dagger}&=&\sum_{P_r\in\pauli_n}\chan{U}_{rs}P_r,\qquad\text{where }\quad \chan{U}_{rs}=\frac{1}{2^n}\tr\left(P_rUP_sU^{\dagger}\right).
    \label{eqn:chanRepDefn}
\end{eqnarray}
This defines a $4^n\times 4^n$ matrix $\chan{U}$ with rows and columns indexed by Paulis $P_r,P_s\in\pauli_n$. We refer to $\chan{U}$ as the \textbf{channel representation} of $U$. By Hermitian conjugation each entry of $\chan{U}$ is real. The channel representation respects matrix multiplication and tensor product i.e. $\chan{UW}=\chan{U}\chan{W}$ and $\left(\chan{U\otimes W}\right)=\chan{U}\otimes\chan{W}$. Setting $V=U^{\dagger}$ it follows that $\chan{U^{\dagger}}=\left(\chan{U}\right)^{\dagger}$, and we see that the channel representation $\chan{U}$ is unitary.
\begin{fact}
The channel representation inherits the decomposition from Theorem \ref{thm:decompose} (and in this decomposition there is no global phase factor).
\begin{eqnarray}
    \chan{U}=\left(\prod_{j=m}^1\chan{G_{P_{1j},P_{2j},P_{3j}}}\right)\chan{C_0}\qquad [G_{P_{1j},P_{2j},P_{3j}}\in\gen_{Tof}]    \nonumber
\end{eqnarray}
    \label{fact:decompose}
\end{fact}
A decomposition in which $m=\tofcount(U)$ is called a \textbf{Toffoli-count-optimal decomposition}. The channel representation of the unitaries in the generating set $\gen_{Tof}$ has some important properties, which we summarize in the following theorem. These facilitate the design of proper data structure and efficient algorithms for matrix operations like multiplication, inverse, involving only integer arithmetic. These properties also give the intuition to develop clever heuristics for Toffoli-count-optimal synthesis algorithms. Further, we can draw inferences about the exact implementability of certain unitaries. 

\begin{theorem}
Let $G_{P_1,P_2,P_3}\in\gen_{Tof}$, where $P_1,P_2,P_3\in\pauli_n\setminus\{\id\}$. Then its channel representation $\chan{G_{P_1,P_2,P_3}}$ has the following properties.
\begin{enumerate}
    \item The diagonal entries are $1$ or $\frac{1}{2}$.

    \item If a diagonal entry is 1 then all other entries in the corresponding row and column is 0.

    \item If a diagonal entry is $\frac{1}{2}$ then there exists three entries in the corresponding row and column that are equal to $\pm\frac{1}{2}$, rest is 0.

    \item Exactly $2^{2n-3}$ i.e. $\frac{1}{8}^{th}$ of the diagonal elements are $1$, while the remaining are $\frac{1}{2}$.
\end{enumerate}
    \label{thm:chanRep}
\end{theorem}
The proof has been given in Appendix \ref{app:chanRep}. Since Cliffords map Paulis to Paulis, up to a possible phase factor of -1, so we have the following.
\begin{fact}[\cite{2014_GKMR}]
    Let $\chan{\cliff_n}=\{\chan{C}:C\in\cliff_n\}$. A unitary $Q\in\chan{\cliff_n}$ if and only if it has one non-zero entry in each row and column, equal to $\pm 1$.
    \label{fact:chanRepCliff}
\end{fact}
The above fact, along with Theorem \ref{thm:chanRep} and Fact \ref{fact:decompose}, implies the following.
\begin{lemma}
If $U$ is a unitary exactly implementable by the Clifford+Toffoli gate set, then each entry of $\chan{U}$ belongs to the ring $\intg\left[\frac{1}{2}\right]=\{\frac{a}{2^k} : a\in 2\intg+1, k\in\nat\}$.
\label{lem:chanRepRing}
\end{lemma}
Since the channel representation of the non-Clifford T gate do not belong to this ring \cite{2021_MM}, so  it cannot be exactly implemented by Clifford+Toffoli. On the other hand, we know that Toffoli can be implemented exactly with at most 7 T gates \cite{2021_MM} or with 4 T gates \cite{2013_J} if extra ancilla and classical resources like measurement is used. Thus every unitary exactly implementable by the Clifford+Toffoli gate set can be exactly implemented by the Clifford+T gate set. But there exists unitaries exactly implementable by the Clifford+T gate set that cannot be exactly implemented by the Clifford+Toffoli gate set. So the following theorem follows.

\begin{theorem}
Let $\clifft_n^T$ and $\clifft_n^{Tof}$ be the set of unitaries exactly implementable by the Clifford+T and Clifford+Toffoli gate set, respectively. Then
$
    |\clifft_n^{Tof}|< |\clifft_n^T|. 
$
    \label{thm:notImplement}
\end{theorem}
In a separate paper we prove that the channel representation of unitaries exactly implementable by the V-basis gate set, has entries in the ring $\intg\left[\frac{1}{\sqrt{5}}\right]$. Thus these gates are not exactly implementable by the Clifford+Toffoli gate set. In fact, Toffoli also cannot be exactly implemented by the V-basis gate set. Using Theorem \ref{thm:tofBound} we can say that we require $\Omega\left(\log_4\left(\frac{1}{\epsilon}\right)\right)$ Toffoli gates in order to implement an $\epsilon$-approximation of T and V-basis gates.
In Section \ref{subsec:implementation} we show that the CS gate unitary, which is defined as,
\begin{eqnarray}
    \cs=\begin{bmatrix}
        1 & 0 & 0 & 0 \\
        0 & 1 & 0 & 0 \\
        0 & 0 & 1 & 0 \\
        0 & 0 & 0 & i
    \end{bmatrix},
    \label{eqn:cs}
\end{eqnarray}
can be implemented optimally with 3 Toffoli gates.  In a separate paper we show that Toffoli can be implemented with 3 CS gates. Thus we have the following theorem.
\begin{theorem}
Let $\clifft_n^{CS}$ and $\clifft_n^{Tof}$ be the set of unitaries exactly implementable by the Clifford+CS and Clifford+Toffoli gate set, respectively. Then
$
    |\clifft_n^{Tof}| = |\clifft_n^{CS}|. \nonumber
$
    \label{thm:implementCS}
\end{theorem}
 Here we define the following.
\begin{definition}
For any non-zero $v\in\intg\left[\frac{1}{2}\right]$ the \textbf{smallest 2-denominator exponent}, denoted by $\sde_2$, is the smallest $k\in\nat$ for which $v=\frac{a}{2^k}$, where $a\in 2\intg+1$.
    \label{defn:sde}
\end{definition}
We define $\sde_2(0)=0$. For a $d_1\times d_2$ matrix $M$ with entries over this ring we define
\begin{eqnarray}
    \sde_2(M)&=&\max_{a\in [d_1],b\in [d_2]}\sde_2(M_{ab}).  \nonumber
\end{eqnarray}

\begin{comment}
----------TO REMOVE IF NOT NEEDED--------------------------------------------
\begin{fact}
    Let $v_1,v_2\in\intg\left[\frac{1}{2}\right]$ such that $\max\{\sde_2(v_1),\sde_2(v_2)\}=k$. Then $\sde_2\left(\frac{1}{2}(v_1\pm v_2)\right)=k+1$ or $\sde_2\left(\frac{1}{2}(v_1\pm v_2)\right)\leq k$.
    \label{fact:sdeChangeNo}
\end{fact}
\begin{proof}
Let $v_1=\frac{a}{2^x}$ and $v_2=\frac{b}{2^y}$, where $a,b\in 2\intg+1$ and $x,y\in\nat$. 
\begin{eqnarray}
    \frac{1}{2}(v_1\pm v_2)&=&\frac{1}{2^{x+1}}\left(a\pm b\cdot 2^{x-y}\right)  \nonumber
\end{eqnarray}
If $x>y$ then $\left(a \pm (b\cdot 2^{x-y}) \right) \in 2\intg+1 $ and hence no further reduction in the fraction is possible. Therefore $\sde_2\left(\frac{1}{2}(v_1\pm v_2)\right)=k+1$.

If $x=y$ then $(a\pm b )\in 2\intg$ and hence the exponent of the denominator can decrease by at least 1. So, $\sde_2\left(\frac{1}{2}(v_1\pm v_2)\right)\leq k$.
\end{proof}
\end{comment}

%-----------Required-------
\begin{fact}
    Let $v_1,v_2,v_3,v_4\in\intg\left[\frac{1}{2}\right]$ such that $\max\{\sde_2(v_1),\sde_2(v_2),\sde_2(v_3),\sde_2(v_4)\}=k$. Then $\sde_2\left(\frac{1}{2}(v_1\pm v_2\pm v_3\pm v_4)\right)=k+1$ or $\leq k$.
    \label{fact:sdeChangeNoCS}
\end{fact}
\begin{proof}
Let $v_1=\frac{a}{2^w}$, $v_2=\frac{b}{2^x}$, $v_3=\frac{c}{2^y}$ and $v_4=\frac{d}{2^z}$, where $a,b,c,d\in 2\intg+1$ and $w,x,y,z\in\nat$. Without loss of generality let $\max\{ w,x,y,z \}=x$
\begin{eqnarray}
    \frac{1}{2}(v_1\pm v_2\pm v_3\pm v_4)&=&\frac{1}{2^{x+1}}\left(a\cdot 2^{x-w}\pm b\pm c\cdot 2^{x-y}\pm d\cdot 2^{x-z}\right)  \nonumber
\end{eqnarray}
If the numerator is odd then there is no further reduction of the fraction and $\sde_2$ of the resultant sum is $k+1$. If the numerator is even then there is a reduction of the fraction and the resultant $\sde_2$ is at most $k$.
\end{proof}

We represent each element $v=\frac{a}{2^k}\in\intg\left[\frac{1}{2}\right]$ by a pair $[a,k]$ of integers. In this way, all arithmetic operations involving elements within this ring, for example, matrix multiplication can be done with integer arithmetic. We have given the pseudocode for some of these operations in Appendix \ref{app:pseudocode} (Algorithms \ref{alg:sde2Red} and \ref{alg:add2}).  This also implies that we do not require any floating point operation for our Toffoli-count-optimal synthesis algorithms for exactly implementable unitaries. 

\paragraph{Compact representation of $\chan{G_{P_1,P_2,P_3}}$ : } It is sufficient to represent the $2^{2n}\times 2^{2n}$ matrix $\chan{G_{P_1,P_2,P_3}}$ with an array of length $7\cdot 2^{2n-3}$. For brevity, we write $\chan{G_{P_1,P_2,P_3}}[P_r,P_s]$ as $\chan{G_{P_1,P_2,P_3}}[r,s]$. Each entry of this array is of the form $(r,\pm s, \pm s', \pm s'' )$, which implies $\chan{G_{P_1,P_2,P_3}}[r,r]=\frac{1}{2}$, $\chan{G_{P_1,P_2,P_3}}[r,s]=\pm\frac{1}{2}$, $\chan{G_{P_1,P_2,P_3}}[r,s']=\pm\frac{1}{2}$ and $\chan{G_{P_1,P_2,P_3}}[r,s'']=\pm\frac{1}{2}$. The remaining entries of the matrix can be easily deduced, using Theorem \ref{thm:chanRep}. In Algorithm \ref{alg:chanTOF} (Appendix \ref{app:pseudocode}) we have provided a pseudocode for computing the channel representation of these basis elements and storing them using this compact representation.

\paragraph{Channel representation of $G_{P_1,P_2,P_3}^{\dagger}$ : } Since Toffoli is a self-inverse unitary, so 
$$\chan{G_{P_1,P_2,P_3}}^{\dagger}=\chan{G_{P_1,P_2,P_3}}.$$

%---------------------------------------------------------------------------
\subsection{Multiplication by $\chan{G_{P_1,P_2,P_3}}$}
\label{subsec:mult}

Let $U_p=\chan{G_{P_1,P_2,P_3}}U$, where $U$ is another matrix of dimension $2^{2n}\times 2^{2n}$. Then,
\begin{eqnarray}
    U_p[r,j] = \sum_{k=1}^{2^{2n}} \chan{G_{P_1,P_2,P_3}}[r,k] U[k,j]   \nonumber
\end{eqnarray}
and it can be computed very efficiently with the following observations.

\begin{enumerate}
    \item Suppose $\chan{G_{P_1,P_2,P_3}}[r,r]=1$. From Theorem \ref{thm:chanRep}, $\chan{G_{P_1,P_2,P_3}}[r,s]=0$ for each $s\neq r$. Then,
    \begin{eqnarray}
        U_p[r,j]=\chan{G_{P_1,P_2,P_3}}[r,r]U[r,j]=U[r,j]  \qquad [\forall j\in\{1,\ldots,2^{2n}\} ] \nonumber
    \end{eqnarray}
and so $U_p[r,.]\leftarrow U[r,.]$ i.e. the $r^{th}$ row of $U$ gets copied into the $r^{th}$ row of $U_p$. 

    \item Let $\chan{G_{P_1,P_2,P_3}}[r,r]=\frac{1}{2}$. From Theorem \ref{thm:chanRep}, we know there are 3 other non-zero off-diagonal elements. Let $\chan{G_{P_1,P_2,P_3}}[r,s]=\pm\frac{1}{2}$, $\chan{G_{P_1,P_2,P_3}}[r,s']=\pm\frac{1}{2}$ and $\chan{G_{P_1,P_2,P_3}}[r,s'']=\pm\frac{1}{2}$.
\begin{eqnarray}
    U_p[r,j]&=&\chan{G_{P_1,P_2,P_3}}[r,r]U[r,j]+\chan{G_{P_1,P_2,P_3}}[r,s]U[s,j]+\chan{G_{P_1,P_2,P_3}}[r,s']U[s',j]
    \nonumber \\
    &&+\chan{G_{P_1,P_2,P_3}}[r,s'']U[s'',j] 
     =\frac{1}{2}\left( U[r,j]\pm U[s,j]\pm U[s',j]\pm U[s'',j] \right)   \nonumber
\end{eqnarray}
Thus we see that the $r^{th}$ row of $U_p$ is a linear combination of the $r^{th}, s^{th}, s'^{th}, s''^{th}$ rows of $U$, multiplied by $\frac{1}{2}$. i.e.
$ U_p[r,.]\leftarrow\frac{1}{2}\left( U[r,.]\pm U[s,.]\pm U[s',.]\pm U[s'',.] \right) $.
\end{enumerate}
We have given a pseudocode of this procedure in Algorithm \ref{alg:multGtof} (Appendix \ref{app:pseudocode}).
Since $\frac{1}{8}^{th}$ of the diagonals in $\chan{G_{P_1,P_2,P_3}}$ is 1 (Theorem \ref{thm:chanRep}), so we get the following result.

\begin{theorem}
$U_p=\chan{G_{P_1,P_2,P_3}}U$ can be computed in time $O\left( 7\cdot 2^{4n-3} \right)$.
    \label{thm:chanRepMult}
\end{theorem}
Currently the fastest algorithm to multiply two $2^{2n}\times 2^{2n}$ matrices has a time complexity of $2^{4.7457278n}$ \cite{2014_lG}. This modest asymptotic improvement in the complexity has a significant impact on the actual running time, especially when many such matrix multiplications are performed, as in our algorithms for exactly implementable unitaries. Also, using the representations discussed earlier, we only work with integers. This is not only faster, but also overcomes other floating point arithmetic issues like precision, etc. 

\paragraph{$\sde_2$ of product matrix : } Let $U_p=\chan{G_{P_1,P_2,P_3}}\chan{U}$, where $\chan{U}=\prod_j\chan{G_{P_{1j},P_{2j},P_{3j}}}\chan{C}$ is an exactly implementable unitary and $P_1,P_2,P_3,P_{1j},P_{2j},P_{3j}\in\pauli_n$, $C\in\cliff_n$. From Lemma \ref{lem:chanRepRing}, $\chan{U}[k,\ell]\in\intg\left[\frac{1}{2}\right]$, where $k,\ell\in\{1,\ldots,2^{2n}\}$. Thus entries of $U_p$ are of the form $\frac{1}{2}\left(v_1\pm v_2\pm v_3\pm v_4\right)$, where $v_1,v_2,v_3,v_4\in\intg\left[\frac{1}{2}\right]$. We can apply Fact \ref{fact:sdeChangeNoCS} and conclude that the resulting $\sde_2$ can increase by 1, remain unchanged or decrease.

Now, we know that when we multiply by inverse then the change in sde should be reversed. Thus sde can decrease by at most 1 because it can increase by at most 1. 
\begin{lemma}
    Let $U_p=\chan{G_{P_1,P_2,P_3}}\chan{U}$, where $\chan{U}=\prod_j\chan{G_{P_{1j},P_{2j},P_{3j} } }\chan{C}$  and $P_1,P_2,P_3,P_{1j},P_{2j},P_{3j}\in\pauli_n$, $C\in\cliff_n$. Then $\sde_2(U_p)-\sde_2(\chan{U})\in\{\pm 1,0\}$.
    \label{lem:sdeChangeMat}
\end{lemma}

%---------------------------------------------------------------------
\subsection{Implementation results}
\label{subsec:implementation}

We implemented our heuristic algorithm EXACT-TOF-OPT in Python on a machine with Intel(R) Core(TM)
i7-5500K CPU at 2.4GHz, having 8 GB RAM and running Windows 10. This algorithm has been explained in Section \ref{subsubsec:heuristic}. There are two other algorithms described in Section \ref{sec:method} but their complexity is exponential in number of qubits as well as Toffoli-count. So they become intractable on a personal laptop, running on 1 core. In this heuristic algorithm we basically try to guess a sequence $\widetilde{U}=\prod_{j=1}^mG_{P_{1_j},P_{2_j},P_{3_j}}$ such that $U\widetilde{U}^{\dagger} = C_0\in\cliff_n$, where $U$ is the input unitary. We work with the channel representation of unitaries. We prune the search space depending upon the sde and Hamming weight of the intermediate product unitaries. We have probed two rules in order to select the unitaries - Rule A, where we group according to sde and Rule B, where we group according to sde and Hamming weight. More details can be found in Section \ref{subsubsec:heuristic}. 

To test the heuristic algorithm we generated a number of random 3-qubit unitaries with Toffoli-count at most 5 and 10. The procedure for generating the random unitaries has been explained in Appendix \ref{app:pseudocode} (RANDOM-CHAN-REP, Algorithm \ref{alg:randChanRep}). We found that using Rule A, the output Toffoli-count is at most the input Toffoli-count for all the cases, but the time consumed is much higher. In fact, the space utilized also becomes too high. Using Rule B, the output is at most the input for nearly 80$\%$ of the cases, but the time consumed is much less. Here, we mention that we terminate the algorithm if it runs for more than 1 hour. In Table \ref{tab:resultExact} we have given the implementation results. We could not apply Rule A on input unitaries with 10 Toffolis because it soon ran out of space. The average is computed over 10 input unitaries for each group. These include only those whose implementation was not aborted due to time being more than 1 hour or space consumption being too high. We have also mentioned the maximum number of unitaries that need to be stored while executing the program. This gives a sense about the space complexity.

We also implemented  the unitary $\id\otimes\cs$ using both the rules and found that its Toffoli-count is 3. It took roughly 19 secs to get a decomposition and maximum number of unitaries stored is at most 3. For this particular case we did an exhaustive search in order to find a lower Toffoli-count decomposition. But we found that this is the optimal Toffoli-count for CS gate. A circuit implementing $\id\otimes\cs$ has been shown in Figure \ref{ckt:cs}. More specifically, using our heuristic algorithm EXACT-TOF-OPT, using both the rules we found that we obtain,
\begin{eqnarray}
    \chan{\id\otimes\cs} = \left(\prod_{j=1}^3 \chan{G_{P_{1_j},P_{2_j},P_{3_j}}  }  \right)\chan{C_0},\qquad [C_0\in\cliff_3]   \nonumber
\end{eqnarray}
where $(P_{1_1},P_{2_1},P_{3_1}) = (\id\id\Z,\id\Z\id,\Z\id\id)$, $(P_{1_2},P_{2_2},P_{3_2}) = (\id\id\Y,\id\Z\id,\Z\id\id)$ and $(P_{1_3},P_{2_3},P_{3_3}) = (\id\id\X,\id\Z\id,\Z\id\id)$. It is straightforward to verify the above equation. Now, as discussed in Section \ref{subsec:genTOF}, for each generating set element $G_{P_{1_j},P_{2_j},P_{3_j}}$ we find conjugating Cliffords $C_j\in\cliff_3$ such that $C_j(\id\id\X)C_j^{\dagger} = P_{1_j}$, $C_j(\id\Z\id)C_j^{\dagger} = P_{2_j}$ and $C_j(\Z\id\id)C_j^{\dagger} = P_{3_j}$. Since $\had\X\had=\Z$ and $\phase\X\phase^{\dagger}=\Y$, so we obtain the circuit as shown in Figure \ref{ckt:cs}. We have boxed the corresponding parts of the circuit that implement each generating set unitary.

\begin{table}[h]
    \centering
    \scriptsize
    \begin{tabular}{|c|c|c|c|c|}
    \hline
      \textbf{Input count} & \textbf{Avg. time (A)} & \textbf{Avg. time (B)} & \textbf{Max $\#$unitaries (A)} & \textbf{Max $\#$unitaries (B)}  \\
      \hline
     5  & $\approx$ 46 mins & $\approx$ 100 s & 902 & 7  \\
     \hline
     10 & - & $\approx$ 250s & - & 10 \\
     \hline
    \end{tabular}
    \caption{Implementation results of EXACT-TOF-COUNT on random 3-qubit unitaries. In the first column we have the number of Toffolis in the input unitary. The second and third column gives the average time using Divide-and-Select Rule A and B respectively. The last two columns give the maximum number of unitaries that need to be stored using rules A and B respectively.}
    \label{tab:resultExact}
\end{table}

\begin{comment}
\begin{figure}
    \centering
    \Qcircuit @C=1em @R=1em{
\lstick{1} &\qw &\multigate{2}{C_0} &\ctrl{1} &\gate{S}&\ctrl{1}&\gate{S^{\dagger}} &\gate{H}&\ctrl{1}&\gate{H} &\qw \\
\lstick{2} &\qw &\ghost{C_0 } &\ctrl{1} &\qw&\ctrl{1}&\qw &\qw&\ctrl{1}&\qw &\qw\\
\lstick{3} &\qw &\ghost{C_0} &\targ &\qw&\targ&\qw &\qw&\targ&\qw &\qw \gategroup{1}{4}{3}{4}{.7em}{--}\gategroup{1}{5}{3}{7}{.7em}{--}\gategroup{1}{8}{3}{10}{.7em}{--}
    }
    \caption{An implementation of $\id\otimes\cs$ with 3 Toffoli gates. $C_0$ is a 3-qubit Clifford. The boxed parts implement the three generating set unitaries. The first box implements $G_{\id\id\X,\id\Z\id,\Z\id\id}$, the second one implements $G_{\id\id\Y,\id\Z\id,\Z\id\id}$ and the third box implements $G_{\id\id\Z,\id\Z\id,\Z\id\id}$.  } 
    \label{ckt:cs}
\end{figure}
\end{comment}

\begin{figure}
 \centering
 \includegraphics[width=6cm, height=2cm]{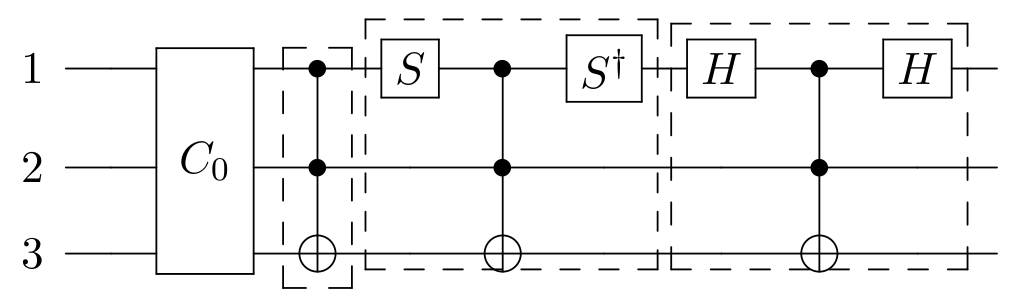}
 \caption{An implementation of $\id\otimes\cs$ with 3 Toffoli gates. $C_0$ is a 3-qubit Clifford. The boxed parts implement the three generating set unitaries. The first box implements $G_{\id\id\X,\id\Z\id,\Z\id\id}$, the second one implements $G_{\id\id\Y,\id\Z\id,\Z\id\id}$ and the third box implements $G_{\id\id\Z,\id\Z\id,\Z\id\id}$.  } 
    \label{ckt:cs}
\end{figure}

%----------------------------------------------------------------------
\subsection{Discussions}
\label{subsec:discussion}

In this paper we develop some theoretical concepts to better understand the unitaries implementable by the Clifford+Toffoli gate set. First is the generating set with which we can represent any unitary exactly or approximately implmentable by this gate set. We derive a bound on the cardinality of this generating set. Second is the analysis of the channel representation of these generating set elements. With the help of these we are able to derive some interesting results like a bound on the Toffoli-count of arbitrary unitaries, impossibility of the exact implementation of T gate, though we know that Toffoli can be exactly implemented by the Clifford+T gate set. This clearly shows that the cardinality of the set of unitaries exactly implementable by these gate sets is different. Such conclusions will help in resource estimates in various applications. 

With these theoretical concepts we develop Toffoli-count-optimal synthesis algorithms for arbitrary multi-qubit unitaries. Such algorithms were previously developed in the context of T-count-optimal synthesis. Complexity of these algorithms depend on the cardinality of the generating set. $|\gen_{Tof}|$ is exponentially more than the size of the generating set for Clifford+T. Also, when we work with channel representation the size of the unitaries increases quadratically. So these algorithms soon become intractable on a personal laptop. Thus we could not implement any of the provable algorithms, because complexity of these depend exponentially on the Toffoli-count. The heuristic algorithm with polynomial dependence on the Toffoli-count, is much faster but it is not always accurate, especially with Rule B. The accuracy depends on the characteristics of the channel representation of the generating set. Thus inspite of using similar heuristics the performance varied for Clifford+T and Clifford+Toffoli gate set. As mentioned, due to the high complexity of the provable algorithms we could not verify if the count returned by the heuristics were optimal. Since output Toffoli-count was at most the Toffoli-count of the input circuit, so we conjecture that they return optimal decomposition.   

Inspite of this, we did get some nice result by implementing the heuristic algorithm, for example, the exact implementation of CS gate. Due to its small Toffoli count we could verify that in this case the heuristics did return the optimal count. Since in a separate paper we show that Toffoli can be exactly implemented by Clifford+CS, so it implies that the set of exactly implementable unitaries by these two universal gate sets is the same. Thus these heuristics are worth probing and can throw some interesting light on the mathematical properties of the unitaries involved. More parameters can be taken into account while dividing and selecting, the selecting criteria maybe changed,etc. 

As mentioned the theoretical concepts and techniques developed in this paper helped us achieve some nice results and develop algorithms in the context of Clifford+Toffoli, that could not be done before. These techniques have a braoder applicability and can be used for other universal gate sets. More efficient algorithms can be designed, possibly by sacrificing the optimality criterion. These concepts can also be used to design Toffoli-depth-optimal quantum circuits, analogous to the work in \cite{2022_GMM} for T-depth-optimality. We leave these for future research.

\section{Method}
\label{sec:method}

In this section we describe algorithms for Toffoli-count-optimal synthesis of multi-qubit unitaries, both exactly and approximately implementable by Clifford+Toffoli gate set.

%-------------------------------------------------------------------
\subsection{Toffoli-count-optimal synthesis algorithm for approximately implementable unitaries}
\label{subsec:algoApprox}

Now we describe an algorithm for determining the $\epsilon$-Toffoli-count of an $n$-qubit unitary $W$. Basically we use the procedure in \cite{2022_GMM2}, but with the generating set $\gen_{Tof}$.  This algorithm has space and time complexity $O\left(|\gen_{Tof}|\right)$ and $O\left(|\gen_{Tof}|^{\tofeps}\right)$ respectively.  We briefly describe our algorithm here. More detail pseudocode can be found in Appendix \ref{app:pseudocode}. The optimization version, APPROX-TOF-OPT (Algorithm \ref{alg:min}), that finds the optimal Toffoli-count is an exhaustive search procedure. In each iteration it calls the decision version APPROX-TOF-DECIDE (Algorithm \ref{alg:decide}) in order to decide if $\tofeps(W)=m$, for increasing values of a variable $m$.  

The main idea to solve the decision version is as follows. Suppose we have to test if $\tofeps(W)=m$. From the definitions given in Section \ref{sec:prelim}, we know that if this is true then $\exists U\in\clifft_n^{Tof}$ such that $\tofcount(U)=m$ and $d(U,W)\leq\epsilon$. Let $U=\left(\prod_{j=m}^1G_{P_{1_j},P_{2_j},P_{3_j}}\right)C_0$ where $C_0\in\cliff_n$ and $G_{P_{1_j},P_{2_j},P_{3_j}}\in\gen_{Tof}$. Also, we can write $W=UE$, for some unitary $E$. Let $\widetilde{U}=\prod_{j=m}^1G_{P_{1_j},P_{2_j},P_{3_j}}$.
\begin{eqnarray}
 && d(W,U)=\sqrt{1-\frac{|\tr(W^{\dagger}U)|}{N}}\leq\epsilon 
 \implies |\tr(W^{\dagger}U)| \geq N(1-\epsilon^2)  \label{eqn:dVU} \\
&& \left|\tr\left(E^{\dagger}\right)\right|=\left|\tr\left(E\right)\right|\geq N(1-\epsilon^2)
 \label{eqn:trE}
\end{eqnarray}
The above implies that $d(E,\id)\leq\epsilon$. We have,
\begin{eqnarray}
\left|\tr\left(W^{\dagger}\left(\prod_{j=m}^1G_{P_{1_j},P_{2_j},P_{3_j}}\right)\right)\right| &=& \left|\tr\left(E^{\dagger}U^{\dagger}\left(\prod_{j=m}^1G_{P_{1_j},P_{2_j},P_{3_j}}\right)\right)\right| \nonumber\\
&=& \left|\tr\left(E^{\dagger}e^{-i\phi}C_0^{\dagger}\left(\prod_{j=1}^mG_{P_{1_j},P_{2_j},P_{3_j}}^{\dagger}\right)\left(\prod_{j=m}^1G_{P_{1_j},P_{2_j},P_{3_j}}\right)\right)\right| \nonumber \\
 &=&\left|\tr\left(E^{\dagger}C_0^{\dagger}\right)\right|=\left|\tr\left(EC_0\right)\right|. \nonumber
\end{eqnarray}
 To test if we have guessed the correct $\widetilde{U}$, we perform the following tests.

\textbf{I. Amplitude test : } We use the following theorem, which says that if $E$ is close to identity then distribution of absolute-value-coefficients of $EC_0$ and $C_0$ in the Pauli basis expansion, is almost similar.
\begin{theorem} [\cite{2022_GMM2}]
Let $E\in\mathcal{U}_n$ be such that $\left|\tr\left(E\right)\right|\geq N\left(1-\epsilon^2\right)$, for some $\epsilon\geq 0$. $C_0=\sum_{P\in\pauli_n}r_PP$ is an $n$-qubit Clifford. If $\left|\left\{P:r_P\neq 0\right\}\right|=M$ then 
\begin{eqnarray}
\frac{1-\epsilon^2}{\sqrt{M}}-\sqrt{M}\sqrt{2\epsilon^2-\epsilon^4}&\leq&\left|\tr\left(EC_0P_1\right)/N\right|\leq \frac{1}{\sqrt{M}}+\sqrt{M}\sqrt{2\epsilon^2-\epsilon^4}\quad [\text{if } r_{P_1}\neq 0]   \label{eqn:traceRPn0} \\
\text{and}\quad 0&\leq&\left|\tr\left(EC_0P_1\right)/N\right|\leq \sqrt{M}\sqrt{2\epsilon^2-\epsilon^4}\quad [\text{if }r_{P_1}=0] \label{eqn:traceRP0}
\end{eqnarray}
 \label{thm:coeffEC}
\end{theorem}
In fact, we can have a more general theorem as follows.
\begin{theorem}[\cite{2022_GMM2}]
Let $E\in\mathcal{U}_n$ be such that $\left|\tr\left(E\right)\right|\geq N\left(1-\epsilon^2\right)$, for some $\epsilon\geq 0$. $Q=\sum_{P\in\pauli_n}q_PP$ is an $n$-qubit unitary. Then for each $P_1\in\pauli_n$,
\begin{eqnarray}
(1-\epsilon^2)|q_{P_1}|-\sum_{P\in\pauli_n\setminus\{P_1\}}|q_P|\sqrt{2\epsilon^2-\epsilon^4} &\leq&\left|\tr\left(EQP_1\right)/N\right|\leq |q_{P_1}|+\sum_{P\in\pauli_n\setminus\{P_1\}}|q_P|\sqrt{2\epsilon^2-\epsilon^4}. \nonumber
\end{eqnarray}

 \label{thm:coeffEQ}
\end{theorem}
So, we calculate $W'=W^{\dagger}\widetilde{U}$ and then calculate the set $\mathcal{S}_c=\left\{|\tr(W'P)/N|:P\in\pauli_n\right\}$, of coefficients. Then we check if we can distinguish the following two subsets $\mathcal{S}_0$ and $\mathcal{S}_1$, for some $1\leq M\leq N^2$. Further details have been given in Algorithm \ref{alg:decide}.
\begin{eqnarray}
 \mathcal{S}_1&=&\left\{\left|\tr\left(EC_0P_1\right)/N\right|:r_{P_1}\neq 0\right\}    \label{eqn:S1} \\
 \mathcal{S}_0&=&\left\{\left|\tr\left(EC_0P_1\right)/N\right|:r_{P_1}= 0\right\}    \label{eqn:S0} 
\end{eqnarray}
After passing this test we have a unitary of the form $E^{\dagger}Q$ where $Q$ is a unitary with \emph{equal} or \emph{nearly equal} amplitudes or coefficients (absolute value) at some points and zero or nearly zero at other points.

\textbf{II. Conjugation test : } 
To ensure $Q$ is a Clifford i.e. $W'=E^{\dagger}C_0$ for some $C_0\in\cliff_n$, we perform the \textbf{conjugation test} (Algorithm \ref{alg:conj}) for further verification. We use the following theorem and its corollary.
\begin{theorem}[\cite{2022_GMM2}]
Let $E,Q\in\mathcal{U}_n$ such that $d(E,\id)\leq\epsilon$. $P'\in\pauli_n\setminus\{\id\}$ such that $QP'Q^{\dagger}=\sum_{P}\alpha_{P}P$, where $\alpha_{P}\in\cmplx$. Then for each $P''\in\pauli_n$,
\begin{eqnarray} 
\min\{0,|\alpha_{P''}|(1-4\epsilon^2+2\epsilon^4) 
-2\epsilon\sum_{P\neq P''}|\alpha_P| \}
&\leq& \left|\tr\left((E^{\dagger}QP'Q^{\dagger}E)P''\right)\right|/N \nonumber \\
&\leq& \max\{ |\alpha_{P''}|+2\epsilon\sum_{P\neq P''}|\alpha_P|, 1\} 
\nonumber   
\end{eqnarray}

\label{thm:conj}
\end{theorem}

\begin{corollary}[\cite{2022_GMM2}]
Let $C_0\in\cliff_n$ and $P'\in\pauli_n$ such that $C_0P'C_0^{\dagger}=\tilde{P}\in\pauli_n$. If $E\in\mathcal{U}_n$ such that $d(E,\id)\leq\epsilon$, then
\begin{eqnarray}
 (1-4\epsilon^2+2\epsilon^4) &\leq&\left|\tr\left((E^{\dagger}C_0P'C_0^{\dagger}E)P''\right)\right|/N\leq 1 \qquad \text{if }  P''=\tilde{P}  \nonumber   \\
0&\leq&\left|\tr\left((E^{\dagger}C_0P'C_0^{\dagger}E)P''\right)\right|/N\leq  2\epsilon  \qquad \text{else. }   \nonumber
\end{eqnarray}
 \label{cor:conjCliff}
\end{corollary}

The above theorem and corollary basically says that $EC_0$ (or $E^{\dagger}C_0$) \emph{approximately} inherits the conjugation property of $C_0$. For each $P'\in\pauli_n$, if we expand $C_0P'C_0^{\dagger}$ in the Pauli basis then the absolute value of the coefficients has value 1 at one point, 0 in the rest. If we expand $EC_0P'C_0^{\dagger}E^{\dagger}$ in the Pauli basis then one of the coefficients (absolute value) will be almost 1 and the rest will be almost 0. From Theorem \ref{thm:conj} this pattern will not show for at least one Pauli $P'''\in\pauli_n$ if we have $E^{\dagger}Q$, where $Q\notin\cliff_n$. If we expand $EQP'''Q^{\dagger}E^{\dagger}$ or $E^{\dagger}QP'''Q^{\dagger}E$ in the Pauli basis then the \emph{spike in the amplitudes} will be in at least two points. 

\subsubsection{Synthesizing Toffoli-count-optimal circuits}
\label{subsec:cktSynth}

From APPROX-TOF-DECIDE (Algorithm \ref{alg:decide}) we can obtain a sequence $\{U_m,\ldots,U_1\}$ of unitaries such that $U=\left(\prod_{j=m}^1U_j\right)C_0e^{i\phi}$, where $U_j = G_{P_{1_j},P_{2_j},P_{3_j}}\in\gen_{Tof}$, $m = \tofcount(U) = \tofeps(W)$ and $C_0\in\cliff_n$. We can efficiently construct circuits for each $U_j\in\gen_{Tof}$, as discussed earlier. In order to determine the trailing $C_0$, we observe that from Algorithm \ref{alg:decide} we can also obtain the following information : (1) set $\mathcal{S}_1$, as defined in Equation \ref{eqn:S1}, (2) $M=|\mathcal{S}_1|$. Thus we can calculate $r=\frac{1}{\sqrt{M}}$ and from step \ref{decide:Sc} we can actually calculate the set $\widetilde{\mathcal{S}_1}=\left\{(t_P,P): |t_P|=|\tr(W'P)/N|\in\mathcal{S}_1\right\}$, where $W'=W^{\dagger}\widetilde{U} = e^{-i\phi}E^{\dagger}C_0^{\dagger}$. We perform the following steps.
\begin{enumerate}
\item Calculate $a_P=\frac{t_P}{t_{\id}}=\frac{\tr(W'P)}{\tr(W')}$, where $(t_P,P)\in\widetilde{\mathcal{S}_1}$ (or equivalently $|t_P|\in\mathcal{S}_1$). We must remember that $(t_{\id},\id)\in\widetilde{\mathcal{S}_1}$.  
 
 It has been argued in \cite{2022_GMM2} that for small enough $\epsilon$ (say $\ll\frac{1}{M}$), $a_P\approx\frac{\overline{r_P}}{\overline{r_{\id}}}$. From Fact \ref{fact:cliffCoeff} we know that $\frac{|\overline{r_P}|}{|\overline{r_{\id}}|}=\frac{|r_P|}{|r_{\id}|}=1$. So we adjust the fractions such that their absolute value is $1$. For small enough $\epsilon$ this adjustment is not much and so with a slight abuse we use the same notation for the adjusted values.
 
 \item Select $c,d\in\real$ such that $c^2+d^2=r^2$. Let $\widetilde{r_{\id}}=c+di$. Then we claim that the Clifford $\widetilde{C_0}=\widetilde{r_{\id}}\sum_{P:r_P\neq 0}\overline{a_P}P$ is sufficient for our purpose. 
\end{enumerate}
It is not hard to see that $\widetilde{C_0}=e^{i\phi'}C_0$ for some $\phi'\in [0,2\pi)$. Thus if $U'=\left(\prod_{j=m}^1U_j\right)\widetilde{C_0}$, then $\tofcount(U')=\tofcount(U)$ and $d(U',W)\leq\epsilon$. After determining the Clifford, we can efficiently construct a circuit for it, for example, by using results in \cite{2004_AG}.

\subsubsection{Complexity} 

 Now, we analyse the time and space complexity of the above algorithm.

\subsubsection*{Time complexity}

We first analyse the time complexity of $\mathcal{A}_{CONJ}$. The outer and inner loop at steps 2 and 6, respectively, can run at most $2^{2n}* 2^{2n}= 2^{4n}$ times, where $N=2^n$. The complexity of the inner loop is dominated by the multiplications at step 13, which takes $O(2^{2n})$ time. Thus overall complexity of this sub-routine is $O(2^{6n})$.

Now we analyse the time complexity of APPROX-TOF-DECIDE, when testing for a particular Toffoli-count $m$. The algorithm loops over all possible products of $m$ unitaries $G_{P_{1_j},P_{2_j},P_{3_j}} \in\gen_{Tof}$. Thus number of iterations (steps 2-14) is at most $|\gen_{Tof}|^m \in O\left( n^{2m} 64^{(n-2)m} \right) \in O\left( n^{2m} 2^{(6n-12)m} \right)   $, from Theorem \ref{thm:genTOFsize} and this is the most expensive part of this algorithm. The $m$ multiplications at step 2 and the multiplications, trace and sorting at step 4 takes $O(m2^{2n}+2^{4n})$ time. The inner loop 5-13 happens $O(2^{2n})$ times. Each of the list elements is checked and so step 8 has complexity $O(2^{2n})$. Now let within the inner loop the conjugation test is called $k_1$ times. So the loop 5-13 incurs a complexity $O(k_1\cdot 2^{6n}+(2^{2n}-k_1)2^{2n})$, when $k_1>0$, else it is $O(2^{4n})$. Let $k$ is the number of outer loops (steps 2-14) for which conjugation test is invoked in the inner loop 5-13 and $k_1$ is the maximum number of times this test is called within any 5-13 inner loop. Then the overall complexity of APPROX-TOF-DECIDE is $O((n^{2m}2^{(6n-12)m}-k)\cdot (m2^{2n}+2^{4n})+k\cdot (m2^{2n}+2^{4n}+k_12^{6n}+(2^{2n}-k_1)2^{2n} ) )$. 

The conjugation test is invoked only if a unitary passes the amplitude test. We assume that $m<2^{2n}$ and the occurrence of non-Clifford unitaries with equal amplitude is not so frequent such that $kk_1< n^{2m}2^{(6n-12)m}-k$. Thus APPROX-TOF-DECIDE has a complexity of $O(n^{2m}2^{(6n-12)m})$, for one particular value of $m$.  Hence, the overall algorithm APPROX-TOF-OPT has a time complexity $ O\left(n^{2\tofeps(W)}  2^{(6n-12)\tofeps(W)}\right) $.

\subsection*{Space complexity}

The input to our algorithm is a $2^n\times 2^n$ unitary, with space complexity $O(2^{2n})$. We need to store $\gen_{Tof}$, either as triples of Paulis or as matrices. It takes space $O\left(|\gen_{Tof}| \right) \in O\left(n^22^{6n-12}\right)$ (Theorem \ref{thm:genTOFsize}).  All the operations in APPROX-TOF-DECIDE or $\mathcal{A}_{CONJ}$ either involve matrix multiplications or storing at most $2^{2n}$ numbers. Hence the overall space complexity is $O\left(n^22^{6n-12}\right)$.

%-------------------------------------------------------------------
\subsection{Toffoli-count-optimal synthesis algorithms for exactly implementable unitaries}
\label{subsec:algoExact}

In this section we describe two algorithms for synthesizing Toffoli-count-optimal circuits of exactly implementable unitaries. For the first one we can prove a rigorous bound on the optimality as well as space and time complexity, both of whom are exponential in the Toffoli-count. The exponentially better second algorithm has both space and time complexity polynomial in the Toffoli-count, but the claimed optimality as well as complexity depends on some conjectures. 

\subsubsection{Nested MITM}
\label{subsubsec:nestMITM}

In this section we describe the nested meet-in-the-middle (MITM) framework, introduced in \cite{2021_MM}, that can be used to synthesize Toffoli-count-optimal circuits. It is based on the following result that can be proved with similar arguments as in Lemma 3.1 of \cite{2022_GMM}.
\begin{lemma}
Let $S_i\subset U(2^n)$ be the set of all unitaries implementable in Toffoli-count $i$ over the Clifford+Toffoli gate set. Given a unitary $U$, there exists a Clifford+Toffoli circuit of Toffoli-count $(m_1+m_2)$ implementing $U$ if and only if $S_{m_1}^{\dagger}U\bigcap S_{m_2}\neq\emptyset$.
 \label{lem:provNest}
\end{lemma}
The pseudocode of the algorithm has been given in Appendix \ref{app:pseudocode} (Algorithm \ref{app:alg:nestMITM}). We briefly describe this algorithm here. The input consists of the unitary $U$, generating set $\gen_{Tof}$, test Toffoli-count $m$ and $c\geq 2$ that indicates the extent of nesting or recursion we want in our meet-in-the-middle approach. If $U$ is of Toffoli-count at most $m$ then the output consists of a decomposition of $U$ into smaller Toffoli-count unitaries, else the algorithm indicates that $U$ has Toffoli-count more than $m$. 

The algorithm consists of $\left\lceil\frac{m}{c}\right\rceil$ iterations and in the $i^{th}$ such iteration we generate circuits of depth $i$ ($S_{i}$) by extending the circuits of depth $i-1$ ($S_{i-1}$) by one more level. Then we use these two sets to search for circuits of depth at most $ci$. In order to increase the efficiency of the search procedure we define the following. Let $\chan{\clifft_n^{Tof}} = \{\chan{W} : W\in\clifft_n^{Tof} \}$.
\begin{definition}[\textbf{Coset label} \cite{2014_GKMR}]
Let $\chan{W}\in\chan{\clifft_n^{Tof}}$. Its coset label $\chan{W}^{(co)}$ is the matrix obtained by the following procedure.
(1) Rewrite $\chan{W}$ so that each nonzero entry has a common denominator, equal to $\sqrt{2}^{\sde(\chan{W})}$. (2) For each column of $\chan{W}$, look at the first non-zero entry (from top to bottom) which we write as $v=\frac{a}{2^{\sde(\chan{W})}}$. If $a<0$, or if $a=0$ and $b<0$, multiply every element of the column by $-1$. Otherwise, if $a>0$, or $a=0$ and $b>0$, do nothing and move on to the next column. (3) After performing this step on all columns, permute the columns so that they are ordered lexicographically from left to right.
 \label{defn:cosetLabel}
\end{definition}
The following result shows that this is a kind of equivalence function. 
\begin{theorem}[\textbf{Proposition 2 in \cite{2014_GKMR}}]
Let $\chan{W},\chan{V}\in\chan{\clifft_n^{Tof}}$. Then $\chan{W}^{(co)}=\chan{V}^{(co)}$ if and only if $\chan{W}=\chan{V}\chan{C}$ for some $C\in\cliff_n$.
 \label{thm:cosetLabel}
\end{theorem}

In our algorithm the search is performed iteratively where in the $k^{th}$ ($1\leq k\leq c-1$) round we generate unitaries of Toffoli-count at most $ki$ by taking $k$ unitaries $\chan{W_1},\chan{W_2},\ldots,\chan{W_k}$ where $\chan{W_i}\in S_i$ or $\chan{W_i}\in S_{i-1}$. Let $\chan{W}=\chan{W_1}\chan{W_2}\ldots \chan{W_k}$ and its Toffoli-count is $k'\leq ki$. We search for a unitary $\chan{W'}$ in $S_i$ or $S_{i-1}$ such that $\left(\chan{W}^{\dagger}\chan{U}\right)^{(co)}=\chan{W'}$. By Lemma \ref{lem:provNest} if we find such a unitary it would imply that Toffoli-count of $U$ is $k'+i$ or $k'+i-1$ respectively. In the other direction if the Toffoli-count of $U$ is either $k'+i$ or $k'+i-1$ then there should exist such a unitary $\chan{W'}$ in $S_i$ or $S_{i-1}$ respectively. Thus if the Toffoli-count of $U$ is at most $m$ then the algorithm terminates in one such iteration and returns a decomposition of $U$. This proves the \textbf{correctness} of this algorithm.

\paragraph{Time and space complexity}

With the concept of coset label we can impose a strict lexicographic ordering on unitaries such that a set $S_i$ can be sorted with respect to this ordering in $O\left(|S_i|\log |S_i|\right)$ time and we can search for an element in this set in $O\left(\log |S_i|\right)$ time. Now consider the $k^{th}$ round of the $i^{th}$ iteration (steps \ref{nestMITM:whileStart}-\ref{nestMITM:whileEnd} of Algorithm \ref{app:alg:nestMITM} in Appendix \ref{app:pseudocode}). We build unitaries $\chan{W}$ of Toffoli-count at most $ki$ using elements from $S_i$ or $S_{i-1}$. Number of such unitaries is at most $|S_i|^{k}$. Given a $\chan{W}$, time taken to search for $W'$ in $S_i$ or $S_{i-1}$ such that $\left(\chan{W}^{\dagger}\chan{U}\right)^{(co)}=\chan{W'}$ is $O\left(\log |S_i|\right)$. Since $|S_j|\leq |\gen_{Tof}|^j$, so the $k^{th}$ iteration of the for loop within the $i^{th}$ iteration of the while loop, takes time $O\left(|\gen_{Tof}|^{(c-1)i}\log |\gen_{Tof}| \right)$. Thus the time taken by the algorithm is $O\left(|\gen_{Tof}|^{(c-1)\left\lceil\frac{m}{c}\right\rceil} \log |\gen_{Tof}| \right)$. 

In the algorithm we store unitaries of depth at most $\left\lceil\frac{m}{c}\right\rceil$. So the space complexity of the algorithm is $O\left(|\gen_{Tof}|^{\left\lceil\frac{m}{c}\right\rceil} \right)$. Since $|\gen_{Tof}|\in O\left(n^22^{6n-12} \right)$ (Theorem \ref{thm:genTOFsize}), so we have an algorithm with space complexity $O\left(n^{2\lceil\frac{m}{c}\rceil}2^{(6n-12)\left\lceil\frac{m}{c}\right\rceil} \right)$ and time complexity $O\left(n^{2(c-1)\lceil\frac{m}{c}\rceil}2^{(6n-12)(c-1)\left\lceil\frac{m}{c}\right\rceil} \right)$. 

%----------------------------------
\subsubsection{A polynomial complexity heuristic algorithm}
\label{subsubsec:heuristic}

In this section we describe a heuristic algorithm that returns a Toffoli-count-optimal circuit with time and space complexity that are polynomial in the Toffoli-count. Our algorithm is motivated by the polynomial complexity T-count-optimal-synthesis algorithm in \cite{2021_MM}. The conjectured exponential advantage comes from an efficient way of pruning the search space. The intermediate unitaries are grouped according to some properties and then one such group is selected according to some criteria. 

In the previous provable nested meet-in-the-middle algorithm we search for a set of $\chan{G_{P_{1_j},P_{2_j},P_{3_j}}}$ such that 
$\chan{U}^{\dag}\prod_{j=\tofcount(U)}^1\chan{G_{P_{1_j},P_{2_j},P_{3_j}} }$ is $\chan{C_0}$ for some $C_0\in\cliff_n$. Alternatively, we can also search for a set of $\chan{G_{P_{1_j},P_{2_j},P_{3_j}} }^{-1}$ such that $\left(\prod_{j=1}^{\tofcount(U)}\chan{G_{P_{1_j},P_{2_j},P_{3_j}} }^{-1}\right)\chan{U}$ is $\chan{C_0}$ for some $C_0\in\cliff_n$, which is the approach taken by our heuristic algorithm. The optimization version EXACT-TOF-OPT (Algorithm \ref{alg:TOFcountOpt} in Appendix \ref{app:pseudocode}) is solved by repeatedly calling EXACT-TOF-DECIDE (Algorithm \ref{alg:TOFcountDecide} in Appendix \ref{app:pseudocode}), the decision version, which tests whether the Toffoli-count of an input unitary is at most an integer $m$ . 

It will be useful to compare the procedure EXACT-TOF-DECIDE with building a tree, whose nodes store some unitary and the edges represent $\chan{G_{P_1,P_2,P_3} }^{-1}$. Thus the unitary in a child node is obtained by multiplying the unitary of the parent node with the $\chan{G_{P_1,P_2,P_3} }^{-1}$ represented by the edge. Number of children nodes of any parent node is at most $|\gen_{Tof}|-1$. The root stores $\chan{U}$ and we assume it is at depth 0 (Figure \ref{fig:heuristic}). We stop building this tree the moment we reach a node which stores $\chan{C_0}$ for some $C_0\in\cliff_n$. This implies that the path from the root to this leaf gives a decomposition of $\chan{U}$. In these kinds of searches the size of the tree is one of the important factors that determine the complexity of the algorithm. To reduce the complexity we prune this tree (Figure \ref{fig:heuristic}).

At each level we group the nodes according to some ``properties'' or ``parameters'' of the unitaries stored in them. We hope that these parameters will ``distinguish'' the ``correct'' nodes at each level or depth of the tree and thus we would get a decomposition. Note there can be more than one decomposition of $\chan{U}$ with the same or different Toffoli-count. By ``correct'' nodes we mean those that occur in a Toffoli-count-optimal decomposition  of $\chan{U}$. If the parameters always select only the correct nodes then we expect to get much fewer nodes at each level of the tree and the number of levels we have to build is $\tofcount(U)$. But the parameters we selected did not always distinguish the correct nodes and there were some false positives. In order for the algorithm to succeed we have to be careful so that we do not lose all correct nodes at any level and to make it efficient we have to ensure that the number of false positives are not too large and they eventually get eliminated.

We selected two parameters - sde and Hamming weight of the unitaries. We know from Lemma \ref{lem:sdeChangeMat} that sde of a child node unitary can differ by at most 1 from its parent node unitary. While building a unitary we start with $\chan{\id_n}$ and multiply by subsequent $\chan{G_{P_{1_j},P_{2_j},P_{3_j}} }$ till we reach $\chan{U}$. We have observed that in most of these multiplications the sde increases by 1 and the Hamming weight also gradually increases until it (Hamming weight) reaches the maximum. So while doing the inverse operations i.e. decomposing $\chan{U}$ we expect that in most of the steps sde will decrease by 1 and as we get close to the identity, the Hamming weight will also likely decrease. If we multiply by a ``wrong'' $\chan{G_{P_{1_j},P_{2_j},P_{3_j}} }^{-1}$ we expect to see the same changes with much less probability, which is the probability of the false positives. This helps us distinguish the ``correct'' and ``wrong'' nodes.

Specifically, at each level we divide the set $S$ of nodes into some subsets and select one of them. Below are two possible ways to divide the nodes that we have found effective. Suppose in one instance of EXACT-TOF-DECIDE, $m$ is the target depth (maximum depth of the tree to be built) and we have built the tree till depth $i$.

\textbf{Divide-and-Select Rules :}
\begin{enumerate}
 \item[A] We divide into two sets - $S_0$ (sde increase) and $S_1$ (sde non-increase). We select the set with the minimum cardinality such that the sde of the unitaries in this set is at most $m-i$.  
 
% \item[B.] We divide into 4 sets - $S_{00}$ (both sde and Hamming weight increase), $S_{01}$ (sde increase, Hamming weight decrease), $S_{10}$ (sde decrease, Hamming weight increase) and $S_{11}$ (both sde and Hamming weight decrease). Nodes with unchanged Hamming weight but sde increase are put in both $S_{00}$ and $S_{01}$, while nodes with unchanged Hamming weight but sde decrease are put in both $S_{10}$ and $S_{11}$. We select the set with the minimum cardinality such that the sde of the unitaries in this set is at most $m-i$. We include in it the nodes with unchanged sde (irrespective of the change in Hamming weight).
 
 \item[B.] We divide into 9 sets - $S_{00}$ (both sde and Hamming weight increase), $S_{01}$ (sde increase but Hamming weight decrease), $S_{02}$ (sde increase but Hamming weight same), $S_{10}$ (sde decrease, Hamming weight increase) and $S_{11}$ (both sde and Hamming weight decrease), $S_{12}$ (sde decrease but Hamming weight same), $S_{20}$ (sde same but Hamming weight increase), $S_{21}$ (sde same but Hamming weight decrease), $S_{22}$ (both sde and Hamming weight same). We select the set with the minimum cardinality such that the sde of the unitaries in this set is at most $m-i$.
\end{enumerate}
We follow any one of the above methods of divide-and-select throughout the algorithm. We note that in each of the above methods, if the cardinality of the unitaries in the selected set is more than $m-i$ then it implies we cannot get sde $0$ nodes within the next few levels.

\begin{figure}[h]
\centering
\includegraphics[width=7.5cm, height=3cm]{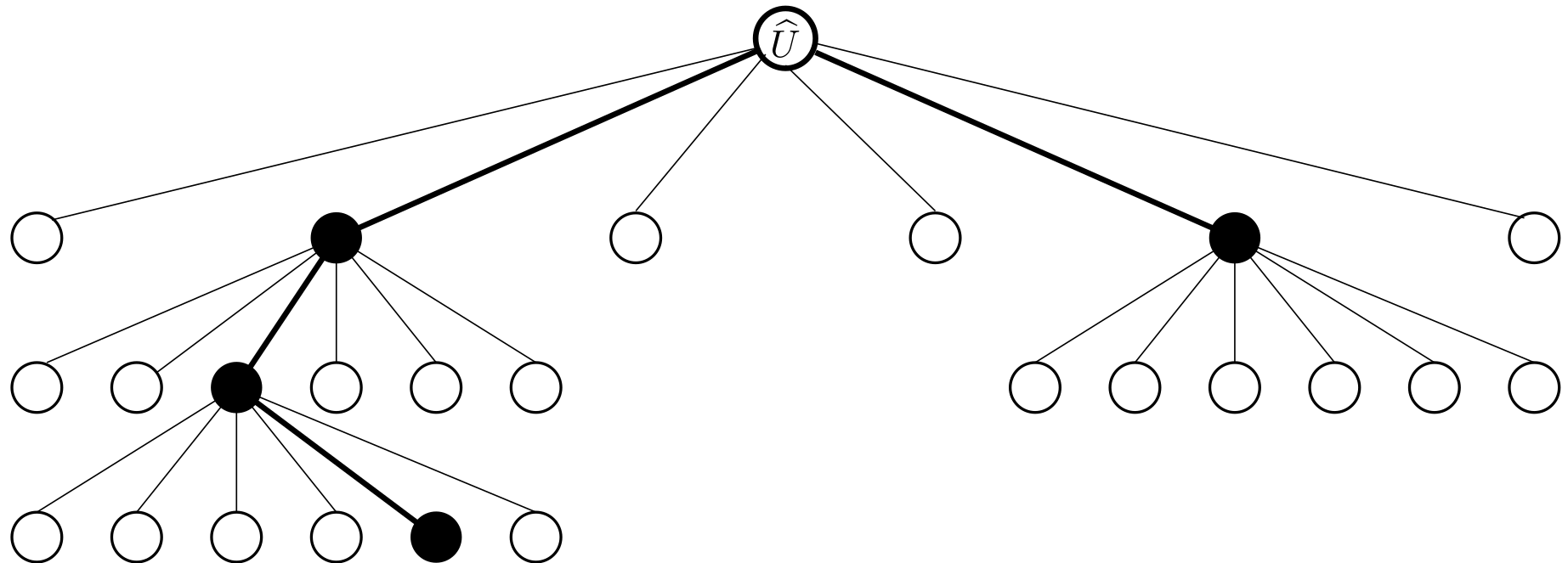}
\caption{The tree built in the heuristic procedure. At each level we select a set of nodes according to some changes in the properties of the child unitaries with respect to their parents, for example change in sde and Hamming weight. Unitaries in the next level are generated from the selected set (black nodes). We stop building the tree as soon as we reach a unitary with sde 0. The path length to this node (in this case 3) is the Toffoli-count of the unitary $U$. It also gives us a Toffoli-count-optimal decomposition of $U$.
}
 \label{fig:heuristic}  
\end{figure}

\paragraph*{Space and time complexity} The analysis of space and time complexity of the algorithm EXACT-TOF-OPT is based on the following assumptions. 
\begin{conjecture}
Let $\widetilde{U}$ is the set of intermediate unitaries stored at any time, or equivalently it is the set of selected children nodes at any level of the above-mentioned tree.
 The cardinality of the set $\widetilde{U}$ in each iteration of EXACT-TOF-DECIDE is at most $\poly(n^22^{6n-12})$, when either method A or B of divide-and-select is applied.  We get at least one Toffoli-count-optimal decomposition.
\label{conj:heuristic}
\end{conjecture}

Consider the sub-routine EXACT-TOF-DECIDE. There are $|\gen_{Tof}|-1$ multiplications by $\chan{G_{P_{1_j},P_{2_j},P_{3_j}} }^{-1}$ in each iteration for each unitary in $\widetilde{U}$. And by the above conjecture $|\widetilde{U}|\in \poly(n^22^{6n-12})$. Thus both the time and space complexity of EXACT-TOF-DECIDE are in $\poly(n^22^{6n-12},m)$.
We call EXACT-TOF-DECIDE at most $\tofcount(U)$ times to solve EXACT-TOF-OPT. Thus space and time complexity of EXACT-TOF-OPT are in $\poly(n^22^{6n-12},\tofcount(U))$.

\section*{Acknowledgements}

P.Mukhopadhyay thanks Nathan Wiebe for helpful discussions.

\section*{Code and data availability}

All codes can be found online at the public repository : https://github.com/PriyankaMukhopadhyay/Toffoli-opt

%-------------------------------------------------------------------------------
%\bibliographystyle{alpha}
%\bibliography{tof_ref}
\newcommand{\etalchar}[1]{$^{#1}$}

%------------------------------------------------------------------
\appendix

\section{Some additional preliminaries}
\label{app:prelim}

\subsection{Cliffords and Paulis}
\label{app:clifford}

The \emph{single qubit Pauli matrices} are as follows.
\begin{eqnarray}
 \X=\begin{bmatrix}
     0 & 1 \\
    1 & 0
    \end{bmatrix} \qquad  
 \Y=\begin{bmatrix}
     0 & -i \\
     i & 0
    \end{bmatrix} \qquad 
 \Z=\begin{bmatrix}
     1 & 0 \\
     0 & -1
    \end{bmatrix}\nonumber
\label{eqn:Pauli1}
\end{eqnarray}
Parenthesized subscripts are used to indicate qubits on which an operator acts For example, $\X_{(1)}=\X\otimes\id^{\otimes (n-1)}$ implies that Pauli $\X$ matrix acts on the first qubit and the remaining qubits are unchanged.

The \emph{$n$-qubit Pauli operators} are :
$
 \pauli_n=\{Q_1\otimes Q_2\otimes\ldots\otimes Q_n:Q_i\in\{\id,\X,\Y,\Z\} \}.
$

The \emph{single-qubit Clifford group} $\cliff_1$ is generated by the Hadamard and phase gates :
$
 \cliff_1=\braket{\had,\phase} 
 $
where
\begin{eqnarray}
 \had=\frac{1}{\sqrt{2}}\begin{bmatrix}
       1 & 1 \\
       1 & -1
      \end{bmatrix}\qquad 
 \phase=\begin{bmatrix}
       1 & 0 \\
       0 & i
      \end{bmatrix}.\nonumber
\end{eqnarray}
When $n>1$ the \emph{$n$-qubit Clifford group} $\cliff_n$ is generated by these two gates (acting on any of the $n$ qubits) along with the two-qubit $\CNOT=\ket{0}\bra{0}\otimes\id+\ket{1}\bra{1}\otimes\X$ gate (acting on any pair of qubits). 

The Clifford group is special because of its relationship to the set of $n$-qubit Pauli operators. Cliffords map Paulis to Paulis, up to a possible phase of $-1$, i.e. for any $P\in\pauli_n$ and any $C\in\cliff_n$ we have
$
    CPC^{\dagger}=(-1)^bP'
$
for some $b\in\{0,1\}$ and $P'\in\pauli_n$. In fact, given two Paulis (neither equal to the identity), it is always possible to efficiently find a Clifford which maps one to the other.
\begin{fact}[\cite{2014_GKMR}]
 For any $P,P'\in\pauli_n\setminus\{\id\} $ there exists a Clifford $C\in\cliff_n$ such that $CPC^{\dagger}=P'$. A circuit for $C$ over the gate set $\{\had,\phase,\CNOT\}$ can be computed efficiently (as a function of $n$).
 \label{fact:cliffConj}
\end{fact}

\begin{fact}[\cite{2022_GMM2}]
Let $Q=\sum_{P\in\pauli_n}q_PP$ be the expansion of a matrix $Q$ in the Pauli basis. Then
$$q_P=\tr\left(QP\right)/N \qquad [N=2^n].
$$
Further, if $Q$ is a unitary then
$$
    \sum_{P\in\pauli_n}\left|q_P\right|^2=1
$$
 \label{app:fact:trCoeff}
\end{fact}

We observe the following when expanding a Clifford in the Pauli basis. 
\begin{fact}[\cite{2010_BS}]
 If $C\in\cliff_n$ then for each $P\in\pauli_n$ $\exists r_P\in\cmplx$, such that $C=\sum_{P\in\pauli_n}r_PP$. Further, if $r_P, r_{P'}\neq 0$ for any pair of $P,P'$, then $|r_P|=|r_{P'}|=r$, for some $r\in\real$.
 \label{fact:cliffCoeff}
\end{fact}

\begin{fact}
If $[A,B]=0$ and $CAC^{\dagger}=A_1$, $CBC^{\dagger}=B_1$, then $[A_1,B_1]=0$.
    \label{fact:commute}
\end{fact}
\begin{proof}
We have the following.
\begin{eqnarray}
    A_1B_1&=&(CAC^{\dagger})(CBC^{\dagger})=C(AB)C^{\dagger}=C(BA)C^{\dagger}   \nonumber \\
    &=&(CBC^{\dagger})(CAC^{\dagger})=B_1A_1    \nonumber
\end{eqnarray}
\end{proof}

%---------------------------------------------------------
\section{Generating set $\gen_{Tof}$}
\label{app:genTof}

\begin{lemma}
If $P_1,P_2,P_3\in\pauli_n\setminus\{\id\}$ pair-wise commutes, then
\begin{enumerate}
    \item $G_{P_1,P_2,P_3}$ remains same for any permutation of the Paulis.

    \item $G_{P_1,-P_1P_2,P_3}=G_{P_1,P_2,P_3}$.

    \item $G_{P_1,-P_2,P_3}=G_{P_1,P_2,P_3}C$, for some $C\in\cliff_n$.

    \item $G_{P_1,P_2,P_1P_2}=\id$.

    \item $G_{P_1,P_2,P_1P_2P_3}=G_{P_1,P_2,P_3}$.
\end{enumerate}
    \label{app:lem:GtofProperties}
\end{lemma}

\begin{proof}
    (1) easily follows from the definition of $G_{P_1,P_2,P_3}$ in Equation \ref{eqn:Gp1p2p3}.

    (2) can be proved as follows.
    \begin{eqnarray}
        G_{P_1,-P_1P_2,P_3}&=&\frac{3}{4}\id+\frac{1}{4}\left(P_1-P_1P_2+P_3+P_1P_1P_2+P_1P_2P_3-P_1P_3-P_1P_1P_2P_3\right)  \nonumber \\
        &=&\frac{3}{4}\id+\frac{1}{4}\left(P_1+P_2+P_3-P_1P_2-P_2P_3-P_3P_1+P_1P_2P_3\right)=G_{P_1,P_2,P_3} \nonumber 
    \end{eqnarray}

    Now we prove (3). We observe that $\tof_{i,j;k}$ can be alternatively written as follows.
    \begin{eqnarray}
        \tof_{i,j;k}&=&\exp\left(\frac{i\pi}{8}(\id-\X_{\q{i}})(\id-\Z_{\q{j}})(\id-\Z_{\q{k}})\right)   \nonumber \\
    \implies CU_{tof} C^{\dagger}&=&\exp\left(\frac{i\pi}{8}(\id-Q)\right)=G_{P_1,P_2,P_3}   \nonumber
    \end{eqnarray}
    and so
    \begin{eqnarray}
        G_{P_1,P_2,P_3}&=&C(\tof_{i,j;k})C^{\dagger}  \nonumber \\
        &=&\exp\left(\frac{i\pi}{8} (\id-P_1-P_2-P_3+P_1P_2+P_2P_3+P_3P_1-P_1P_2P_3)  \right).  \nonumber
    \end{eqnarray}
   And,
   \begin{eqnarray}
       G_{P_1,-P_2,P_3}&=&\exp\left(\frac{i\pi}{8}(\id-P_1+P_2-P_3-P_1P_2-P_2P_3+P_1P_3+P_1P_2P_3)\right)   \nonumber \\
       &=&G_{P_1,P_2,P_3}\exp\left(\frac{i\pi}{8}(2P_2-2P_1P_2-2P_2P_3+2P_1P_2P_3)\right)   \nonumber   \\
       &=&G_{P_1,P_2,P_3}\exp\left(\frac{i\pi}{4}P_2\right)\exp\left(-\frac{i\pi}{4}P_1P_2\right)\exp\left(-\frac{i\pi}{4}P_2P_3\right)\exp\left(\frac{i\pi}{4}P_1P_2P_3\right) \nonumber
   \end{eqnarray}
    where the last line follows since each term commutes with the others. We prove that $\exp\left(\frac{i\pi}{4}P_2\right)$, $\exp\left(-\frac{i\pi}{4}P_1P_2\right)$, $\exp\left(-\frac{i\pi}{4}P_2P_3\right)$ and $\exp\left(\frac{i\pi}{4}P_1P_2P_3\right)$ are Cliffords and this will prove (3). Let $B\in\pm\pauli_n$. Then,
    \begin{eqnarray}
    \exp(\frac{i\pi}{4}P_2)B\exp(\frac{i\pi}{4}P_2)^{\dagger}&=&\frac{\id+iP_2}{\sqrt{2}}B\frac{\id-iP_2}{\sqrt{2}} \nonumber\\
    &=&P_2\left(\frac{P_2B+BP_2}{2}\right)+i\left(\frac{P_2B-BP_2}{2}\right)  \nonumber \\
    &=&B\quad \text{if }P_2B=BP_2;\quad iP_2B\quad\text{if }P_2B=-BP_2    \nonumber
\end{eqnarray}
Since $B$ is a Pauli, so $iP_2B$ is also a Paulis This proves that $\exp\left(\frac{i\pi}{4}P_2\right)\in\cliff_n$. 

Similarly, $\exp\left(-\frac{i\pi}{4}P_1P_2\right), \exp\left(-\frac{i\pi}{4}P_2P_3\right), \exp\left(\frac{i\pi}{4}P_1P_2P_3\right) \in\cliff_n$. This proves (3).

(4) can be proved as follows.
\begin{eqnarray}
    G_{P_1,P_2,P_1P_2}&=&\frac{3}{4}\id+\frac{1}{4}\left(P_1+P_2+P_1P_2-P_1P_2-P_2P_1P_2-P_1P_2P_1+P_1P_2P_1P_2\right)   \nonumber \\
    &=&\frac{3}{4}\id+\frac{1}{4}\left(P_1+P_2-P_1-P_2+\id\right)=\id   \nonumber
\end{eqnarray}

(5) can be proved as follows.
\begin{eqnarray}
        G_{P_1,P_2,P_1P_2P_3}&=&\frac{3}{4}\id+\frac{1}{4}\left(P_1+P_2+P_1P_2P_3-P_1P_2-P_2P_1P_2P_3-P_1P_2P_3P_1+P_1P_2P_1P_2P_3\right)    \nonumber \\
        &=&\frac{3}{4}\id+\frac{1}{4}\left(P_1+P_2+P_3-P_1P_2-P_2P_3-P_3P_1+P_1P_2P_3\right)=G_{P_1,P_2,P_3} \nonumber
\end{eqnarray}

\end{proof}

We know that $\gen_{Tof}$ is the set of unitaries 
\begin{eqnarray}
G_{P_1,P_2,P_3}=\frac{3}{4}\id_n+\frac{1}{4}(P_1+P_2+P_3-P_1P_2-P_2P_3-P_3P_1+P_1P_2P_3), \nonumber
\end{eqnarray} such that if $\pi_3$ is a permutation of 3 elements, then
\begin{enumerate}
    \item $P_1,P_2,P_3\in\pauli_n\setminus\{\id\}$, $P_1\neq P_2\neq P_3$ and $[P_1,P_2]=[P_2,P_3]=[P_3,P_1]=0$.

    \item $(P_1,P_2,P_3)\equiv (P_1,\pm P_1P_2,P_3)\equiv (P_1,P_2,P_1P_2P_3)\equiv\pi_3(P_1,P_2,P_3)$. The equivalence implies that we include only one of them in the set.
\end{enumerate}
We call this a \emph{generating set} (modulo Clifford), because any unitary exactly implementable by the Clifford+Toffoli gate set can be written (up to a global phase) as a product of unitaries from this set and a Clifford.

Before proceeding to prove a bound on $|\gen_{Tof}|$, let us consider a set $\gen$ of pairs of Paulis $(P_1,P_2)$ such that  
\begin{enumerate}
    \item $P_1,P_2\in\pauli_n\setminus\{\id\}$, $P_1\neq P_2$ and $[P_1,P_2]=0$.

    \item $(P_1,P_2)\equiv (P_2,P_1)\equiv (P_1,\pm P_1P_2)$. The equivalence implies that we include only one of them in the set.
\end{enumerate}
This set is relevant for constructing a generating set ($\gen_{CS}$) for the unitaries exactly implementable by the Clifford+CS gate set. In fact, a bound on $|\gen|$ provides a bound on $|\gen_{CS}|$, which we have proved in a separate paper. We are reproducing the proof, for completeness. 

\begin{theorem}

\begin{eqnarray}
    |\gen_{CS}|\leq \frac{1}{8}(16^n-13^n-4^n+1)+\frac{1}{12}(12^n-2\cdot 6^n)  \nonumber
\end{eqnarray}

    \label{app:thm:genCSsize}
\end{theorem}

\begin{proof}
Let $P=\bigotimes_{j=1}^nP_j\in\pauli_n\setminus\{\id_n\}$ such that it is the tensor product of $m$ non-identity single-qubit Paulis. Without loss of generality, let us assume that $P_j\neq\id$ when $1\leq j\leq m$. Let $S_{1Pm}$ is the set of non-identity Paulis that commute with $P$ and have $\id$ on the first $m$ 1-qubit subspaces i.e. $Q=\bigotimes_{j=1}^nQ_j\in S_{1Pm}$ if $Q_j=\id$ for $1\leq j\leq m$. This implies $[P,Q]=0$. When $m+1\leq j\leq n$, then $Q_j\in\pauli_1$ and so there can be 4 possible options. Excluding the condition when $Q_j=\id$ for all $j$, we have $|S_{1Pm}|=4^{n-m}-1$.

Let $S_{2Pm}$ is the set of non-identity Paulis that commute with $P$ and have $\id$ on the last $n-m$ 1-qubit subspaces i.e. $Q'=\bigotimes_{j=1}^nQ_j'\in S_{2Pm}$ if $Q_j'=\id$ for $m+1\leq j\leq n$. Here, we want to enforce the constraint that $(P,P')\equiv (P, \pm PP')$ i.e. $\nexists R\in S_{2Pm}$ such that $Q'=\pm PR$. Now, $[P,Q']=0$ if and only if $k=\left| \{j : P_j\neq Q_j'; 1\leq j\leq m  \} \right|$ is even. In the remaining $m-k$ subspaces $Q_j'=P_j$ or $\id$. There can be at most $2^{m-k}$ combinations. In the $k$ unequal (but neither is identity) subspaces there can be at most $2^k$ combinations, but out of these half can be obtained by myltiplying $P$ with the other half. For example, consider $P=\X\X$, then $Q'\in\{\Y\Y,\Y\Z\}$. Though $\{\Z\Z,\Z\Y\}$ also commute with $P$ but these can be obtained by multiplying $P$ with $\Y\Y$ and $\Y\Z$. Thus, there can be $2^{k-1}\cdot 2^{m-k}=2^{m-1}$ possibilities for $Q'$. Now we can select $k$ subspaces in $\binom{m}{k}$ ways. When $m<n$ then $k$ can vary from $0$ to $m$ or $m-1$, whichever is even. When $m=n$ then $k$ varies from $2$ to $m$ or $m-1$, because we want to avoid the all identity case. Hence,
\begin{eqnarray}
    |S_{2Pm}|&\leq& 2^{m-1}\sum_{k=0}^{m'}\binom{m}{k} := 2^{m-1}\cdot A \quad [m<n, m'=m\text{ or }m-1]  \nonumber \\
    S_{2Pm}|&\leq& 2^{m-1}\sum_{k=2}^{m'}\binom{m}{k}=2^{m-1}\left(\sum_{\ell=0}^k\binom{m}{\ell}-1\right) := 2^{m-1}\cdot (A-1) \quad [m=n, m'=m\text{ or }m-1]  \nonumber 
\end{eqnarray}
When $m$ is even then,
\begin{eqnarray}
    A&=&\sum_{k=0}^m\binom{m}{k}=\frac{1}{2}\left(\sum_{i=0}^m\binom{m}{i}+\sum_{i=0}^m(-1)^i\binom{m}{i} \right)=\frac{1}{2}((1+1)^m+(1-1)^m)=2^{m-1}; \nonumber
\end{eqnarray}
while when $m-1$ is even then,
\begin{eqnarray}
    A&=&\sum_{k=0}^{m-1}\binom{m}{k}=\frac{1}{2}\left(\sum_{i=0}^m\binom{m}{i}+\sum_{i=0}^m(-1)^i\binom{m}{i}\right)=2^{m-1}.   \nonumber    
\end{eqnarray}
So,
\begin{eqnarray}
    |S_{2Pm}|&\leq& 2^{m-1}2^{m-1}=2^{2m-2} \qquad [1\leq m< n]    \\
   |S_{2Pm}| &\leq& 2^{m-1}\left(2^{m-1}-1\right)=2^{2m-2}-2^{m-1}   \qquad [m=n]   \nonumber
\end{eqnarray}
Let $d_{Pm}$ is the number of Paulis that commute with $P$ and we maintain the equivalence $(P,Q)\equiv (P,\pm PQ)$ i.e. we count either one of $P$ and $PQ$. Then,
\begin{eqnarray}
    d_P&\leq& |S_{1Pm}|\leq 4^{n-1}-1 \qquad [m=1]    \nonumber \\
    &\leq& |S_{2Pm}|\leq  2^{2n-2}-2^{n-1} \qquad [m=n]    \nonumber \\
    &\leq& |S_{1Pm}|\cdot |S_{2Pm}|\leq \left(2^{2n-2m}-1\right)2^{2m-2}=2^{2n-2}-2^{2m-2}=\frac{1}{4}\left(4^n-4^m\right)\qquad [1<m<n] \nonumber
\end{eqnarray}
We count the number of Paulis $P'=\bigotimes_{j=1}^nP_j'$, which are tensor product of $m$ non-identity Paulis. We denote it by $d_m$.
There can be $\binom{n}{m}3^m$ of them. When $m=n$, among them there are few which are products of the other. Specifically we can take 2 possible values for $P_1'$ and 3 possible values for the remaining $P_j'$, where $1<j\leq n$. The remaining $n$-qubit Paulis that have $m$ non-identity tensored components, can be obtained as product of these. And we do not count these, due to the equivalence constraint. When $m=1$ there can be 3 possibilities for the only non-identity component. So number of possible Paulis with $m$ non-identity tensored components is at most $3\cdot \binom{n}{1}$ if $m=1$; $\binom{n}{n}2\cdot 3^{n-1}$ if $m=n$ and $\binom{n}{m}3^m$ if $1<m<n$.  Hence, 
\begin{eqnarray}
   d_m &\leq&\binom{n}{1}3(4^{n-1}-1)=3n(4^{n-1}-1) \qquad \text{when } m=1;    \nonumber \\
   d_m &\leq&\binom{n}{m} \frac{3^m}{4}\left(2^{2n}-2^{2m}\right)=\binom{n}{m}\left(4^{n-1}\cdot 3^m-\frac{1}{4}\cdot 12^m\right) \qquad \text{when } n>m\geq 2;  \nonumber \\
    d_m &\leq&\binom{n}{n}2\cdot 3^{n-1} \left(2^{2n-2}-2^{n-1}\right)=2\left(12^{n-1}-6^{n-1}\right) \qquad \text{when } m=n.  \nonumber
\end{eqnarray}

Total number of Pauli pairs is,
\begin{eqnarray}
    \sum_md_m&\leq&3n(4^{n-1}-1)+2(12^{n-1}-6^{n-1})+\sum_{m=2}^{n-1}\binom{n}{m}\left(4^{n-1}\cdot 3^m-\frac{1}{4}\cdot 12^m\right)  \nonumber \\
    &:=&\frac{3n}{4}4^n-3n+\frac{12^n}{6}-\frac{6^n}{3}+B;  \nonumber 
\end{eqnarray}
where
\begin{eqnarray}
    B&=&\frac{4^n}{4}\left[\sum_{m=0}^n\binom{n}{m}3^m-\binom{n}{0}-\binom{n}{1}3-\binom{n}{n}3^n\right]-\frac{1}{4}\left[\sum_{m=0}^n\binom{n}{m}12^m-\binom{n}{0}-\binom{n}{1}12-\binom{n}{n}12^n\right]  \nonumber \\
    &=&\frac{4^n}{4}\left[4^n-1-3n-3^n\right]-\frac{1}{4}\left[13^n-1-12n-12^n\right]=\frac{16^n}{4}-\frac{13^n}{4}-4^n\frac{3n+1}{4}+3n+\frac{1}{4};    \nonumber
\end{eqnarray}
and so,
\begin{eqnarray}
    \sum_md_m&\leq&\frac{16^n}{4}-\frac{13^n}{4}+\frac{12^n}{6}-\frac{6^n}{3}-\frac{4^n}{4}+\frac{1}{4}. \nonumber
\end{eqnarray}
Since the number of pairs are permutation invariant i.e. $(P_1,P_2)\equiv (P_2,P_1)$, so,
\begin{eqnarray}
    |\mathcal{G}_{CS}|=\frac{1}{2}\sum_md_m\leq \frac{16^n}{8}-\frac{13^n}{8}+\frac{12^n}{12}-\frac{6^n}{6}-\frac{4^n}{8}+\frac{1}{8}  .
\end{eqnarray}

This proves the theorem.

\end{proof}

Now we prove a bound on $|\gen_{Tof}|$.
\begin{theorem}

\begin{eqnarray}
    |\gen_{Tof}|&\leq& \frac{1}{384}[64^n-52^n-17(16^n-4^n)]+\frac{1}{576}[48^n-2\cdot 24^n]-\frac{1}{36}[12^n-2\cdot 6^n]+\frac{1}{24}(13^n-1)    \nonumber \\
    &\in& O\left( n^264^{n-2} \right)    \nonumber
\end{eqnarray}

    \label{app:thm:genTOFsize}
\end{theorem}

\begin{proof}
 First we count the number of commuting Pauli pairs $(P,Q)$ such that $(P,Q)\equiv (P,\pm PQ)$. Let $\gen$ denote the set of all such pairs. From Theorem \ref{app:thm:genCSsize} we know that
\begin{eqnarray}
    |\gen|\leq \frac{1}{8}\left(16^n-13^n-4^n+1\right)+\frac{1}{12}\left(12^n-2\cdot 6^n\right). \nonumber
\end{eqnarray}
Now, let us fix a particular $(P,Q)$ and count the number of commuting Paulis $R$ that satisfy the constraints in Lemma \ref{app:lem:GtofProperties}. Now number of Paulis that commute with both $P$ and $Q$ but is not equal to $\id, P,Q, PQ$ is $\frac{4^n}{4}-4=4^{n-1}-4$. The constraints imply that $(P,Q,R)\equiv (P,Q,PR)\equiv (P,Q,QR)\equiv (P,Q,PQR)$. Let $R'$ is a Pauli that commutes with both $P,Q$. If $R'=PQ$ then $PR'=PPQ=Q$, $QR'=PQR$, $PQR'=PQPR=QR$. Thus we get a cycle. Similary, for $R'=QR, PQR$. For each commuting Pauli $R$ we get such a cycle of length 4. Thus if $S_{PQ}$ is the number of commuting Paulis $R$ that satisfy the constraints, then
\begin{eqnarray}
    |S_{PQ}|\leq \frac{4^{n-1}-1}{4}=4^{n-2}-1. \nonumber
\end{eqnarray}
So total number of triples which we can have is
\begin{eqnarray}
    |\gen|\cdot |S_{PQ}|&\leq& (4^{n-2}-1)\left[\frac{1}{8}\left(16^n-13^n-4^n+1\right)+\frac{1}{12}\left(12^n-2\cdot 6^n\right)\right]    \nonumber \\
    &=&\frac{64^n}{128}-\frac{52^n}{128}+\frac{48^n}{192}-\frac{24^n}{96}-\frac{16^n}{128}+\frac{4^n}{128}-\frac{16^n}{8}+\frac{13^n}{8}-\frac{12^n}{12}+\frac{6^n}{6}+\frac{4^n}{8}-\frac{1}{8}    \nonumber \\
    &=&\frac{1}{128}[64^n-52^n-17(16^n-4^n)]+\frac{1}{192}[48^n-2\cdot 24^n]-\frac{1}{12}[12^n-2\cdot 6^n]+\frac{1}{8}(13^n-1) \nonumber
\end{eqnarray}
Now there are $\binom{3}{2}=3$ ways in which we can first place the pair $(P,Q)$ and then the position of $R$ is uniquely determined. Thus,
\begin{eqnarray}
    |\gen_{tof}|&\leq& \frac{1}{3}\left( |\gen|\cdot |S_{PQ}| \right)   \nonumber \\
    &\leq&\frac{1}{384}[64^n-52^n-17(16^n-4^n)]+\frac{1}{576}[48^n-2\cdot 24^n]-\frac{1}{36}[12^n-2\cdot 6^n]+\frac{1}{24}(13^n-1).    \nonumber
\end{eqnarray}

We observe the fact that for any $a,b \in\real$ such that $a\geq b>1$ and $m \in \intg_{+}$ we have
\begin{eqnarray}
    a^m - b^m = (a-b)(a^{m-1}+a^{m-2}b + a^{m-3}b^2+\ldots+b^{m-1}  ),   \nonumber 
\end{eqnarray}
and so
\begin{eqnarray}
    (a-b)mb^{m-1} \leq a^m-b^m \leq (a-b)ma^{m-1}.  \nonumber
\end{eqnarray}
Thus,
\begin{eqnarray}
    |\gen_{Tof}|&\leq&\frac{1}{384}\left[(64-52)n64^{n-1}-17(16-4)n4^{n-1}  \right]+\frac{2(48-24)}{576}n48^{n-1}+\frac{(13-1)}{24}n13^{n-1} \nonumber \\
    &\leq&\frac{12n}{384}\left[64^{n-1}-4^{n-1}\right]+\frac{n48^{n-1}}{12}+\frac{n13^{n-1}}{2} \nonumber \\
    &\leq&\frac{n}{32}(n-1)(64-4)64^{n-2}+\frac{n48^{n-2}48}{12}+\frac{n13^{n-2}13}{2} \nonumber \\
    &=&\frac{15n(n-1)}{8}64^{n-2}+4n48^{n-2}+\frac{13n}{2}13^{n-2} \in O\left(n^264^{n-2}  \right),  \nonumber
\end{eqnarray}
thus proving the theorem.
\end{proof}

%----------------------------------------------------------------------------
\section{Bound on Toffoli-count of arbitrary unitaries}
\label{app:sec:lowBound}

In this section we derive some results required to prove a lower bound on the Toffoli-count of arbitrary $n$-qubit unitaries, as discussed in Section \ref{subsec:bound}. 

\begin{lemma}
If $\widetilde{U}= \prod_{j=1}^mG_{P_{1_j},P_{2_j},P_{3_j}}$, where $G_{P_{1_j},P_{2_j},P_{3_j}} = a\id+bQ_j $, $Q_j = P_{1_j}+P_{2_j}+P_{3_j}-P_{1_j}P_{2_j}-P_{2_j}P_{3_j}-P_{3_j}P_{1_j}+P_{1_j}P_{2_j}P_{3_j}$, $a = \frac{3}{4}$ and $b = \frac{1}{4}$, then
\begin{eqnarray}
    \widetilde{U} &=& a^m\id+a^{m-1}b\left(\sum_{j}Q_j\right)+a^{m-2}b^2\left(\sum_{j_1 < j_2} Q_{j_1}Q_{j_2} \right)+a^{m-3}b^3\left(\sum_{j_1< j_2 < j_3}Q_{j_1}Q_{j_2}Q_{j_3}  \right)+\cdots    \nonumber \\
    &&\cdots+ab^{m-1} \left(\sum_{j_1< j_2< \cdots < j_{m-1}} Q_{j_1}Q_{j_2}\cdots Q_{j_{m-1}}  \right) + b^m \prod_{j=1}^mQ_j  \nonumber
\end{eqnarray}
\label{app:lem:widetildeU}
\end{lemma}

\begin{proof}
    We prove this by induction. The statement holds trivially for $m=1$. For $m=2$ we have
    \begin{eqnarray}
        \prod_{j=1}^2G_{P_{1_j},P_{2_j},P_{3_j}}=(a\id+bQ_1)(a\id+bQ_2) = a^2\id+ab(Q_1+Q_2)+b^2Q_1Q_2. \nonumber
    \end{eqnarray}
We assume that the statement holds for $m-1$, that is
\begin{eqnarray}
    \prod_{j=1}^{m-1}G_{P_{1_j},P_{2_j},P_{3_j}} &=& a^{m-1}\id+a^{m-2}b\left(\sum_{j=1}^{m-1}Q_j\right)+a^{m-3}b^2\left(\sum_{\substack{j_1, j_2 =1 \\ j_1<j_2}}^{m-1} Q_{j_1}Q_{j_2} \right)+\cdots\nonumber \\
    &&\cdots+ab^{m-2} \left(\sum_{\substack{j_1,j_2,\ldots,j_{m-2}=1 \\ j_1< j_2< \cdots < j_{m-2}}}^{m-1} Q_{j_1}Q_{j_2}\cdots Q_{j_{m-2}}  \right) + b^{m-1} \prod_{j=1}^{m-1}Q_j.  \nonumber
\end{eqnarray}
Therefore,
\begin{eqnarray}
   \widetilde{U} &=& \left(  \prod_{j=1}^{m-1}G_{P_{1_j},P_{2_j},P_{3_j}}\right)(a\id+bQ_m) \nonumber \\
  &=& \left(a^{m-1}\id+a^{m-2}b\left(\sum_{j=1}^{m-1}Q_j\right)+a^{m-3}b^2\left(\sum_{\substack{j_1, j_2 =1 \\ j_1<j_2}}^{m-1} Q_{j_1}Q_{j_2} \right)+\cdots \right. \nonumber \\
    &&\left.\cdots+ab^{m-2} \left(\sum_{\substack{j_1,j_2,\ldots,j_{m-2}=1 \\ j_1< j_2< \cdots < j_{m-2}}}^{m-1} Q_{j_1}Q_{j_2}\cdots Q_{j_{m-2}}  \right) + b^{m-1} \prod_{j=1}^{m-1}Q_j\right)(a\id+bQ_m)  \nonumber \\
    &=&a^m\id+a^{m-1}b\left(\sum_{j=1}^{m-1}Q_j\right)+a^{m-2}b^2\left(\sum_{\substack{j_1,j_2=1 \\ j_1<j_2}}^{m-1}Q_{j_1}Q_{j_2} \right)+\cdots+a^2b^{m-2} \left(\sum_{\substack{j_1,j_2,\ldots,j_{m-2}=1 \\ j_1< j_2< \cdots < j_{m-2}}}^{m-1} Q_{j_1}Q_{j_2}\cdots Q_{j_{m-2}}  \right)  \nonumber \\
    &&+ab^{m-1}\prod_{j=1}^{m-1}Q_j  +a^{m-1}bQ_m+ a^{m-2}b^2 \left(\sum_{j=1}^{m-1}Q_j \right)Q_m +a^{m-3}b^3 \left(\sum{\substack{j_1,j_2=1 \\ j_1<j_2}}^{m-1}Q_{j_1}Q_{j_2}  \right)Q_m+\cdots\nonumber \\
    &&\cdots+ab^{m-1}\left(\sum_{\substack{j_1,j_2,\ldots,j_{m-2}=1 \\ j_1< j_2< \cdots < j_{m-2}}}^{m-1} Q_{j_1}Q_{j_2}\cdots Q_{j_{m-2}}  \right)Q_m +b^m \left(\prod_{j=1}^{m-1}Q_j\right)Q_m    \nonumber \\
    &=&a^m\id+a^{m-1}b\left(\sum_{j}Q_j\right)+a^{m-2}b^2\left(\sum_{j_1 < j_2} Q_{j_1}Q_{j_2} \right)+a^{m-3}b^3\left(\sum_{j_1< j_2 < j_3}Q_{j_1}Q_{j_2}Q_{j_3}  \right)+\cdots    \nonumber \\
    &&\cdots+ab^{m-1} \left(\sum_{j_1< j_2< \cdots < j_{m-1}} Q_{j_1}Q_{j_2}\cdots Q_{j_{m-1}}  \right) + b^m \prod_{j=1}^mQ_j , \nonumber
\end{eqnarray}
    thus proving the lemma.
\end{proof}

%---------------------------------------------------------------------
\subsection{Pauli basis expansion of rotation unitaries}
\label{app:subsec:pauliBasis}

Now we derive the Pauli basis expansions of the rotation unitaries considered in Section \ref{subsubsec:boundRot}. For convenience, especially when working with more than 2-qubit unitaries, we use $\diag(d_1,d_2,\ldots,d_N)$ to denote an $N\times N$ diagonal matrix with entries $d_1, d_2, \ldots, d_N$ along the diagonal and 0 in the remaining places.

First we consider the single qubit z-rotation $R_z$ gate. 
\begin{eqnarray}
    R_z(\theta) &=& e^{-i\frac{\theta}{2}} \begin{bmatrix}
    1 & 0 \\
    0 & e^{i\theta}
    \end{bmatrix} = e^{-i\frac{\theta}{2}}\frac{1}{2}\begin{bmatrix} (1+e^{i\theta})+(1-e^{i\theta}) & 0 \\ 0 & (1+e^{i\theta})-(1-e^{i\theta})\end{bmatrix} \nonumber \\
    &=&e^{-i\frac{\theta}{2}}\frac{1}{2}\left((1+e^{i\theta})\begin{bmatrix} 1 & 0 \\ 0 & 1  \end{bmatrix}+(1-e^{i\theta})\begin{bmatrix} 1 & 0 \\ 0 & -1 \end{bmatrix}  \right)    \nonumber \\
    &=&e^{-i\frac{\theta}{2}} \left( \frac{1+e^{i\theta}}{2}\id + \frac{1-e^{i\theta}}{2}\Z \right) 
    \label{eqn:app:Rz}
\end{eqnarray}
$R_z(\theta) = e^{-i\frac{\theta}{2}}R_n(\theta)$ i.e. these two rotation gates are equivalent upto a global phase, but not their controlled versions. Now we consider $R_n(\theta)$, controlled on 1 qubit. With similar derivation as in Equation \ref{eqn:app:Rz}, we first express
\begin{eqnarray}
    cR_n(\theta) &=& \begin{bmatrix}
    1 & 0 & 0 & 0 \\
    0 & 1 & 0 & 0 \\
    0 & 0 & 1 & 0 \\
    0 & 0 & 0 & e^{i\theta}
    \end{bmatrix} = \frac{1+e^{i\theta}}{2} (\id\otimes\id) + \frac{1-e^{i\theta}}{2} \cz,  \label{eqn:app:cRn0}
\end{eqnarray}
where the $CZ$ gate can be expressed,
\begin{eqnarray}
    CZ &=& \begin{bmatrix}
    1 & 0 & 0 & 0 \\
    0 & 1 & 0 & 0 \\
    0 & 0 & 1 & 0 \\
    0 & 0 & 0 & -1
    \end{bmatrix} = \frac{1}{2}\left(\id\otimes\id + \Z\otimes\id + \id\otimes\Z - \Z\otimes\Z \right)   \label{eqn:app:CZ}
\end{eqnarray}
and thus
\begin{eqnarray}
    cR_n(\theta) &=& \frac{1+e^{i\theta}}{2} (\id\otimes\id) + \frac{1-e^{i\theta}}{2} \left( \frac{1}{2}\left(\id\otimes\id + \Z\otimes\id + \id\otimes\Z - \Z\otimes\Z \right) \right)  \nonumber \\
    &=& \frac{3+e^{i\theta}}{4} (\id\otimes\id) + \frac{1-e^{i\theta}}{4} \left( (\Z\otimes\id) + (\id\otimes\Z) - (\Z\otimes\Z)  \right).
    \label{eqn:app:cRn}
\end{eqnarray}
Next, we consider $cR_z(\theta)$ unitary i.e. z-rotation controlled on 1 qubit.
\begin{eqnarray}
cR_z(\theta) &=& \begin{bmatrix}
    1 & 0 & 0 & 0 \\
    0 & 1 & 0 & 0 \\
     0 & 0 & e^{-i\theta/2} & 0 \\
    0 & 0 & 0 & e^{i\theta/2}
    \end{bmatrix} \nonumber \\
    &=& \diag\left(1,1, \cos(\theta/2), \cos(\theta/2) \right)+\diag\left(0,0,-i\sin(\theta/2),i\sin(\theta/2)  \right)   \nonumber \\
    &=&\frac{1}{2}\left( \left(1+\cos\frac{\theta}{2}\right)\diag(1,1,1,1)+\left(1-\cos\frac{\theta}{2}\right)\diag(1,1,-1,-1)  \right)    \nonumber \\
    &&+\frac{i}{2}\left(\left(\sin\frac{\theta}{2}\right)\diag(1,-1,-1,1)-\left(\sin\frac{\theta}{2}\right)\diag(1,-1,1,-1)  \right)    \nonumber \\
    &=&\left(\frac{1+\cos\frac{\theta}{2}}{2} \right)(\id\otimes\id)+\left(\frac{1+\cos\frac{\theta}{2}}{2} \right)(\Z\otimes\id)   \nonumber \\
    &&-i\left(\frac{\sin\frac{\theta}{2}}{2}\right)(\id\otimes\Z)+i\left(\frac{\sin\frac{\theta}{2}}{2}  \right)(\Z\otimes\Z)
    \label{eqn:app:cRz}
\end{eqnarray}
Now, we consider the 2-qubit Given's rotation.
\begin{eqnarray}
    Givens(\theta) &=& \begin{bmatrix}
    1 & 0 & 0 & 0 \\
    0 & \cos(\theta) & -\sin(\theta) & 0 \\
    0 & \sin(\theta) & \cos(\theta) & 0 \\
    0 & 0 & 0 & 1
    \end{bmatrix}   \nonumber \\
    &=&\begin{bmatrix}
    1 & 0 & 0 & 0 \\
    0 & \cos(\theta) & 0 & 0 \\
    0 & 0 & \cos(\theta) & 0 \\
    0 & 0 & 0 & 1
    \end{bmatrix}+ \begin{bmatrix}
    0 & 0 & 0 & 0 \\
    0 & 0 & -\sin(\theta) & 0 \\
    0 & \sin(\theta) & 0 & 0 \\
    0 & 0 & 0 & 0
    \end{bmatrix}:=A+B  \nonumber
\end{eqnarray}
Using arguments as before, we can express $A$ as follows.
\begin{eqnarray}
    A &=& \diag(1,\cos\theta,\cos\theta,1) = \frac{1}{2}\left((1+\cos\theta)\diag(1,1,1,1)+(1-\cos\theta)\diag(1,-1,-1,1)  \right)    \nonumber \\
    &=&\left( \frac{1+\cos\theta}{2} \right)(\id\otimes\id)+\left(\frac{1-\cos\theta}{2}\right)(\Z\otimes\Z)    \nonumber
\end{eqnarray}
We expand $B$ as follows.
\begin{eqnarray}
    B&=&\frac{1}{2i}\left(\sin\theta\begin{bmatrix} 0 & 0 & 0 & i \\ 0 & 0 & -i & 0 \\ 0 & i & 0 & 0 \\ -i & 0 & 0 & 0  \end{bmatrix}+\sin\theta\begin{bmatrix} 0 & 0 & 0 & -i \\ 0 & 0 & -i & 0 \\ 0 & i & 0 & 0 \\ i & 0 & 0 & 0  \end{bmatrix}   \right) \nonumber \\
    &=&-i\frac{\sin\theta}{2}\left(-(\X\otimes\Y)+(\Y\otimes\X)\right)
\end{eqnarray}
So,
\begin{eqnarray}
    Givens(\theta)= \left( \frac{1+\cos\theta}{2} \right)(\id\otimes\id)+i\left(\frac{\sin\theta}{2}\right)(\X\otimes\Y)-i\left(\frac{\sin\theta}{2}\right)(\Y\otimes\X)+\left(\frac{1-\cos\theta}{2}\right)(\Z\otimes\Z)   .\nonumber
\end{eqnarray}

Now we consider the following rotation gates, controlled on 2 qubits. Similar to Equation \ref{eqn:app:cRn0},
\begin{eqnarray}
    ccR_n(\theta) &=& \diag\left(1,1,1,1,1,1,1,e^{i\theta}  \right) = \left(\frac{1+e^{i\theta}}{2}\right)(\id\otimes\id\otimes\id)+\left(\frac{1+e^{i\theta}}{2}\right)CCZ, \nonumber
\end{eqnarray}
where the double-controlled Z unitary can be expanded as,
\begin{eqnarray}
    CCZ&=&\diag\left(1,1,1,1,1,1,1,-1  \right)   
    =\frac{3}{4}(\id\otimes\id\otimes\id)+\frac{1}{4}\left((\id\otimes\id\otimes\Z)+(\id\otimes\Z\otimes\id)+(\Z\otimes\id\otimes\id)  \right)  \nonumber \\
    &&-\frac{1}{4}\left((\id\otimes\Z\otimes\Z)+(\Z\otimes\Z\otimes\id)+(\Z\otimes\id\otimes\Z)  \right)+\frac{1}{4}(\Z\otimes\Z\otimes\Z), \nonumber
\end{eqnarray}
and therefore,
\begin{eqnarray}
    ccR_n(\theta)&=&\left(\frac{7+e^{i\theta}}{8}\right)(\id\otimes\id\otimes\id)+\left(\frac{1-e^{i\theta}}{8}\right)\left((\id\otimes\id\otimes\Z)+(\id\otimes\Z\otimes\id)+(\Z\otimes\id\otimes\id)  \right) \nonumber \\
    &&-\left(\frac{1-e^{i\theta}}{8}\right)\left((\id\otimes\Z\otimes\Z)+(\Z\otimes\Z\otimes\id)+(\Z\otimes\id\otimes\Z)  \right)+\left(\frac{1-e^{i\theta}}{8}\right)(\Z\otimes\Z\otimes\Z).   \nonumber
    \label{eqn:app:ccRn}
\end{eqnarray}
And finally,
\begin{eqnarray}
    ccR_z(\theta) &=&\diag\left(1,1,1,1,1,1,e^{-i\theta/2},e^{i\theta/2}  \right)   \nonumber \\
    &=&\diag\left(1,1,1,1,1,1,\cos\frac{\theta}{2},\cos\frac{\theta}{2}  \right)+i\diag\left(0,0,0,0,0,0,-\sin\frac{\theta}{2},\sin\frac{\theta}{2}  \right)        \nonumber \\
    &:=&A+B .  \nonumber
\end{eqnarray}
We can expand $A$ as,
\begin{eqnarray}
    A&=&\left(\frac{1+\cos\frac{\theta}{2}}{2}\right)(\id\otimes\id\otimes\id)+\left(\frac{1-\cos\frac{\theta}{2}}{2}\right)\diag(1,1,1,1,1,1,-1,-1)   \nonumber \\
    &=&\left(\frac{1+\cos\frac{\theta}{2}}{2}\right)(\id\otimes\id\otimes\id)+\left(\frac{1-\cos\frac{\theta}{2}}{2}\right)(\cz\otimes\id), 
\end{eqnarray}
and since from Equation \ref{eqn:app:CZ},
\begin{eqnarray}
    \cz\otimes\id&=&\frac{1}{2}\left(\id\otimes\id\otimes\id + \Z\otimes\id\otimes\id + \id\otimes\Z\otimes\id -\Z\otimes\Z\otimes\id \right),   \nonumber
\end{eqnarray}
so we have
\begin{eqnarray}
    A&=&\left( \frac{3+\cos\frac{\theta}{2} }{4} \right)(\id\otimes\id\otimes\id)+\left(\frac{1-\cos\frac{\theta}{2}}{4}\right)(\Z\otimes\id\otimes\id + \id\otimes\Z\otimes\id -\Z\otimes\Z\otimes\id). \nonumber
\end{eqnarray}
We can expand $B$ as follows.
\begin{eqnarray}
    B&=&i\frac{\sin\frac{\theta}{2}}{4}\left(-\diag(1,-1,1,-1,1,-1,1,-1)+\diag(1,-1,-1,1,1,-1,-1,1) \right. \nonumber \\
    &&\left.+\diag(1,-1,1,-1,-1,1,-1,1)-\diag(1,-1,-1,1,-1,1,1,-1)  \right) \nonumber \\
    &=&i\frac{\sin\frac{\theta}{2}}{4}\left(-\id\otimes\id\otimes\Z+\id\otimes\Z\otimes\Z+\Z\otimes\id\otimes\Z-\Z\otimes\Z\otimes\Z  \right)
\end{eqnarray}
Therefore, we obtain the following.
\begin{eqnarray}
    ccR_z(\theta)&=&\left( \frac{3+\cos\frac{\theta}{2} }{4} \right)(\id\otimes\id\otimes\id)+\left(\frac{1-\cos\frac{\theta}{2}}{4}\right)(\Z\otimes\id\otimes\id + \id\otimes\Z\otimes\id -\Z\otimes\Z\otimes\id) \nonumber \\
    &&+i\frac{\sin\frac{\theta}{2}}{4}\left(-\id\otimes\id\otimes\Z+\id\otimes\Z\otimes\Z+\Z\otimes\id\otimes\Z-\Z\otimes\Z\otimes\Z  \right)   
    \label{eqn:app:ccRz}
\end{eqnarray}

%---------------------------------------------------------------------
\section{Channel representation of $\chan{G_{P_1,P_2,P_3}}$}
\label{app:chanRep}

We recall that 
\begin{eqnarray}
G_{P_1,P_2,P_3}=\frac{3}{4}\id+\frac{1}{4}Q,
\label{eqn:app:Gp1p2p3}
\end{eqnarray}
where $Q=P_1+P_2+P_3-P_1P_2-P_2P_3-P_3P_1+P_1P_2P_3$.

\begin{lemma}
    We have $\tr(Q)=0$. Also, $Q^2=7\id-6Q$, implying $\tr(Q^2)=7\cdot 2^n$.
    \label{lem:trQ_Q2}
\end{lemma}

\begin{proof}
    $Q$ is sum of non-identity Paulis and so $\tr(Q)=0$.

    \begin{eqnarray}
        Q^2&=&P_1^2+P_2^2+P_3^2+(P_1P_2)^2+(P_2P_3)^2+(P_3P_1)^2+(P_1P_2P_3)^2  \nonumber \\
        &&+2P_1(P_2+P_3-P_1P_2-P_2P_3-P_3P_1+P_1P_2P_3)+2P_2(P_3-P_1P_2-P_2P_3-P_3P_1+P_1P_2P_3)    \nonumber \\
        &&+2P_3(-P_1P_2-P_2P_3-P_3P_1+P_1P_2P_3)-2P_1P_2(-P_2P_3-P_3P_1+P_1P_2P_3)  \nonumber \\
        &&-2P_2P_3(-P_3P_1+P_1P_2P_3)-2P_3P_1(P_1P_2P_3)    \nonumber \\
        &=&7\id+2(P_1P_2+P_1P_3-P_2-P_1P_2P_3-P_3+P_2P_3)+2(P_2P_3-P_1-P_3-P_1P_2P_3+P_3P_1)    \nonumber \\
        &&+2(-P_1P_2P_3-P_2-P_1+P_1P_2)+2(P_1P_3+P_2P_3-P_3)+2(P_1P_2-P_1)-2P_2 \nonumber \\
        &=&7\id-6Q  \nonumber 
    \end{eqnarray}
   Thus $\tr[Q^2]=7\cdot 2^n$. 
\end{proof}

\begin{lemma}
   \begin{eqnarray}
       \tr[(P_rQ)^2]&=&\sum_{j=1}^3\tr[(P_rP_j)^2]+\sum_{j<k}\tr[(P_rP_jP_k)^2]+\tr[(P_rP_1P_2P_3)^2]   \nonumber
   \end{eqnarray} 
   \label{lem:tr_PrQ2}
\end{lemma}

\begin{proof}
    Let $Q=\sum_{j=1}^7Q_j$, where $Q_j$ are the Paulis appearing in the expression of $Q$. Then $P_rQ=\sum_{j=1}^7P_rQ_j$. Thus,
    \begin{eqnarray}
        (P_rQ)^2&=&\sum_{j=1}^7(P_rQ_j)^2+\sum_{j\neq k}(P_rQ_j)(P_rQ_k)    \nonumber \\
        \tr[(P_rQ)^2]&=&\sum_{j=1}^7\tr[(P_rQ_j)^2]+\sum_{j\neq k}\tr[(P_rQ_j)(P_rQ_k)] \nonumber \\
        &=&\sum_{j=1}^7\tr[(P_rQ_j)^2], \nonumber
    \end{eqnarray}
    proving the lemma.
\end{proof}

\begin{theorem}
Let $G_{P_1,P_2,P_3}\in\gen_{Tof}$, where $P_1,P_2,P_3\in\pauli_n\setminus\{\id\}$. Then its channel representation $\chan{G_{P_1,P_2,P_3}}$ has the following properties.
\begin{enumerate}
    \item The diagonal entries are $1$ or $\frac{1}{2}$.

    \item If a diagonal entry is 1 then all other entries in the corresponding row and column is 0.

    \item If a diagonal entry is $\frac{1}{2}$ then there exists three entries in the corresponding row and column that are equal to $\pm\frac{1}{2}$, rest is 0.

    \item Exactly $2^{2n-3}$ i.e. $\frac{1}{8}^{th}$ of the diagonal elements are $1$, while the remaining are $\frac{1}{2}$.
\end{enumerate}
    \label{app:thm:chanRep}
\end{theorem}

\begin{proof}
    We express $G_{P_,P_2,P_3}$ as in Equation \ref{eqn:app:Gp1p2p3}. $\chan{G_{P_1,P_2,P_3}}$ is a $2^{2n}\times 2^{2n}$ matrix, whose rows and columns are labeled by $n$-qubit Paulis. Let $P_r,P_s\in\pauli_n$. Using cyclic property of trace, we have the following.
    \begin{eqnarray}
    \chan{G_{P_1,P_2,P_3}}[P_r,P_s]&=&\frac{1}{2^n}\tr\left[P_r\left(\frac{3\id+Q}{4}\right)P_s\left(\frac{3\id+Q}{4}\right)\right]   \nonumber \\
    &=&\frac{1}{16\cdot 2^n}\tr\left[9P_rP_s+3P_rP_sQ+3P_rQP_s+P_rQP_sQ\right]  \nonumber \\
    &=&\frac{9}{16\cdot 2^n}\tr(P_rP_s)+\frac{3}{16\cdot 2^n}\tr(P_rP_sQ)+\frac{3}{16\cdot 2^n}\tr(P_sP_rQ)+\frac{1}{16\cdot 2^n}\tr(P_rQP_sQ)  \nonumber \\
    &:=&a+b+c+d \label{eqn:abcd}
    \end{eqnarray}

\paragraph{First row and column : } We look at the first entry i.e. $P_r=P_s=\id$. Using Lemma \ref{lem:trQ_Q2} we have
\begin{eqnarray}
    \chan{G_{P_1,P_2,P_3}}[\id,\id]&=&\frac{1}{16\cdot 2^n}[9\tr(\id)+6\tr(Q)+\tr(Q^2)]=1.   \nonumber
\end{eqnarray}
Let us consider any other entry of the first row and column. Let $P_r=\id$ and $P_s\neq\id$. 
\begin{eqnarray}
    \chan{G_{P_1,P_2,P_3}}[\id,P_s]&=&\frac{1}{16\cdot 2^n}[9\tr(P_s)+6\tr(P_SQ)+\tr(Q^2P_s)] \nonumber \\
    &=&\frac{1}{16\cdot 2^n}[9\tr(P_s)+6\tr(P_SQ)+\tr((7\id - 6Q)P_s)] \qquad [\text{Lemma \ref{lem:trQ_Q2}}]\nonumber \\
    &=&\frac{1}{16\cdot 2^n}[6\tr(P_sQ)+7\tr(P_s)-6\tr(P_sQ)]=0 \qquad [\text{Cyclic property of trace}]    \nonumber
\end{eqnarray}
Similarly $\chan{G_{P_1,P_2,P_3}}[P_r,\id]=0$. Thus all entries in the first row and column, except the first one is 0. The first one is 1.

\paragraph{Diagonal : } Let $P_r=P_s$. From Equation \ref{eqn:abcd} we have,
\begin{eqnarray}
    a&=&\frac{9}{16\cdot 2^n}\tr[P_r^2]=\frac{9}{16}    \nonumber \\
    b+c&=&\frac{3}{16\cdot 2^n}[\tr(P_r^2Q)-\tr(QP_r^2)]=0  \nonumber \\
    d&=&\frac{1}{16\cdot 2^n}\tr[(P_rQ)^2]  \nonumber \\
    &=&\frac{1}{16\cdot 2^n}[\tr[(P_rP_1)^2]+\tr[(P_rP_2)^2]+\tr[(P_rP_3)^2]+\tr[(P_rP_1P_2)^2]+\tr[(P_rP_2P_3)^2]  \nonumber \\
    &&+\tr[(P_rP_3P_1)^2]+\tr[(P_rP_1P_2P_3)^2] ]   \nonumber
\end{eqnarray}

\textbf{Case i : } If $P_r$ commmutes with each of $P_1, P_2, P_3$, then $d=\frac{7}{16}$ and $a+b+c+d=1$.

\textbf{Case ii : } $P_r$ anti-commutes with one of the Paulis, let $[P_r,P_1]\neq 0$. Then
\begin{eqnarray}
    d&=&\frac{1}{16}(-1+1+1-1+1-1-1)=-\frac{1}{16}  \nonumber 
\end{eqnarray}

\textbf{Case iii : } $P_r$ anti-commutes with two of the Paulis, let $[P_r,P_1],[P_r,P_2]\neq 0$. Then
\begin{eqnarray}
    d&=&\frac{1}{16}(-1-1+1+1-1-1+1)=-\frac{1}{16}  \nonumber 
\end{eqnarray}

\textbf{Case iv : } $P_r$ anti-commutes with all the Paulis, let $[P_r,P_1],[P_r,P_2],[P_r,P_3]\neq 0$. Then
\begin{eqnarray}
    d&=&\frac{1}{16}(-1-1-1+1+1+1-1)=-\frac{1}{16}  \nonumber 
\end{eqnarray}
So, in cases ii, iii and iv we have $a+b+c+d=\frac{1}{2}$. Therefore, diagonal elements are in $\{1,\frac{1}{2}\}$. 

\paragraph{Off-diagonal elements : } In this case $P_r\neq P_s$. So from Equation \ref{eqn:abcd} we have
\begin{eqnarray}
    a&=&\frac{9}{16\cdot 2^n}\tr[P_rP_s]=0,  \nonumber \\
    b+c&=&\frac{3}{16\cdot 2^n}[\tr(P_rP_sQ)+\tr(P_sP_rQ)],  \nonumber \\
    d&=&\frac{1}{16\cdot 2^n}\tr[P_rQP_sQ],  \nonumber \\
    &=&\frac{1}{16\cdot 2^n}[\tr(P_rP_1P_sQ)+\tr(P_rP_2P_sQ)+\tr(P_rP_3P_sQ)-\tr(P_rP_1P_2P_sQ) \nonumber \\
    &&-\tr(P_rP_2P_3P_sQ)-\tr(P_rP_3P_1P_sQ)+\tr(P_rP_1P_2P_3P_sQ). \label{eqn:offDiagCaseI}
\end{eqnarray}

\textbf{Case I : $\mathbf{[P_r,P_s]\neq 0}$ : } In this case $b+c=0$ from the above equation. 

\textbf{Case Ia : $\mathbf{P_r=\pm iP_sP_1}$ : } Then from the above equation we have
\begin{eqnarray}
    d&=&\pm\frac{i}{16\cdot 2^n}[\tr(Q)+\tr(P_sP_1P_2P_sQ)+\tr(P_sP_1P_3P_sQ)-\tr(P_sP_2P_sQ)   \nonumber \\
    &&-\tr(P_sP_1P_2P_3P_sQ)-\tr(P_sP_3P_sQ)+\tr(P_sP_2P_3P_sQ) ]   \nonumber \\
    &=&\pm\frac{i}{16\cdot 2^n}[-\tr[(P_sP_1P_2)^2]-\tr[(P_sP_1P_3)^2]-\tr[(P_sP_2)^2]   \nonumber \\
    &&-\tr[(P_sP_1P_2P_3)^2]-\tr[(P_sP_3)^2]-\tr[(P_sP_2P_3)^2] ].   \nonumber
\end{eqnarray}
In this case $[P_s,P_1]\neq 0$. If $[P_s,P_2]=[P_s,P_3]=0$ then $d=\pm\frac{i}{16\cdot 2^n}(+1+1-1+1-1-1)=0$. If any one of them anti-commutes, say $[P_s,P_2]\neq 0$ but $[P_s,P_3]=0$. Then $d=\pm\frac{i}{16\cdot 2^n}(-1+1+1-1-1+1)=0$. If $[P_s,P_2],[P_s,P_3]\neq 0$ then $d=\pm\frac{i}{16\cdot 2^n}(-1-1+1+1+1-1)=0$. 

\textbf{Case Ib : $\mathbf{P_r=\pm iP_sP_2}$ } and \textbf{Case Ic : $\mathbf{P_r=\pm iP_sP_3}$ :} Similarly it can be shown that $d=0$. 

\textbf{Case Id : $\mathbf{P_r=\pm iP_sP_1P_2}$ : } From Equation \ref{eqn:offDiagCaseI} we have
\begin{eqnarray}
    d&=&\pm\frac{i}{16\cdot 2^n}\left[\tr(P_sP_2P_sQ)+\tr(P_sP_1P_sQ)+\tr(P_sP_1P_2P_3P_sQ)-\tr(Q)  \right. \nonumber \\
    &&\left. -\tr(P_sP_1P_3P_sQ)-\tr(P_sP_2P_3P_sQ)+\tr(P_sP_3P_sQ) \right] \nonumber \\
    &=&\pm\frac{i}{16\cdot 2^n}\left[\tr[(P_sP_2)^2]+\tr[(P_sP_1)^2]+\tr[(P_sP_1P_2P_3)^2]+\tr[(P_sP_1P_3)^2] \right.   \nonumber \\
    &&\left.+\tr[(P_sP_2P_3)^2]+\tr[(P_sP_3)^2]  \right].   \nonumber
\end{eqnarray}
Here $P_s$ anti-commutes with either one of $P_1,P_2$. Let $[P_s,P_1]\neq 0$ and $[P_s,P_2]=0$. First, we assume that $[P_s,P_3]=0$. Then $d=\pm\frac{i}{16}(1-1-1-1+1+1)=0$. Next, we assume that $[P_s,P_3]\neq 0$. Then $d=\pm\frac{i}{16}(1-1+1+1-1-1)=0$.

\textbf{Case Ie : $\mathbf{P_r=\pm iP_sP_2P_3}$ } and \textbf{Case If : $\mathbf{P_r=\pm iP_sP_3P_1}$ : } With similar arguments we can show that $d=0$. 

\textbf{Case Ig : $\mathbf{P_r=\pm iP_sP_1P_2P_3}$ : } Here from Equation \ref{eqn:offDiagCaseI} we have 
\begin{eqnarray}
    d&=&\pm\frac{i}{16\cdot 2^n}\left[\tr(P_sP_2P_3P_sQ)+\tr(P_sP_1P_3P_sQ)+\tr(P_sP_1P_2P_sQ)-\tr(P_sP_3P_sQ)  \right. \nonumber \\
      &&\left.-\tr(P_sP_1P_sQ)-\tr(P_sP_2P_sQ)+\tr(Q)  \right]  \nonumber \\
      &=&\pm\frac{i}{16\cdot 2^n}\left[-\tr[(P_sP_2P_3)^2]-\tr[(P_sP_1P_3)^2]-\tr[(P_sP_1P_2)^2]-\tr[(P_sP_3)^2]  \right. \nonumber \\
      &&\left.-\tr[(P_sP_1)^2]-\tr[(P_sP_2)^2]  \right].  \nonumber
\end{eqnarray}
In this case $P_s$ anti-commutes with either all three of $P_1,P_2,P_3$ or any of these. First, let us assume that $[P_s,P_1],[P_s,P_2],[P_s,P_3]\neq 0$. Then $d=\pm\frac{i}{16}(-1-1-1+1+1+1)=0$. Next, without loss of generality assume that $[P_s,P_1]\neq 0$ but $[P_s,P_1]=[P_s,P_2]=0$. Then $d=\pm\frac{i}{16}(-1+1+1-1+1-1)=0$.

Hence for every condition in Case I we have $a+b+c+d = 0$.

\textbf{Case II : $\mathbf{[P_r,P_s]=0 }$ : } In this case
\begin{eqnarray}
    b+c&=&\frac{6}{16\cdot 2^n}\tr[P_rP_sQ]. \label{eqn:offDiagCaseII}
\end{eqnarray}

\textbf{Case IIa : $\mathbf{P_r=\pm P_sP_1}$ :} From Equations \ref{eqn:offDiagCaseI} and \ref{eqn:offDiagCaseII} we have
\begin{eqnarray}
    b+c&=&\pm\frac{6}{16\cdot 2^n}\tr[P_1^2]=\pm\frac{6}{16},    \nonumber \\
    d&=&\pm\frac{1}{16\cdot 2^n}[-\tr[(P_sP_1P_2)^2]-\tr[(P_sP_1P_3)^2]-\tr[(P_sP_2)^2]   \nonumber \\
    &&-\tr[(P_sP_1P_2P_3)^2]-\tr[(P_sP_3)^2]-\tr[(P_sP_2P_3)^2] ].   \nonumber
\end{eqnarray}
In this case $[P_s,P_1]=0$. If both $[P_s,P_2]=[P_s,P_3]=0$ then $d=\mp\frac{6}{16}$ and $a+b+c+d=0$. If any one anti-commutes i.e. $[P_s,P_2]\neq 0$ then $d=\pm\frac{1}{16}(+1-1+1+1-1+1)=\pm\frac{2}{16}$ and $a+b+c+d=\pm\frac{1}{2}$. If both $[P_s,P_3],[P_s,P_2]\neq 0$ then $d=\pm\frac{1}{16}(+1+1+1-1+1-1)=\pm\frac{2}{16}$ and $a+b+c+d=\pm\frac{1}{2}$.

\textbf{Case IIb : $\mathbf{P_r=\pm P_sP_2}$ } and  \textbf{Case IIc : $\mathbf{P_r=\pm P_sP_3}$ :} Similar observations can be made.

\textbf{Case IId : $\mathbf{P_r=\pm P_sP_1P_2}$ :} Then from Equations \ref{eqn:offDiagCaseI} and \ref{eqn:offDiagCaseII},
\begin{eqnarray}
    b+c&=&\mp\frac{6}{16\cdot 2^n}\tr[(P_1P_2)^2]=\mp\frac{6}{16}   \nonumber \\
    d&=&\pm\frac{1}{16\cdot 2^n}[\tr[(P_sP_2)^2]+\tr[(P_sP_1)^2]+\tr[(P_sP_1P_2P_3)^2]-\tr[Q]+\tr[(P_sP_3P_1)^2]    \nonumber \\
    &&+\tr[(P_sP_2P_3)^2]+\tr[(P_SP_3)^2]  ].    \nonumber
\end{eqnarray}
In this case either both $[P_s,P_1],[P_s,P_2]\neq 0$ or $[P_s,P_1]=[P_s,P_2]=0$. If both are zero consider the case when $[P_s,P_3]=0$. Then $d=\pm\frac{1}{16}(1+1+1+1+1+1)=\pm\frac{6}{16}$ and hence $a+b+c+d=0$. Next consider $[P_s,P_3]\neq 0$. Then $d=\pm\frac{6}{16}(1+1-1-1-1-1)=\mp\frac{6}{16}$ and hence $a+b+c+d=\mp\frac{1}{2}$.

Now consider if both $[P_s,P_1],[P_s,P_2]\neq 0$. Then if $[P_s,P_3]=0$ then $d=\pm\frac{1}{16}(-1-1+1-1-1+1)=\mp\frac{2}{16}$ and so $a+b+c+d=\mp\frac{1}{2}$. If $[P_s,P_3]\neq 0$ then $d=\pm\frac{1}{16}(-1-1-1+1+1-1)=\mp\frac{2}{16}$ and hence $a+b+c+d=\mp\frac{1}{2}$. 

\textbf{Case IIe : $\mathbf{P_r=\pm P_sP_2P_3}$ } and \textbf{Case IIf : $\mathbf{P_r=\pm P_sP_3P_1}$ :} Similar conclusions can be drawn.

\textbf{Case IIg : $\mathbf{P_r=\pm P_sP_1P_2P_3}$ :} Then from Equations \ref{eqn:offDiagCaseI} and \ref{eqn:offDiagCaseII},
\begin{eqnarray}
    b+c&=&\pm\frac{6}{16\cdot 2^n}\tr[(P_1P_2P_3)^2]=\pm\frac{6}{16}   \nonumber \\
    d&=&\pm\frac{1}{16\cdot 2^n}[-\tr[(P_sP_1P_2)^2]-\tr[(P_sP_2P_3)^2]-\tr[(P_sP_1P_3)^2]-\tr[(P_sP_3)^2]-\tr[(P_sP_1)^2]    \nonumber \\
    &&-\tr[(P_sP_2)^2]+\tr[Q]  ].    \nonumber
\end{eqnarray}
In this case either $[P_s,P_1P_2P_3]=0$, implying either all commutes with $P_s$ or any two anti-commutes with $P_s$. In the former case, $d=\pm\frac{1}{16}(-1-1-1-1-1-1)=\mp\frac{6}{16}$ and so $a+b+c+d=0$. For the latter let $[P_s,P_1],[P_s,P_2]\neq 0$ and $[P_s,P_3]=0$. Then $d=\pm\frac{1}{16}(-1+1+1-1+1+1)=\pm\frac{2}{16}$ and so $a+b+c+d=\pm\frac{1}{2}$. 

To summarize we have the following inferences.
\begin{enumerate}
    \item $\chan{G_{P_1,P_2,P_3}}[P_r,P_r]=1$ when $[P_r,P_1]=[P_r,P_2]=[P_r,P_3]=0$, else it is $\frac{1}{2}$. This proves point (1) of the theorem. Now $P_1$ commutes with half of $\pauli_n$ i.e. $|\{P_r:[P_1,P_r]=0\}|=\frac{4^n}{2}$. Out of these $P_2$ commutes with half of them i.e. $|\{P_r:[P_1,P_r]=[P_2,P_r]=0\}|=\frac{4^n}{4}$. Out of these, $P_3$ commutes with half of them i.e. $|\{P_r:[P_1,P_r]=[P_2,P_r]=[P_3,P_r]=0\}|=\frac{4^n}{8}=2^{2n-3}$. This proves point (4) of the theorem.

    \item The constraints imply that an off-diagonal element $\chan{G_{P_1,P_2,P_3}}[P_r,P_s]=0$ when $[P_r,P_1]=[P_r,P_2]=[P_r,P_3]=0$. This proves point (2).

    \item The above arguments also imply that an off-diagonal element $\chan{G_{P_1,P_2,P_3}}[P_r,P_s]=\pm\frac{1}{2}$ whenever $P_r$ anti-commutes with at least one of $P_1, P_2$ or $P_3$. The sign depends on the commutativity relations between $P_s$ and $P_1,P_2,P_3$. This proves point (3) of the theorem.
    
\end{enumerate}
 
\end{proof}

 %--------------------------------------------------------------------------------
\section{Pseudocode}
\label{app:pseudocode}

In this section we explain briefly the main pseudocodes for the implementation of the algorithms described in Sections \ref{sec:results} and \ref{sec:method}. We have not included the pseudocode for some small sub-routines. We have simply described them in relevant places. We denote $N=2^n$, where $n$ is the number of qubits. 

\paragraph{I. Generating set $\mathbf{\gen_{Tof}}$ : } In Section \ref{subsec:genTOF} we have introduced the generating set $\gen_{Tof}$ and described its elements that are unitaries represented by a triple of distinct Paulis, satisfying some properties given by Lemma \ref{lem:GtofProperties}. In Algorithm \ref{alg:genTOF} (\textbf{GEN-TOF}) we have given the pseudocode for constructing $\gen_{Tof}$. The output are the triples of Paulis that represent the unitaries. We have eliminated triples using the equalities listed in Lemma \ref{lem:GtofProperties} and some of their implications. For example, $G_{P_1,-P_1P_2,-P_1P_3}=G_{P_1,-P_1P_2,-P_2P_3}=G_{P_1,-P_1P_2,P_1P_2P_3}$. Also we know that if $[P_3,P_1P_2]=0$ then $P_3$ either commutes or anti-commutes with both $P_1$ and $P_2$. Since each triple in $\gen_{Tof}$ consists of mutually commuting Paulis, so if we consider all triples such as $(P_1,P_2,.)$, we need not consider triples such as $(P_1,P_1P_2,.)$. So the check in step \ref{genTOF:checkP1P2} takes care of this redundancy as well as the equality in points (2) and (3) of Lemma \ref{lem:GtofProperties}. We use the sub-routine PAULI-COMM in order to test if two $n$-qubit Paulis commute. This can be done by checking if there are even number of qubit-wise anti-commutations.

\paragraph{II. Channel representation : } In Algorithm \ref{alg:chanTOF} (\textbf{CHAN-TOF}) we have given the pseudocode for computing the channel representation of each element in $\gen_{Tof}$. $\chan{G_{P_1,P_2,P_3}}$ is stored as an array $A_{P_1,P_2,P_3}$ of size $7\cdot 2^{2n-3}$, as discussed in Section \ref{subsec:chanRep}. Each element of this array is of the form $[i,\pm j_1,\pm j_2, \pm j_3]$, implying $\chan{G_{P_1,P_2,P_3}}[i,i]=\frac{1}{2}$, $\chan{G_{P_1,P_2,P_3}}[i,j_1]=\pm\frac{1}{2}$, $\chan{G_{P_1,P_2,P_3}}[i,j_2]=\pm\frac{1}{2}$ and $\chan{G_{P_1,P_2,P_3}}[i,j_3]=\pm\frac{1}{2}$. 

Alternatively, we can compute the channel representation using only the commutation relations, as shown in the proof of Theorem \ref{app:thm:chanRep}. For example, we know that $\chan{G_{P_1,P_2,P_3}}[P_r,P_r]=1$ when $[P_r,P_1]=[P_r,P_2]=[P_r,P_3]=0$, else it is $\frac{1}{2}$. For the latter case the position of the off-diagonal non-zero entries will be at $P_s$, where $P_s = \pm (i)^x P_r P' $, $x\in\{0,1\}$ and $P'\in\{P_1,P_2,P_3,P_1P_2,P_2P_3,P_3P_1,P_1P_2P_3\}$. The sign is determined by the multiplication, as described explicitly in the proof of Theorem \ref{app:thm:chanRep}. This procedure is more efficient that CHAN-TOF.

\paragraph{III. Multiplication by $\mathbf{\chan{G_{P_1,P_2,P_3}}}$ : } In \textbf{MULT-}$\mathbf{\mathcal{G}_{Tof}}$ (Algorithm \ref{alg:multGtof}) we give the pseudocode for efficiently multiplying a unitary with $\chan{G_{P_1,P_2,P_3}}$, as explained in Section \ref{subsec:mult}. We keep in mind that this algorithm is used while working with exactly implementable unitaries, whose channel representation are matrices with elements in $\intg\left[\frac{1}{2}\right]$ and hence are represented as $[a,k]$, implying $\frac{a}{2^k}$.
It calls the sub-routine \textbf{ADD-2} (Algorithm \ref{alg:add2}), which adds two elements in $\intg\left[\frac{1}{2}\right]$, where each element $v=\frac{a}{2^k}$ is represented as $[a,k]$. This in turn calls Algorithm \ref{alg:sde2Red} (\textbf{sde}$\mathbf{_2}$\textbf{-REDUCE}), which reduces a fraction $v=\frac{a}{2^k}$ to $\frac{a'}{2^{k'}}$ such that $k'=\sde_2(v)$.

\paragraph{IV. Approximately implementable unitaries : } The pseudocodes for the implementation of the algorithm to find optimal Toffoli-count of approximately implementable unitaries, as outlined in Section \ref{subsec:algoApprox}, are as follows.

The optimization version, \textbf{APPROX-TOF-OPT} (Algorithm \ref{alg:min}),takes as input a unitary $W\in\mathcal{U}_n$ and precision or error parameter $\epsilon > 0$. It iteratively calls the decision version, \textbf{APPROX-TOF-DECIDE} (Algorithm \ref{alg:decide}), that has an additional parameter $m \in\intg_{>0}$ and outputs YES if there exists an exactly implementable unitary $U$ within distance $\epsilon$ of $W$ and with Toffoli-count at most $m$. It loops over all products of $m$ unitaries from $\gen_{Tof}$ and in each iteration it calculates a set of amplitudes (step \ref{decide:Sc}) in order to perform the amplitude test (step 8), as described in Section \ref{subsec:algoApprox}. If it passes this test then it performs a conjugation test $\mathbf{\mathcal{A}}_{\mathbf{CONJ}}$ (Algorithm \ref{alg:conj}). It decides if an input unitary $W'$ is ``close to a Clifford'' i.e. $W'=EC_0$, where $C_0\in\cliff_n, E\in\mathcal{U}_n$ and $d(E,\id)\leq\epsilon$. APPROX-TOF-DECIDE returns a YES if it passes both the tests.

\paragraph{V. Nested meet-in-the-middle : } We have given the pseudocode for the nested meet-in-the-middle (MITM) algorithm described in Section \ref{subsubsec:nestMITM} in Algorithm \ref{app:alg:nestMITM} (\textbf{Nested MITM}).

\paragraph{VI. Exactly implementable unitaries : } The pseudocodes for the implementation of the heuristic algorithm to find a Toffoli-count-optimal decomposition of exactly implementable unitaries, as explained in Section \ref{subsubsec:heuristic}, are as follows.

The optimization version, \textbf{EXACT-TOF-OPT} (Algorithm \ref{alg:TOFcountOpt}) iteratively calls the decision version, \textbf{EXACT-TOF-DECIDE} (Algorithm \ref{alg:TOFcountDecide}). Since the $\sde_2$ of a unitary can change by at most 1 after multiplication by $\chan{G_{P_1,P_2,P_3}}$ (Lemma \ref{lem:sdeChangeMat}), so we start testing from $\sde_2(\chan{U})$, where $U$ is the input unitary. 

Algorithm \ref{alg:TOFcountDecide} tests if the Toffoli-count of an input unitary is at most a certain integer $m$, and if so it returns a decomposition. The procedure is described in Section \ref{subsubsec:heuristic}, so we give a brief explanation here. In short, it builds a pruned tree, where the input unitary is the root and a Clifford is a leaf. Each edge is multiplication by a generating set element (step \ref{exactDecide:multGv}). We use a $3\times 3$ integer matrix $SH$, whose rows index sde increase, unchanged and decrease. The columns index hamming weight increase, unchanged and decrease. This matrix is used in order to divide the children nodes according to their change in sde and hamming weight with respect to the parent node (step \ref{exactDecide:updateSH}). Then we select a subset of the children nodes which either belong to the minimum cardinality set from the SH matrix (steps \ref{exactDecide:minSH}, \ref{exactDecide:par2}) or have sde 1 (step \ref{exactDecide:par1}). If sde of any node is 0 (implying a Clifford or leaf) then we return the decomposition (step \ref{exactDecide:reachCliff}).  

We use a number of sub-routines. \textbf{HAM-WT-MAT} finds the Hamming weight of an input unitary. \textbf{UPDATE-SH} updates the SH matrix according to the Divide-and-Select rule used. \textbf{MIN-SH} returns the row and column index of the minimum non-zero entry of the SH matrix. 

\paragraph{VII. Random channel representation : } In Algorithm \ref{alg:randChanRep} (\textbf{RANDOM-CHAN-REP}) we generate the channel representation of a random unitary whose Toffoli-count is at most some input integer $tof_{in}$. First $tof_{in}$ number of unitaries are randomly selected from $\chan{\gen_{Tof}}$ and multiplied. Then we randomly permute the columns of the resultant unitary. After that we multiply each column by -1 with probability 1/2. The last two steps reflect multiplication by a random Clifford.

%------------------------------------------------------------
\begin{algorithm}
\scriptsize
 \caption{GEN-TOF}
 \label{alg:genTOF}

  \KwIn{$n$=number of qubits.}
 
 \KwOut{ Generating set $\gen_{Tof}$ as set of triples of Paulis that represent the unitaries in this set.}

$\gen_{Tof} = []$ \;

\For{each $P_1\in\pauli_n\setminus\{\id\}$}
{
    \For{each $P_2\in\pauli_n\setminus\{\id\}$ and $P_2>P_1$}
    {
        \For{each $P_3\in\pauli_n\setminus\{\id\}$ and $P_3>P_2>P_1$}
        {
            \If{PAULI-COMM($P_1,P_2$)==PAULI-COMM($P_2,P_3$)==PAULI-COMM($P_3,P_1$)==YES}
            {
                \If{$P_1P_2=\pm P_3$}
                {
                    \textbf{continue}   \;
                }
                $flag = 1$  \;
                \For{each $(P_a,P_b,P_c)\in \gen_{Tof}$}
                {
                    $prod_{12}=P_aP_b$; $prod_{23}=P_bP_c$; $prod_{31}=P_cP_a$; $prod_{123}=prod_{12}P_c$; $match = 0$ \;
                    $R=[P_a,P_b,P_c]$;    $R'=[P_1,P_2,P_3]$; $Q=[\id,\id,\id]$; $Q'=[\id,\id,\id]$; $indx_1 = [0,1,2]$; $indx_2 = [0,1,2]$ \;
                    \For{$j=0,1,2$}
                    {
                        \For{$k=0,1,2$}
                        {
                            \If{$R[j]==R'[k]$}
                            {
                                $Q[j] = R'[k]$; $Q'[j] = R'[k]$ \; 
                                $match = 1$; $indx_1 = indx_1\setminus\{j\}$; $indx_2 = indx_2\setminus\{k\}$  \;
                            }
                        }
                    }
                    \If{$match==1$}
                    {
                        $\ell =$ size.($indx_1$);     \;
                        \If{$\ell==2$}
                        {
                            $Q[indx_1[0]] = R'[indx_2[0]]$; $Q[indx_1[1]] = R'[indx_2[1]]$  \;
                            $Q'[indx_1[1]] = R'[indx_2[0]]$; $Q'[indx_1[0]] = R'[indx_2[1]]$  \;
                        }
                        \If{$\ell==1$}
                        {
                            $Q[indx_1[0]] = R'[indx_2[0]]$  \;                       
                        }
                        \If{($prod_{12}==\pm Q[0]$ or $\pm Q[1]$) or ($prod_{23}==\pm Q[1]$ or $\pm Q[2]$) or ($prod_{31}==\pm Q[2]$ or $\pm Q[0]$) \label{genTOF:checkP1P2} }
                        {
                            $flag = 0$  \;
                            \textbf{break}  \;
                        }
                        \If{($\ell==1$) and $prod_{123}==\pm Q[0]$ or $\pm Q[1]$ or $\pm Q[2]$}
                        {
                            $flag = 0$  \;
                            \textbf{break}  \;
                        }
                        \If{$\ell==2$}
                        {
                            \If{($prod_{12}==\pm Q'[0]$ or $\pm Q'[1]$) or ($prod_{23}==\pm Q'[1]$ or $\pm Q'[2]$) or ($prod_{31}==\pm Q'[2]$ or $\pm Q'[0]$) }
                            {
                                $flag = 0$  \;
                                \textbf{break}  \;
                            }
                        }    
                    }
                }
                \If{$flag==1$}
                {
                     $\gen_{Tof}$.append($P_1,P_2,P_3$) \;
                }
            }
        }
    }
}

\textbf{return} $\gen_{Tof}$ \;

\end{algorithm} 

\begin{comment}
\begin{algorithm}
\scriptsize
 \caption{PAULI-COMM}
 \label{alg:pauliComm}

  \KwIn{(i) $P=\bigotimes_{j=1}^nP_j \in \pauli_n, P_j\in\pauli_1$, (ii) $Q=\bigotimes_{j=1}^nQ_j \in \pauli_n, Q_j\in\pauli_1$.}
 
 \KwOut{ YES if $[P,Q]=0$, else NO.}

$x = 0$    \tcp*{This variable is a counter that counts number of qubit-wise anti-commutation.} 

 \For{$j=1$ to $n$}
 {
    \If{$P[j]\neq \id$ and $Q[j]\neq\id$ and $P[j]\neq Q[j]$}
    {
        $x = x+1$  \;
    }
 }

 \eIf{$x$ is even}
 {
    \Return YES \;
 }
 {
    \Return NO \;
 }

\end{algorithm} 
\end{comment}
%-------------------------------------------------------------
 \begin{algorithm}
\scriptsize 
 \caption{CHAN-TOF}
 \label{alg:chanTOF}

 \KwIn{(i) $n$ = number of qubits; (ii) $\gen_{Tof}$}
 
 \KwOut{ $\chan{\gen_{Tof}} = \{A_{P_1,P_2,P_3} : P_1,P_2,P_3\in\pauli_n;\quad\text{array }A_{P_1,P_2,P_3}\text{ represents }\chan{G_{P_1,P_2,P_3}} \}$.}

 $\chan{\gen_{Tof}} = []$ \;
 
 \For{$(P_1,P_2,P_3) \in \gen_{Tof}$}
 {
    $A_{P_1,P_2,P_3} = []$; $\qquad Q = P_1+P_2+P_3-P_1P_2-P_2P_3-P_3P_1+P_1P_2P_3$  \;
    \For{$P_r\in\pauli_n$}
    {
        $tuple = []$  \;
        $val = \frac{1}{2^n}\left(\frac{9}{16}\tr(\id)+\frac{1}{16}\tr((P_rQ)^2)  \right)$   \;
        \If{$val == \frac{1}{2}$}
        {
            $ tuple$.append$(P_r)$; $\qquad num = 0$ \;
            \For{$P_s\in\pauli_n\setminus\{P_r\}$}
            {
                $val = \frac{1}{2^n}\left(\frac{3}{16}\tr(P_rP_sQ)+\frac{3}{16}\tr(P_sP_rQ)+\frac{1}{16}\tr(P_rQP_sQ)  \right)$   \;
                \If{$val == \pm \frac{1}{2}$}
                {
                    $tuple$.append$(\pm P_s)$; $\qquad num = num+1$ \;
                    \If{$num == 3$}
                    {
                        $A_{P_1,P_2,P_3}$.append$(tuple)$   \;
                        \textbf{break}  \;
                    }
                }
            }
        }
    }
    $\chan{\gen_{Tof}}$.append$(A_{P_1,P_2,P_3})$  \;
 }

\textbf{return} $\chan{\gen_{Tof}}$ \;

\end{algorithm}

%--------------------------------------------------------------------
 \begin{algorithm}
\scriptsize 
 \caption{$\sde_2$-REDUCE}
 \label{alg:sde2Red}

  \KwIn{$v=(a,k)\in\intg\left[\frac{1}{2}\right]$.}
 
 \KwOut{ $v=(a',k')$ such that $k'=\sde_2(v)$.}

 \While{$a\% 2== 0$ and $k!=0$}
 {
    $a\leftarrow a/2$   \;
    $k\leftarrow k-1$   \;
 }

 \textbf{return} $(a,k)$    \;

\end{algorithm} 

\begin{algorithm}
    \scriptsize  
    \caption{ADD-2}
    \label{alg:add2}

    \KwIn{$v_1=(a_1,k_1)$, $v_2=(a_2,k_2)\in\intg\left[\frac{1}{2}\right]$.}

    \KwOut{$v=(a,k)=v_1+v_2$}

    \eIf{$k_1\geq k_2$}
    {
        $num = a_1+a_2\cdot 2^{k_1-k_2} $    \;
        $den = k_1$ \;
    }
    {
        $num = a_1\cdot 2^{k_2-k_1}+a_2 $    \;
        $den = k_2$ \;
    }

    \textbf{return } $\sde_2$-REDUCE $((num,den))$  \;
    
\end{algorithm}

\begin{algorithm}
    \scriptsize 
    \caption{MULT-$\gen_{Tof}$}
    \label{alg:multGtof}

    \KwIn{(i) $\chan{G_{P_1,P_2,P_3}}$ as array $A_{P_1,P_2,P_3}$, (ii) $U$ - both of size $N^2\times N^2$.}

    \KwOut{(i) $U_p=\chan{G_{P_1,P_2,P_3}}U$.}

    $U_p\leftarrow U$   \;
    \For{$i=1,\ldots,\frac{7N^2}{8}$}
    {
        $diag = A_{P_1,P_2,P_3}[i][0]$; $\quad oDiag_1=A_{P_1,P_2,P_3}[i][1]$; $\quad oDiag_2=A_{P_1,P_2,P_3}[i][2]$; $\quad oDiag_3=A_{P_1,P_2,P_3}[i][3]$ \;
        \eIf{$oDiag_1 < 0$}    
        {
            $oDiagIndx_1 = -oDiag_1$    \;
        }
        {
            $oDiagIndx_1 = oDiag_1$ \;
        }
        \eIf{$oDiag_2 < 0$}    
        {
            $oDiagIndx_2 = -oDiag_2$    \;
        }
        {
            $oDiagIndx_2 = oDiag_2$ \;
        }
        \eIf{$oDiag_3 < 0$}    
        {
            $oDiagIndx_3 = -oDiag_3$    \;
        }
        {
            $oDiagIndx_3 = oDiag_3$ \;
        }
        \For{$j=1,\ldots,N^2$}
        {
            $v_1 = (U[diag,j][0],U[diag,j][1]+1)$ \;
            \eIf{$oDiag_1<0$}
            {
                $v_2 = (-U[oDiagIndx_1,j][0],U[oDiagIndx_1,j][1]+1)$    \;
            }
            {
                $v_2 = (U[oDiagIndx_1,j][0],U[oDiagIndx_1,j][1]+1)$    \;
            }
            \eIf{$oDiag_2<0$}
            {
                $v_3 = (-U[oDiagIndx_2,j][0],U[oDiagIndx_2,j][1]+1)$    \;
            }
            {
                $v_3 = (U[oDiagIndx_2,j][0],U[oDiagIndx_2,j][1]+1)$    \;
            }
            \eIf{$oDiag_3<0$}
            {
                $v_4 = (-U[oDiagIndx_3,j][0],U[oDiagIndx_3,j][1]+1)$    \;
            }
            {
                $v_4 = (U[oDiagIndx_3,j][0],U[oDiagIndx_3,j][1]+1)$    \;
            }
            $U_p[diag,j] = $ ADD-2(ADD-2$(v_1,v_2)$, ADD-2$(v_3,v_4)$)  \;
        }
        
    }
    \textbf{return} $U_p$  \;
    
\end{algorithm}

%----------------------------------------------------

\begin{algorithm}
\scriptsize
 \caption{APPROX-TOF-OPT}
 \label{alg:min}
 
 \KwIn{(i) $W\in \mathcal{U}_n$, (ii) $\epsilon\geq 0$.}
 
 \KwOut{ $\tofeps(W)$.}

  $m\leftarrow 1$, $decision\leftarrow\text{NO}$    \;
 \While{(1)}
 {
    $decision\leftarrow$ APPROX-TOF-DECIDE $(W,m,\epsilon)$ \;
    \eIf{$decision==\text{YES}$}
    {
        \textbf{return} $m$ \;
    }
    {
        $m\leftarrow m+1$   \;
    }
 }
\end{algorithm}

 \begin{algorithm}
\scriptsize 
 \caption{APPROX-TOF-DECIDE}
 \label{alg:decide}
 
 \KwIn{(i) $W\in \mathcal{U}_n$, (ii) integer $m > 0$, (iii) $\epsilon\geq 0$.}
 
 \KwOut{ $\text{YES}$ if $\exists U\in\clifft_n^{Tof}$ such that $d(U,W)\leq\epsilon$ and $\tofcount(U)\leq m$; else $\text{NO}$.}
 
 $\gen_{Tof}\leftarrow $ GEN-TOF$(n)$ \label{Adec:S}\;
 
 \For{every $\widetilde{U}=\prod_{j=m}^1G_{P_{1_j},P_{2_j},P_{3_j}}$ such that $G_{P_{1_j},P_{2_j},P_{3_j} }\in\gen_{Tof}$ and $P_{1j}\neq P_{1,j+1}$,$P_{2j}\neq P_{2,j+1}$,$P_{3j}\neq P_{3,j+1}$ }
 {
    $W'=W^{\dagger}\widetilde{U}$    \label{decide:W}\;
    $\mathcal{S}_c\leftarrow\{\left|\tr\left(W'P\right)/N\right|:P\in\pauli_n\}$ and sort this set in descending order \label{decide:Sc}\;
    \For{$M=1,2\ldots N^2$}
    {
        $\mathcal{S}_1\leftarrow $ First $M$ terms in $\mathcal{S}_c$  \label{decide:S1}\;
        $\mathcal{S}_0=\mathcal{S}_c\setminus\mathcal{S}_1$ \label{decide:S0}\;
        \If{each term in $\mathcal{S}_1 \in \left[\frac{1-\epsilon^2}{\sqrt{M}}-\sqrt{M(2\epsilon^2-\epsilon^4)},\frac{1}{\sqrt{M}}+\sqrt{M(2\epsilon^2-\epsilon^4)}\right]$ and each term in $\mathcal{S}_0\in\left[0,\sqrt{M(2\epsilon^2-\epsilon^4)}\right]$}
        {
            \If{YES $\leftarrow \mathcal{A}_{CONJ}(W',\epsilon)$}
            {
            \textbf{return} $\text{YES}$ \;
            }
        }
    }
 }
 \textbf{return} $\text{NO}$\;
\end{algorithm}

\begin{algorithm}
\scriptsize 
 \caption{$\mathcal{A}_{CONJ}$}
 \label{alg:conj}
 
 \KwIn{(i) $W'\in \mathcal{U}_n$, (ii) $\epsilon\geq 0$.}
 
 \KwOut{ $\text{YES}$ if $\exists C_0\in\cliff_n, E\in\mathcal{U}_n$ such that $W'=E^{\dagger}C_0$, where $d(E,\id)\leq\epsilon$ ; else $\text{NO}$.}
 
 $p\leftarrow 1$    \;
 
 \For{every $P_{out}\in \pauli_n$}
 {
    \If{$p==1$}
    {
        $p\leftarrow 0$ \;
    }

    \For{every $P_{in}\in\pauli_n$}
    {
        \If{$(1-4\epsilon^2+2\epsilon^4) \leq\left|\tr\left(W'P_{out}W'^{\dagger}P_{in}\right)\right|/N\leq 1$}
        {
            $p\leftarrow p+1$   \;
            \If{$p>1$}
             {
                \textbf{return} NO  \;
            }
        }
        \If{$2\epsilon <\left|\tr\left(W'P_{out}W'^{\dagger}P_{in}\right)\right|/N < (1-4\epsilon^2+2\epsilon^4) $}
        {
            \textbf{return} NO \;
        }
    }
 }
 
 \textbf{return} YES    \;
 \end{algorithm} 
%----------------------------------------------------------------------------------------------

\begin{algorithm}
\scriptsize
 \caption{Nested MITM}
 \label{app:alg:nestMITM}
 \KwIn{(i) A unitary $U\in\clifft_n^{Tof}$, (ii)  generating set $\gen_{Tof}$, (iii) test-count $m$, (iv) $c\geq 2$}
 \KwOut{A circuit (if it exists) for $U$ such that Toffoli-count is at most $m$. }
 
 $S_0 \leftarrow \{\id\}$;  $\quad i\leftarrow 1$     \;  \label{nestMITM:start}
 \While{$i\leq \left\lceil\frac{m}{c}\right\rceil$ \label{nestMITM:whileStart}}
 {
    $S_i\leftarrow \{ \left(\chan{G_{P_1,P_2,P_3}}\chan{W}\right)^{(co)} : G_{P_1,P_2,P_3}\in\gen_{Tof}, W\in S_{i-1}  \} $; $\qquad k=c-1$ \; \label{nestMITM:depthInc}
    
        \For{$\chan{W}=\chan{W_1}\chan{W_2}\ldots \chan{W_k}$ where $\chan{W_i}\in S_i$ or $\chan{W_i}\in S_{i-1}$ \label{nestMITM:for}}
        {
            \If{$\exists \chan{W'}\in S_i$ such that $\left(\chan{W}^{\dagger}\chan{U}\right)^{(co)}=\chan{W'}$ \label{nestMITM:check1}}
            {
                  \Return $\chan{W_1},\chan{W_2},\ldots,\chan{W_k},\chan{W'}$  \label{nestMITM:return1}\;  
                   break \; 
            }
            \ElseIf{$\exists \chan{W'}\in S_{i-1}$ such that $\left(\chan{W}^{\dagger}\chan{U}\right)^{(co)}=\chan{W'}$ \label{nestMITM:check2}}
            {
                  \Return $\chan{W_1},\chan{W_2},\ldots,\chan{W_k},\chan{W'}$  \label{nestMITM:return2}\;  
                   break \; 
            }
        }\label{nestMITM:forEnd}
   $i\leftarrow i+1$ \;
 }  \label{nestMITM:whileEnd}
 \If{no decomposition found}
 {
    \Return "$U$ has Toffoli-count more than $m$."  \;
 }
 
\end{algorithm}
    
%-------------------------------------------------------------------
\begin{algorithm}
    \scriptsize
    \caption{EXACT-TOF-OPT}
    \label{alg:TOFcountOpt}

    \KwIn{$\chan{U}$, where $U\in\clifft_n^{Tof}$}

    \KwOut{$m$, the minimum Toffoli-count of $U$ and its decomposition $\chan{U}=\left(\prod_{j=m}^1\chan{G_{P_{1_j},P_{2_j},P_{3_j} }} \right)\chan{C_0} $, where $G_{P_{1_j},P_{2_j},P_{3_j}}\in\gen_{Tof}$.}

 $m=\sde_{\chan{U}}$    \;
 \eIf{$m==0$}
 {
    \textbf{return } 0  \;
 }
 {
     \While{1}
    {
        $(m',\mathcal{D})\leftarrow$ EXACT-TOF-DECIDE$(\chan{U},m)$    \;
        \eIf{$m'==-1$}
        {
            $m\leftarrow m+1$   \;
        }
        {
            \textbf{return} $(m',\mathcal{D})$  \;
        }
    }
 }

\end{algorithm}

\begin{algorithm}
    \scriptsize
    \caption{EXACT-TOF-DECIDE}
    \label{alg:TOFcountDecide}

    \KwIn{(i)$\chan{U}$ where $U\in\clifft_n^{Tof}$; (ii) $m\in\intg_{>0}$; (iii) $\chan{\gen_{Tof}}$.}

    \KwOut{$(m',\mathcal{D})$, where $m'\in\intg$ is -1 if Toffoli-count is more than $m$ and $\mathcal{D}$ is a decomposition if Toffoli-count is $m'\geq 1$.}

     $ \widetilde{U} = []$; $\qquad parNode=[\chan{U},\Path_{\chan{U}},\sde_{\chan{U}},\ham_{\chan{U}}]$;  $\qquad parNode =\{\widetilde{U}\}$ \tcp*[h]{Root node} \;

    \For{$i=1,\ldots,m$}
    {
        $childNode = []$; $ \qquad SH=[0]_{3\times 3}$   \;
        \For{each $\widetilde{U}=[\chan{U},\Path_{\chan{U}},\sde_{\chan{U}},\ham_{\chan{U}}] \in parNode$}
        {
            $P_{prev} = \Path_{\chan{U}}[i-1]$; $\quad \sde_{par} = \sde_{\chan{U}}$; $\quad \ham_{par} = \ham_{\chan{U}}$   \;
            \For{each $\chan{G_{P_1,P_2,P_3}} \in \chan{\gen_{Tof}}$}
            {
                \If{$(P_1,P_2,P_3)\neq P_{prev}$}
                {
                    $W \leftarrow$ MULT-$\gen_{Tof} (\chan{G_{P_1,P_2,P_3}}, \chan{U})$; $\quad\ham_W\leftarrow$ HAM-WT-MAT$(W)$; $\quad\sde_W\leftarrow$ GET-SDE$(W)$ \label{exactDecide:multGv}  \;
                    \eIf{$\sde_W == 0$}
                    {
                        \textbf{return} $(i,\Path_{\chan{U}}\cup \{ (P_1,P_2,P_3) \} )  $ \tcp*[h]{Reached Clifford} \label{exactDecide:reachCliff} \;
                    }
                    {
                        $(SH,s,h)\leftarrow$UPDATE-SH $(SH,\sde_W,\ham_W,\sde_{par},\ham_{par},rule)$ \label{exactDecide:updateSH}   \;
                    }
                    $\Path_W = \Path_{\chan{U}}\cup \{ (P_1,P_2,P_3) \}$  \;
                    $childNode$.append$(W,s,h,\Path_W,\sde_W,\ham_W )$  \label{exactDecide:childAppend}  \;
                }
            }
        }
        $(s_{indx},h_{indx})\leftarrow$ MIN-SH$(SH)$ \label{exactDecide:minSH}     \;
        \For{each $\widetilde{W}=[W,s,h,\Path_W,\sde_W,\ham_W]\in childNode$}
        {
            \If{$\sde_W > m+1-i$}
            {
                \textbf{continue}   \;
            }
            \If{$\sde_W==1$}
            {
                $parNode$.append$([W,\Path_W,\sde_W,\ham_W])$ \label{exactDecide:par1} \;
            }
            \If{$s==s_{indx}$ and $h==h_{indx}$}
            {
                $parNode$.append$([W,\Path_W,\sde_W,\ham_W])$  \label{exactDecide:par2} \;
            }
        }
    }

    \textbf{return} $(-1,[])$   \;
    
\end{algorithm}

%------------------------------------------------------------
\begin{algorithm}
    \scriptsize
    \caption{RANDOM-CHAN-REP}
    \label{alg:randChanRep}

    \KwIn{$tof_{in}$ : Input Toffoli-count}

    \KwOut{$\chan{U}$ : Channel representation of a unitary with Toffoli-count at most $tof_{in}$}
    
    $i=0$; $\qquad\chan{U}=\chan{\id}$   \;
    \While{$i < tof_{in}$}
    {
        Randomly sample $(P_1,P_2,P_3)$ from $\gen_{Tof}$ \tcp*[h]{As output by GEN-TOF (Algorithm \ref{alg:genTOF} )}\;
        \eIf{$i==0$}
        {
            $P_{prev}=(P_1,P_2,P_3)$    \;
            $\chan{U}\leftarrow$ MULT-$\gen_{Tof}(\chan{G_{P_1,P_2,P_3}},\chan{U})$    \;
        }
        {
            \If{$(P_1,P_2,P_3)\neq P_{prev}$}
            {
                $P_{prev} = (P_1,P_2,P_3)$  \;
                $\chan{U}\leftarrow$ MULT-$\gen_{Tof}(\chan{G_{P_1,P_2,P_3}},\chan{U})$    \;
            }
        }
    }
    
    Randomly permute the columns of $\chan{U}$  \;
    Multiply each column of $\chan{U}$ with $-1$ with probability $\frac{1}{2}$   \;
    \textbf{return} $\chan{U}$  \;
    
\end{algorithm}

\end{document}